\documentclass[a4paper,reprint,twocolumn,notitlepage,groupedaddress,aip,nofootinbib]{revtex4-2}

\usepackage[left=1.5cm,right=1.5cm,top=1cm,bottom=1.5cm,includeheadfoot]{geometry}

\usepackage[english]{babel}
\usepackage[utf8x]{inputenc}
\usepackage{tgtermes} 
\usepackage{tgheros} 
\usepackage[T1]{fontenc}

\usepackage{amssymb}
\usepackage{amsmath}
\usepackage{amsthm}
\usepackage{mathtools}
\usepackage[dvipsnames]{xcolor}
\usepackage[shortlabels]{enumitem}
\usepackage[normalem]{ulem}
\PassOptionsToPackage{hyphens}{url} 
\usepackage[breaklinks=true]{hyperref}
\bibpunct{[}{]}{;}{n}{}{} 

\usepackage{graphicx}
\usepackage[newcommands]{ragged2e} 
\usepackage{subcaption}

\usepackage{tikz}
\usetikzlibrary{shapes.geometric}
\usetikzlibrary{calc}
\usetikzlibrary{matrix}
\usetikzlibrary{patterns}
\usetikzlibrary{arrows.meta}
\tikzset{
  graphnode/.style={draw,circle,fill=SkyBlue,draw=black}
}

\theoremstyle{plain}
\newtheorem{theorem}{Theorem}
\newtheorem{proposition}[theorem]{Proposition}
\newtheorem{lemma}[theorem]{Lemma}
\newtheorem{corollary}[theorem]{Corollary}

\theoremstyle{definition}
\newtheorem{definition}[theorem]{Definition}

\theoremstyle{remark}

\newcommand{\R}{\mathbb{R}}
\newcommand{\C}{\mathbb{C}}

\renewcommand{\d}{\mathrm{d}}
\newcommand{\e}{\mathrm{e}}
\renewcommand{\i}{\mathrm{i}}

\newcommand{\eps}{\varepsilon}
\renewcommand{\H}{\mathcal{H}}

\newcommand{\be}{\begin{equation}}
\newcommand{\ee}{\end{equation}}

\newcommand{\Tr}{\mathop\mathrm{Tr}}

\newcommand{\ch}{\mathrm{ch}}

\begin{document}
\title{
    Density-Functional Theory on Graphs
}

\author{Markus Penz}
\address{Department of Mathematics, University of Innsbruck, Austria}
\email{m.penz@inter.at}

\author{Robert van Leeuwen}
\address{Department of Physics, Nanoscience Center, University of Jyv\"askyl\"a, Finland}

\begin{abstract}
    The principles of density-functional theory are studied for finite lattice systems represented by graphs. Surprisingly, the fundamental Hohenberg--Kohn theorem is found void in general, while many insights into the topological structure of the density-potential mapping can be won. We give precise conditions for a ground state to be uniquely $v$-representable and are able to prove that this property holds for almost all densities. A set of examples illustrates the theory and demonstrates the non-convexity of the pure-state constrained-search functional.
\end{abstract}
\maketitle

\vspace{-2em}
{\small
\tableofcontents
}

\section{Introduction}
\label{sec:main:intro}

In the theoretical development of density-functional theory (DFT) following the pioneering paper of \citet{Hohenberg-Kohn1964}, the works of \citet{Levy79} and \citet{Lieb1983} marked cornerstones on which practically all the following investigations built. DFT has since been extended to many different settings, including lattice systems \cite{katriel1981mapping,englisch2,CCR1985,schonhammer1995density,xianlong2006,giesbertz2019-1RDM} that in practice mostly appear in the form of the Hubbard model \cite{carrascal2015hubbard,lima2003,ijas2010lattice,saubanere2014lattice}. 
The solidity of the theoretical investigations primarily established for continuum systems has led many to believe that the same statements automatically hold for fermionic lattice systems.
First and foremost, this concerns the Hohenberg--Kohn (HK) theorem \cite{Hohenberg-Kohn1964} that states the existence of a unique mapping from ground-state densities back to the external potentials included in the Hamiltonian and that has only recently found a rigorous underpinning for continuum systems \cite{Garrigue2018}.
This work is aimed at scrutinizing the foundations of DFT by studying finite, discrete systems, in their most general form represented by graphs. Such systems also find their special relevance in that they are currently the only ones that allow for an assuredly convergent formulation of the Kohn--Sham method \cite{penz2019guaranteed,penz2020erratum}. In the course of this work, we reveal some surprising and crucial differences to the continuum theory: While we find clear counterexamples to the HK theorem, a tremendous amount of insight into the structure of the (multi-valued) density-potential mapping can be gained. It turns out that the finite-dimensional setting is all but trivial and displays a pronounced mathematical richness.

Several misconceptions about the state of the HK theorem for lattice systems can be found in the research literature, usually due to unjustified claims about non-vanishing components of the ground-state wave function, simply assumed to hold ``for all practical purposes'', e.g., by \citet{coe2015uniqueness}. That such vanishing components in the wave function are indeed possible and do not have to be rare is demonstrated here by explicit examples.
Other works simply assume or state the validity of 
the lattice HK theorem \cite{dimitrov2016exact,xianlong2006}, where the system under study is a one-dimensional Hubbard chain. Incidentally, the linear chain is the only many-particle system for which we can actually prove a full HK theorem.

A central concept introduced and used in this work is that of \emph{unique $v$-representability} (uv) of ground-state densities, a notion that merges $v$-representability with an assumed validity of the HK theorem for the given density. While $v$-representability alone was already thoroughly studied on lattices by \citet{CCR1985} and received a positive answer for ensemble densities that are neither 0 nor 1 on any vertex, the question if such a density also comes from a unique potential (modulo constants) was left unanswered: ``The HK theorem for fermions at zero temperature remains an open problem.'' \cite{CCR1985} We will here present an adaption of their ensemble-$v$-representability proof for the graph setting and a partial answer to the uniqueness problem.


The main tool of our analysis is a matrix theorem by \citet{Odlyzko} that allows us to define classes of eigenstates that are uniquely $v$-representable. Now the question of uniqueness is transformed into determining which systems have ground states that fall into one of those classes.
While for one-particle graph systems and many particles on a linear chain uniqueness can be guaranteed on the basis of the Perron--Frobenius theorem from linear algebra, just an ``almost all'' statement was found to hold for more general systems. This means that the set of problematic, non-uv densities is small when compared to the set of all densities and that a random sampling should usually yield a uv density. Complementary to this result, Rellich's theorem about the analyticity of eigenvalues and eigenstates under perturbations is used to give an argument that almost all potentials lead to uv ground-state densities. Yet, no general statement about the validity of the HK theorem could be found for finite lattice systems.

If one moves beyond the requirement of ``uniqueness'', a great topological variety can still be analyzed. We will give a detailed characterization of the \emph{density-potential mapping}, the (multi-valued) map from the set of possible ground-state densities to the space of external potentials. Now, $v$-representability of all densities in the domain means that the mapping is ``well-defined'', i.e., a potential can be assigned to every density. If we even have \emph{unique} $v$-representability (uv) then there is only a unique way how to define the mapping (modulo constants). 
This then gives the density region where the Hohenberg--Kohn theorem is indeed valid.
Due to possible degeneracy of the ground state the mapping might still be many-to-one.

The outline of the paper is as follows: Section~\ref{sec:main:fermionic-HS} introduces the Hilbert space structure for fermions on graphs, a setting that is equivalent to finite lattice systems. Then in Section~\ref{sec:main:HK-violation} simple counterexamples show that the HK theorem will in general be violated for such systems, but we also give a detailed analysis of the situation based on Odlyzko's theorem and try to save as much as possible from the HK statement that still holds for \emph{almost all} densities. This analysis is then continued in Section~\ref{sec:main:Rellich} by invoking Rellich's theorem that allows us to shed some light on the topology of the density-potential mapping and includes the rectification of a result by \citet{KohnPRL} on the openness of the set of uv densities that come from non-degenerate states. Another powerful theorem, this time the Perron--Frobenius theorem from linear algebra employed in Section~\ref{sec:main:PF}, shows that non-interacting systems will always have purely positive ground-state densities and that the full HK theorem is valid for a linear chain. A large part on constrained-search functionals, Section~\ref{sec:main:CS-v-rep}, reflects on this other important ingredient of DFT that establishes an immensely valuable connection to convex analysis. Three examples in this section help to illustrate the results from the previous sections and also scrutinize and amend a proof by \citet{Lieb1983} about the non-convexity of the pure-state constrained-search functional: For the triangle graph we map out the complete density-potential mapping and derive the explicit form of the (convex) constrained-search functional, the complete graph demonstrates that Lieb's argument cannot hold in general, while the beautiful example of the cuboctahedron graph yields the desired counterexample to pure-state $v$-representability and indeed shows that the pure-state constrained-search functional is non-convex in general. Despite the large amount of results it seems these investigations into DFT on graphs merely opened up a Pandora's box and we collect a number of open problems and further ideas in the concluding Section~\ref{sec:main:conclusions}.

\section{Fermionic Hilbert space structures on graphs}
\label{sec:main:fermionic-HS}

\subsection{General definitions}

We start by defining the general setting of our approach.  
This consists of a quantum system defined on a finite-dimensional Hilbert space $\H=\mathbb{C}^L$ with the standard Hermitian inner product (linear in the second component like usually in physics literature). For now, we make no assumptions on the nature of the quantum system,
neither on the number or statistics of the particles involved, nor on the type of interactions or external potentials. We only assume for now that there are no time-dependent external fields present such that all properties of the system are fully described by the time-independent Schr\"odinger
equation
\be
H \Psi = E \Psi,
\ee
where $H: \H \to \H$ is a linear, self-adjoint operator representing the Hamiltonian and $\Psi \in \H$ is consequently the eigenvector to eigenvalue $E$. When we make a choice for an orthonormal basis the  
Hamiltonian can be represented by a Hermitian matrix with coefficients $H_{ij} = H_{ji}^* \in \mathbb{C}$.
For example, for a fermionic Hubbard system we can choose the basis functions to be the many-particle Slater determinants built out of 
any preferred set of orthonormal one-particle states.

At this point, we introduce a central concept of our approach, the finite graph $G(H)$ associated with a Hamiltonian $H$.
By graph we here always mean the following:

\begin{definition}
A graph $G$ consists of a vertex set  $X=\{ 1, \ldots,M \}$ and an irreflexive, symmetric adjacency relation $\sim$
defined on $X \times X$. 
\end{definition}

This means that for any pair of vertices $(i,j) \in X \times X$ we can say whether they are adjacent, written $i \sim j$, or non-adjacent, written $i \not\sim j$. Irreflexivity means that no vertex is adjacent to itself, i.e., $i \not\sim i$, and symmetry means that if $i \sim j$ then also $j \sim i$.
A graph can naturally be displayed graphically by drawing the elements of the vertex set on a plane and connecting two vertices $i$ and $j$ with a line whenever $i \sim j$.

\begin{definition}
A graph is called connected when for every two vertices $i$ and $j$ there exists a sequence of adjacent vertices starting with $i$ and terminating with $j$. Otherwise the graph is called disconnected.
\end{definition}

In other words, a graph is connected when there exists a path between any pair of vertices. 
In the next step we will associate a graph to a given Hermitian matrix, which in our case will always be the Hamiltonian matrix of our quantum system.

\begin{definition}\label{def:H-graph}
To any $M \times M$ Hermitian matrix $H$, i.e., the coefficients fulfil $H_{ij} = H_{ji}^*\in \mathbb{C}$, $i,j \in X=\{ 1, \ldots,M \}$, we associate a graph $G(H)$ by 
taking $X$ to be the vertex set and by defining the adjacency relation $i \sim j$ whenever $H_{ij} \neq 0$ for $i \neq j$, and $i \not \sim j$ otherwise.
\end{definition}

If $G(H)$ is disconnected then the Hamilton matrix blocks into submatrices corresponding to connected subgraphs. 
Some results in this work require the connectedness of the graph and we will explicitly indicate when this is the case. Note that by Definition~\ref{def:H-graph} a Hubbard model with more than nearest neighbor hopping is represented by a graph having additional edges. In other words, hopping is only admitted along edges of the graph.

We use the following convention for enumerating indices throughout this work: $i,j \in X = \{1,\ldots,M\}$ for counting one-particle states that correspond to vertices in the original graph $G(h)$, $k,l \in \{1,\ldots,N\}$ for counting particles, and $I,J$ are multi-indices like introduced in the next section. A sum over any of those indices will always mean summation over the full index set if not otherwise stated.

\subsection{Fermionic Hilbert space}

So far, our discussion was very general, as we made no assumptions on the nature of the system. From now on our goal is to describe a system of $N$ fermions.
We start by considering the Hamiltonian $h$ of a single particle defined on a one-particle Hilbert space $\H_1 = \C^M$ of dimension $M$. 
To this system corresponds a Hamiltonian graph $G(h)$ given by Definition~\ref{def:H-graph} in which the vertices $i$ label the one-particle states.
In particular, we can consider the $M$ orthonormal vectors
\be
e_i = (0, \ldots,1, \ldots,0) \in \C^M = \H_1,
\ee
where the vector contains a $1$ on position $i$ and zeroes elsewhere. These vectors physically describe quantum states in which the particle is located at vertex $i$ with certainty.
The physical nature of the states $e_i$ themselves depends on the initial basis choice that defines our Hamiltonian matrix and will be generally left completely open. However, for the purpose of illustration, we can always take the vertices as representing positions in space, in which case the vector $e_i$ is a quantum state that has the particle at position $i$ with certainty.
The graph $G(h)$ of the Hamiltonian then has a spatial representation in which a particle can hop from position $i$ to $j$ if $i \sim j$ on the graph.

In the next step, we put $N$ spinless fermions ($N\leq M$) on the graph $G(h)$ and allow for interactions between them.
The Hilbert space of the many-particle system is then given by the anti-symmetric $N$-fold tensor product $\H_N = \Lambda^N \H_1$, which is the linear span of the $N$-fold wedge products
\be
e_I = e_{i_1} \wedge \ldots \wedge e_{i_N}
\label{e_I}
\ee
that furnish an orthonormal basis of $\H_N$. Here we use the ordered multi-index $I = (i_1, \ldots,i_N) \in X^N$ with $i_1 < \ldots < i_N$ to label the many-particle basis states of $\H_N$.
There are $L= \binom{ M }{ N }$ basis states of the form $\eqref{e_I}$ and the fermionic Hilbert space $\H_N$ is therefore $L$-dimensional.
Since in the case $N=M$ the resulting fermionic Hilbert space is one-dimensional and thus trivial, we will generally assume $N<M$ in what follows.
Physically, the $e_I$ are quantum states in which there is a particle with certainty at vertices $i\in I$ and zero particles at all the other vertices.
For vectors $v_k\in\H_1$
the wedge product itself is defined as the anti-symmetrized tensor product
\be
v_1 \wedge \ldots \wedge v_N = \frac{1}{\sqrt{N!}}\sum_{\sigma} (-1)^{|\sigma|} 
v_{\sigma(1)} \otimes \ldots \otimes v_{\sigma(N)},
\ee
where the sum is over all permutations of $N$ symbols and $|\sigma|$ is the sign of the permutation $\sigma$.
The inner product in $\H_N$ is given by
\be
\langle v_1 \wedge \ldots \wedge v_N, w_1 \wedge \ldots \wedge w_N \rangle 
= \det (\langle v_k, w_l \rangle )_{kl},
\ee
from which it follows that the basis vectors $e_I$ of $\H_N$ are orthogonal and normalized to 1 again.
Since in quantum mechanics a probability interpretation is assigned only to such normalized states, we define the unit sphere in $\H_N$
\begin{equation}
    \mathcal{I}_N = \{ \Psi \in \H_N \mid \|\Psi\|=1 \}
\end{equation}
as the basic set of possible configurations.
The corresponding set of density matrices (the set of all positive, semi-definite, Hermitian operators on $\H_N$ of trace one that represent ensemble states) is denoted as $\mathcal{D}_N$.

To a multi-particle Hamiltonian $H$ corresponds a new graph $G(H)$ which we will refer to as the \emph{fermionic graph}. The labels of the vertices in this graph will consequently be indexed by the same multi-index $I$ as the many-particle basis $\{e_I\}_I$. The topology of $G(H)$ is by Definition~\ref{def:H-graph} determined by the Hamiltonian $H$ for the $N$ fermions for which in general we will take
\be\label{eq:def-H}
H = h_1+ \ldots + h_N + W,
\ee
where
\be
h_k = 1 \otimes \ldots \otimes 1\otimes  h \otimes 1 \otimes \ldots \otimes 1
\ee
with $h$ appearing on position $k$ in the $N$-fold tensor product. The term $W$ is a general interaction that we do not further specify at the moment. Yet, its most important features must be that it is self-adjoint and leaves the Hamiltonian invariant under any permutation of the $N$ particles to ensure the fermionic nature of the particles.
Often, it will be convenient to use second quantization to represent the Hamiltonian and other operators, which among other things will guarantee that the system has the right permutational symmetry. For this reason, for $v,w_k \in \H_1$ we define 
\begin{align}
\hat{a}_v^\dagger (w_1 \wedge \ldots \wedge w_N) &= v \wedge  w_1 \wedge \ldots \wedge w_N \quad\quad\text{and}\\
\hat{a}_v (w_1 \wedge \ldots \wedge w_N) &= \sum_{k=1}^N (-1)^{k+1} \langle v ,w_k \rangle \\
&\quad w_1 \wedge \ldots \wedge w_{k-1} \wedge w_{k+1} \wedge \ldots \wedge w_N. \nonumber
\end{align}
These creation and annihilation operators are, as the notation suggests, each others adjoints and satisfy the following anti-commutation relations
\be
[ \hat{a}_v, \hat{a}_w ]_+ = [ \hat{a}^\dagger_v, \hat{a}^\dagger_w ]_+ = 0 \;\; \text{and} \;\; [ \hat{a}_v, \hat{a}^\dagger_w ]_+ = \langle v, w \rangle.
\ee
In the special case that $v=e_i$ we write $\hat{a}_v=\hat{a}_i$ and $ \hat{a}^\dagger_v=\hat{a}_i^\dagger$, thus $ [ \hat{a}_i, \hat{a}^\dagger_j ]_+=\delta_{ij}$.
In terms of these operators, Hamiltonian \eqref{eq:def-H} is then written in a compact form as
\be\label{eq:many-particle-H}
H = \sum_{i,j} h_{ij} \hat{a}^\dagger_i \hat{a}_j + W.
\ee
The interaction term $W$ was not specified, but a very relevant case is that of a two-body interaction of the form
\be
W= \frac{1}{2} \sum_{i,j} w_{ij} \, \hat{a}^\dagger_i \hat{a}^\dagger_j  \hat{a}_j  \hat{a}_i
\ee
with $w_{ij}=w_{ji} \in \R$
which is diagonal in the $\{e_I\}_I$ basis. This interaction has a similar form as the familiar Coulomb interaction in continuum systems, which is also diagonal in position basis. 
Form \eqref{eq:many-particle-H} without an external potential $v$ will later be denoted as the internal part $H_0$ that typically contains the kinetic energy and the interaction $W$ (usually of the form of a two-body interaction like above). An external potential $v : X \rightarrow \R$ (equivalent to $v\in\R^M$) alone acting on the many-particle wave function becomes
\be 
V = \sum_i v_i \hat{a}^\dagger_i \hat{a}_i.
\ee
When particular attention is paid to the dependency of $H$ on the external potential $v$, we denote it as $H=H(v)=H_0+V$. A lemma about the connectedness of the fermionic graph concludes this section.

\begin{lemma}\label{lem:ferm-graph-connected}
Let $H$ be like in \eqref{eq:many-particle-H}, with $W$ being diagonal in the $\{e_I\}_I$ basis and $G(h)$ connected, then $G(H)$ is connected as well.
\end{lemma}

\begin{proof}
We start by taking multi-indices $I=(i_1,\ldots,i_k,\ldots,\allowbreak i_N),I'=(i_1,\ldots,i'_k,\ldots,i_N)$ that just differ in one single entry. Then
\be\label{eq:H-matrix-el}
    \langle e_I,H e_{I'} \rangle = \sum_{i,j} h_{ij} \langle \hat{a}_i e_I, \hat{a}_j e_{I'} \rangle = h_{i_k i'_k},
\ee
so $I \sim I'$ on the fermionic graph if $i_k \sim i'_k$. Proving that $G(H)$ is connected can then be equivalently reformulated as a game-theoretical problem: Put $N$ pawns on the vertices $i\in I$ of $G(h)$, then each turn move one pawn along an edge where no vertex can be doubly occupied. Is it possible to reach an arbitrary, different configuration $J=(j_1,\ldots,j_N)$ this way? In showing this, we first reduce $G(h)$ to a spanning tree \cite[Prop.~1.5.6]{diestel-graph-theory-book} with less edges.
Then we conduct the following algorithm: By cutting a single edge, arbitrarily split the tree graph into two subgraphs, $G_1,G_2$, that are also always trees. If both subgraphs incidentally contain the correct number of pawns in order to fill the desired positions within them, repeat the algorithm for both of them. Else, let $G_1$ be the one that contains too many pawns that are instead missing in $G_2$. If necessary, make space in the treelike $G_2$ by moving pawns towards the leaves, which is always possible since this graph contains too little pawns. Then move the desired number of pawns from $G_1$ to $G_2$ along the cut edge in the original graph. Then again repeat the algorithm for the treelike subgraphs that still need to be adjusted. Since the considered subgraphs shrink continuously, the correct configuration will be eventually reached.
The collection of all turns involved in this procedure yields a path between $I$ and $J$ and thus shows that $G(H)$ is connected.
\end{proof}

Note that this lemma cannot make any general statement about a Hamiltonian with an interaction term that is not diagonal, since that might disconnect the fermionic graph by adding a further term in \eqref{eq:H-matrix-el}.

\subsection{Densities and \texorpdfstring{$N$}{N}-representability}
\label{sec:N-rep}

We now turn our attention to the particle density that is the principal protagonist of DFT. We consider an $N$-particle fermionic system and let $\Psi\in \mathcal{I}_N$, which we can expand in the basis $\{e_I\}_I$ with coefficients $\{\Psi_I\}_I$ as
\be
\Psi = \sum_I \Psi_I \, e_I,
\label{Psi_expansion}
\ee
in which we sum over all ordered multi-indices $I = (i_1, \ldots,i_N)$.
Since the state $\Psi$ was assumed to be normalized we further have
\be
1 = \sum_I |\Psi_I |^2.
\label{normalize}
\ee
The particle density $\rho_i$ at vertex $i$ is now defined as 
\be
\rho_i = \sum_{I \ni i} |\Psi_I |^2,
\label{rho_def}
\ee
where we sum just over the multi-indices $I$ that include the given vertex $i\in X$.
From this it immediately follows by comparison with \eqref{normalize} that $0 \leq \rho_i \leq 1$. More generally, we will denote by $\rho$ the mapping $\rho: X \to [0,1]$ that assigns $\rho_i$ to vertex $i$ and refer to this quantity as ``the density $\rho$''. The creation and annihilation operators can be used for the alternative expression
\begin{equation}\label{eq:def-rho}
	\rho_i = \|\hat a_i \Psi\|^2 = \langle \hat a_i \Psi,\hat a_i \Psi \rangle = \langle \Psi,\hat a^\dagger_i \hat a_i \Psi \rangle,
\end{equation}
or with an ensemble state $\Gamma\in\mathcal{D}_N$ as $\rho_i = \Tr (\Gamma \hat a^\dagger_i \hat a_i)$.
The quantity $\rho_i$ receives a physical interpretation as the probability to find a particle at vertex $i$
if we know the system to be in state $\Psi$ ($\Gamma$). 
Another property of the density that follows from our definitions is that
\be
N = \sum_{i} \rho_i.
\label{rho_sum}
\ee
The fact that the probabilities do not sum up to one is not a contradiction, but is due to the fact that these probabilities are non-exclusive, i.e.,
the fact that we find a particle at vertex $i$ does not exclude the possibility of finding another particle at $j$. Indeed the state $e_I$ gives
probability one for finding a particle at any vertex $i$ contained in $I$.
Let the corresponding density be denoted by $E_I$, such that $E_{I,i}=1$ if $i \in I$ ($N$ times) and $0$ otherwise ($M-N$ times). An alternative definition using the annihilation operators is $E_I = (\|\hat a_i e_I\|)_i$, because $\|\hat a_i e_I\| = 1$ if $i\in I$ and $0$ otherwise.
We collect all densities satisfying the above conditions in the set of physical densities,
\be\label{eq:def-P_N}
    \mathcal{P}_{M,N} = \left\{ \rho : X \to \R \;\middle|\; 0 \leq \rho_i \leq 1, \sum_{i=1}^M\rho_i = N \right\}.
\ee
Note that the set $\mathcal{P}_{M,N}$ is a $(M-1)$-simplex with normalization $N$, cropped by the $\rho_i\leq 1$ conditions, resulting in an $(M-1)$-dimensional convex polytope that is known as $(M,N)$-hypersimplex \cite{convex-polytopes-book,Rispoli2008-hypersimplex}. This means the border points of $\mathcal{P}_{M,N}$ have either at least one $\rho_i=0$ and lie on the border of the simplex or have at least one $\rho_i=1$. Conversely, if for all $i\in X$ a density has $0<\rho_i<1$ it lies in the interior of $\mathcal{P}_{M,N}$, a set later denoted as $\mathcal{P}_{M,N}^+$. The extreme points of $\mathcal{P}_{M,N}$ are then those densities that saturate all conditions, exactly the vectors $\{ E_I \}_I$, which lie at the corners and generate the convex polytope $\mathcal{P}_{M,N}$ \cite[§2.4]{convex-polytopes-book}.
This insight will allow us to show that all densities in $\mathcal{P}_{M,N}$ are (pure-state) $N$-representable, i.e., for any density in $\mathcal{P}_{M,N}$ there is a state in $\mathcal{I}_N$ that yields exactly this density. 

\begin{proposition}\label{prop:N-rep}
$\mathcal{P}_{M,N}$ is the maximal set of pure-state $N$-representable densities.
\end{proposition}

\begin{proof}
Because of convexity any density in $\mathcal{P}_{M,N}$ can be written as $\rho_i = \sum_I \rho_I E_{I,i} = \sum_I \rho_I \|\hat a_i e_I\|$ with $\rho_I \in [0,1]$ and $\sum_I\rho_I=1$. We rewrite $\rho_I = |\Psi_I|^2$ with arbitrary phase choice for $\Psi_I$ and substitute $\|\hat a_i e_I\|$ by $\|\hat a_i e_I\|^2$ which holds since it has a $(0,1)$-value anyway. Then
\begin{equation}
\begin{aligned}
    \rho_i &= \sum_I |\Psi_I|^2 \|\hat a_i e_I\|^2 = \sum_{I,I'} \Psi_I^*\Psi_{I'} \langle \hat a_i e_I, \hat a_i e_{I'} \rangle \\
    &= \left\langle \hat a_i \sum_{I} \Psi_I e_I, \hat a_i \sum_{I'} \Psi_{I'} e_{I'} \right\rangle = \langle \hat a_i \Psi, \hat a_i \Psi \rangle 
\end{aligned}
\end{equation}
with $\Psi = \sum_I \Psi_I e_I$,
where the extension to a double sum is possible since $\langle \hat a_i e_I, \hat a_i e_{I'} \rangle = 0$ if $I\neq I'$. That the given set is maximal is clear because no (physical) density can lie outside of $\mathcal{P}_{M,N}$.
\end{proof}

This demonstrates pure-state $N$-representability for $\rho$ and of course automatically implies ensemble $N$-representability, although a simpler construction is feasible in this case by just forming the convex combination $\Gamma = \sum_I\rho_I\Gamma_I$ where $\Gamma_I$ are the projections onto the states $e_I$.
Because of the proposition above it indeed makes sense to take $\mathcal{P}_{M,N}$ as the set of all possible physical densities. In the continuum setting a comparable construction is possible but way more complicated \cite{harriman1981orthonormal}. The much subtler matter of $v$-representability will be discussed in Section~\ref{sec:v-rep}.

\subsection{Graph Laplacians}
\label{sec:graph-Laplacian}

In the previous discussions we left the specific form of $h$ open, but here we want to make a brief connection to the form of the Hamiltonian that appears, for example, in finite-difference approaches for discretizing the continuum Schr\"odinger operator.
This requires the graph theoretical definition of the Laplace operator representing the kinetic energy.

\begin{definition}
For a graph $G$, let the degree $d(i)$ of a vertex $i$ be the number of vertices adjacent to $i$. Then the graph Laplacian $\Delta$ is defined here as
\be\label{eq:graph-laplacian}
\Delta_{ij} = \left\{ \begin{array}{ll}  -d(i) \quad  & \mbox{if $i=j$} \\  1 & \mbox{if $i \sim j$}  \\ 0  & \mbox{otherwise.} \end{array} \right.
\ee
\end{definition}
For a cubic lattice this definition is readily seen to agree with the standard three-point finite-differencing formula for the second derivative.
Note that in graph theory usually an opposite sign convention is employed and the graph Laplacian is $L=-\Delta$ \cite{spectral-graph-book,graph-eigenvectors-book}. We used the definition given here such that the negative graph Laplacian fulfills the same semi-positivity condition as the continuum Laplacian,
\begin{align}
- \langle \psi, \Delta \psi \rangle &= -\sum_{i,j} \psi_i^* \Delta_{ij} \psi_j = \sum_i d(i) |\psi_i|^2 - \sum_{i\sim j} \psi_i^*\psi_j \nonumber\\&= \frac{1}{2} \sum_{i\sim j} |\psi_i-\psi_j|^2 \geq 0
\end{align}
for any $\psi \in \H_1$, which physically implies positivity of the kinetic energy.
With this ingredient we now define the one-particle Hamiltonian on a graph in analogy to the usual Schrödinger equation,
\be\label{eq:h-with-v}
h_{ij} = - \Delta_{ij} + v_i \delta_{ij},
\ee
where the first term represents the kinetic energy and the second term represents the external one-body potential. More precisely, we define the potential $v$
to be the mapping $v:X \to \mathbb{R}$ that assigns to every vertex $i$ the real number $v_i$. 
Note, however, that the results derived in the course of this paper do not depend on any special form of $h$, just on the local character of the external potential as in \eqref{eq:h-with-v}.

Contrary to the kinetic energy operator it is not possible to straightforwardly define a momentum operator on a general graph. Being the generator of translations such an operator can only be specified on graphs with certain symmetries, such as circular graphs.

\subsection{Simple graph examples}
\label{sec:graph-examples}

Two very simple graphs will be used to illustrate the above constructions, both of them also taking a role in providing counterexamples to the standard Hohenberg--Kohn statement formulated later in Section~\ref{sec:ex-non-unique}.
The first is the situation of two fermions of the triangle graph, so $N=2$ and $M=3$. The corresponding graph Laplacian is
\be
\Delta = \begin{pmatrix*}[r]
-2 & 1 & 1  \\
1 & -2 & 1  \\
1 & 1 & -2  
\end{pmatrix*}.
\ee
The simplest multi-particle Hamiltonian is then $H=-\Delta\otimes 1 - 1 \otimes\Delta$, not including any interactions or an external potential. Its matrix representation in the resulting fermionic Hilbert-space basis $\{e_I\}_I = \{e_1 \wedge e_2 , e_1 \wedge e_3,e_2 \wedge e_3\}$ with $L=3$ elements is then
\be
H = \begin{pmatrix*}[r]
4 & -1 & 1  \\
-1 & 4 & -1  \\
1 & -1 & 4  
\end{pmatrix*},
\ee
which can be calculated by applying $H$ on the basis elements as usual. For the first basis element this means
\begin{align}
H(e_1\wedge &e_2) = -\Delta e_1\wedge e_2 - e_1 \wedge \Delta e_2 \nonumber\\
&= -(-2e_1+e_2+e_3)\wedge e_2 - e_1 \wedge (e_1-2e_2+e_3) \nonumber\\
&= 2e_1\wedge e_2 - e_3\wedge e_2 + 2e_1\wedge e_2 - e_1\wedge e_3 \nonumber\\
&= 4e_1\wedge e_2 - e_1\wedge e_3 + e_2\wedge e_3 \quad\text{etc.},
\end{align}
where $e_i\wedge e_j = -e_j\wedge e_i$ and $e_i\wedge e_i=0$ have been used. This means the fermionic graph $G(H)$ is a complete graph as well, i.e., with edges between all vertices. The graphs $G(\Delta)$ and $G(H)$ are illustrated in Figure~\ref{fig:triangle-graph}.

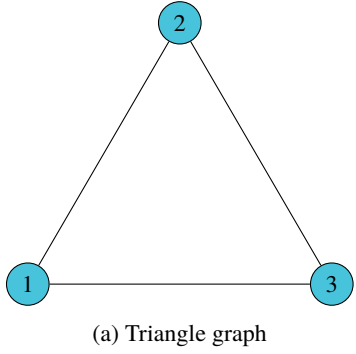
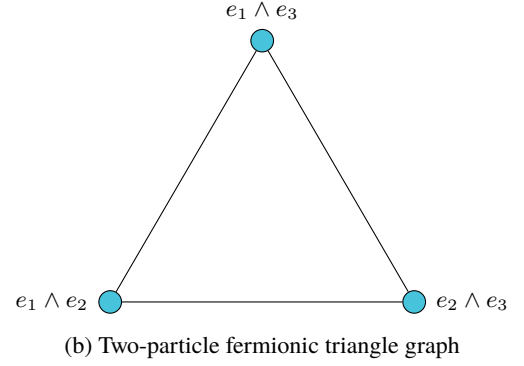
\begin{figure*}[ht]
\begin{subfigure}{.45\textwidth}
  \centering
  \begin{tikzpicture}[scale=2.0]
    \coordinate (C1) at (-1,0);
    \coordinate (C2) at (0,{sqrt(3)});
    \coordinate (C3) at (1,0);
    \foreach \m in {1,...,3}
    	\node[graphnode] (N\m) at (C\m) {\m};
    \draw (N1) -- (N2) -- (N3) -- (N1);
  \end{tikzpicture}
  \caption{Triangle graph}
\end{subfigure}
\hfill 
\begin{subfigure}{.45\textwidth}
  \centering
  \begin{tikzpicture}[scale=2.0]
    \coordinate (C1) at (-1,0);
    \coordinate (C2) at (0,{sqrt(3)});
    \coordinate (C3) at (1,0);
    \foreach \m in {1,...,3}
    	\node[graphnode] (N\m) at (C\m) {};
    \node[left=5pt] at (C1) {$e_1\wedge e_2$};
    \node[above=5pt] at (C2) {$e_1\wedge e_3$};
    \node[right=5pt] at (C3) {$e_2\wedge e_3$};
    \draw (N1) -- (N2) -- (N3) -- (N1);
  \end{tikzpicture}
  \caption{Two-particle fermionic triangle graph}
\end{subfigure}
\caption{Triangle graph example.}
\label{fig:triangle-graph}
\end{figure*}

The second example is two fermions on a square graph, so $N=2$ and $M=4$ with \be
\Delta = \begin{pmatrix*}[r]
-2 & 1 & 0 & 1  \\
1 & -2 & 1 & 0  \\
0 & 1 & -2 & 1  \\
1 & 0 & 1 & -2  
\end{pmatrix*},
\ee
so no diagonal edges are included. The basis set for the fermionic Hilbert space then is 
$\{ e_I \}_I = \{ e_1 \wedge e_2, e_1 \wedge e_3, e_1 \wedge e_4, e_2 \wedge e_3, e_2 \wedge e_4 , e_3 \wedge e_4 \}$
and the same Hamiltonian $H=-\Delta\otimes 1 - 1 \otimes\Delta$ as before in this basis comes out as
\be
H = \begin{pmatrix*}[r]
4 & -1 & 0 & 0 & 1 & 0  \\
-1 & 4 & -1 & -1 & 0 & 1  \\
0 & -1 & 4 & 0 & -1 & 0 \\
0 & -1 & 0 & 4 & -1 & 0 \\
1 & 0 & -1 & -1 & 4 & -1 \\
0 & 1 & 0 & 0 & -1 & 4
\end{pmatrix*}.
\ee
Here, the fermionic graph, such as the square graph itself, is not complete and some edges are missing, as illustrated in Figure~\ref{fig:square-graph}. The same example was used in \citet{tenHaaf1995MonteCarlo} to illustrate a method employing simplified, effective Hamiltonians in Green-function Monte Carlo calculations, where certain edges in the fermionic graph get cut, thus avoiding the ``sign problem''.

\begin{figure*}[ht]
\begin{subfigure}{.45\textwidth}
  \centering
  \begin{tikzpicture}[scale=2]
    \coordinate (C1) at (-1,1);
    \coordinate (C2) at (1,1);
    \coordinate (C3) at (1,-1);
    \coordinate (C4) at (-1,-1);
    \foreach \m in {1,...,4}
    	\node[graphnode] (N\m) at (C\m) {\m};
    \draw (N1) -- (N2) -- (N3) -- (N4) -- (N1);
  \end{tikzpicture}
  \caption{Square graph}
\end{subfigure}
\hfill 
\begin{subfigure}{.45\textwidth}
  \centering
  \begin{tikzpicture}[scale=2.0]
    \coordinate (C1) at (-1.25,0);
    \coordinate (C2) at (0,1);
    \coordinate (C3) at (1.25,0);
    \coordinate (C4) at (0.75,0);
    \coordinate (C5) at (0,-1);
    \coordinate (C6) at (-0.75,0);
    \foreach \m in {1,...,6}
    	\node[graphnode] (N\m) at (C\m) {};
    \node[left=5pt] at (C1) {$e_1\wedge e_2$};
    \node[above=5pt] at (C2) {$e_1\wedge e_3$};
    \node[right=5pt] at (C3) {$e_1\wedge e_4$};
    \node[left=5pt] at (C4) {$e_2\wedge e_3$};
    \node[below=5pt] at (C5) {$e_2\wedge e_4$};
    \node[right=5pt] at (C6) {$e_3\wedge e_4$};
    \draw (N1) -- (N2) -- (N3) -- (N5) -- (N1);
    \draw (N2) -- (N4) -- (N5) -- (N6) -- (N2);
  \end{tikzpicture}
  \caption{Two-particle fermionic square graph}
\end{subfigure}
\caption{Square graph example.}
\label{fig:square-graph}
\end{figure*}
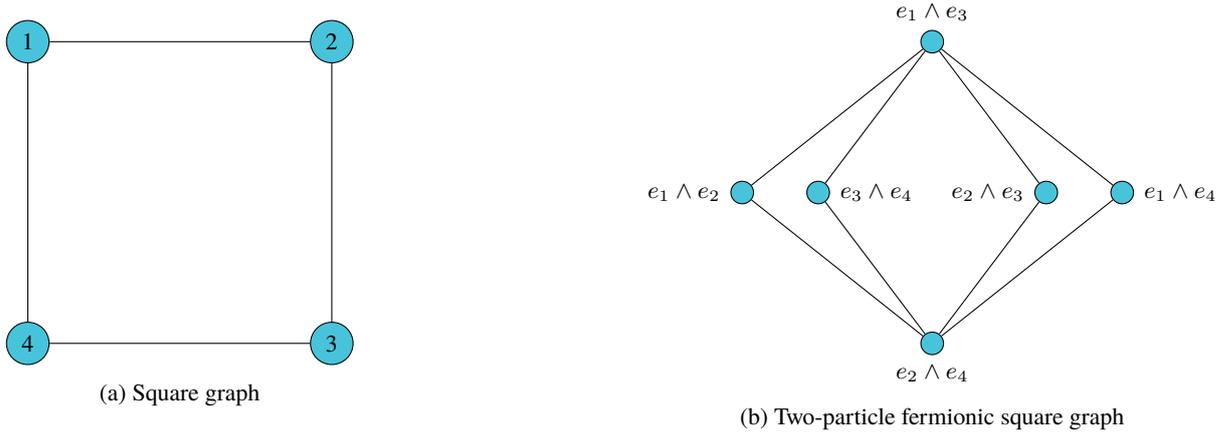

\section{Violations of the Hohenberg--Kohn theorem}
\label{sec:main:HK-violation}

\subsection{Unique \texorpdfstring{$v$}{v}-representability}
\label{sec:uv}

The Hohenberg-Kohn (HK) theorem for continuum systems \cite{Hohenberg-Kohn1964} states that if two Hamiltonians $H(v)$ and $H (v')$ corresponding to two different external potentials $v$ and $v'$
share the same ground-state density then $v$ and $v'$ can only differ by a constant. This statement holds independently whether the ground state is degenerate or not. Following the proof in \citet{pino2007} the theorem can be split into two parts, where the first one is very general and holds in all settings where the external potential couples only to the density, so in particular in the usual continuum case and the graph setting discussed here.

\begin{theorem}[HK part 1]\label{th:HK1}
Assume that two Hamiltonians $H = H_0 + V,H' = H_0 + V'$ that differ only in their external potentials $v,v'$ share a common ground-state density $\rho$. Then an (ensemble) ground state $\Psi$ $(\Gamma)$ of $H$ with density $\rho$ is also an (ensemble) ground state of $H'$ and vice versa.
\end{theorem}

\begin{proof}
Let $\Psi,\Psi'$ be ground states of $H$ and $H'$ with ground-state energies $E,E'$, respectively, having the same density $\rho$. Then by the Rayleigh--Ritz variational principle
\begin{align}
    &E = \langle \Psi, H \Psi \rangle \leq \langle \Psi', H \Psi' \rangle = \langle \Psi', H' \Psi' \rangle + \sum_i (v_i-v'_i)\rho_i \nonumber\\
    &\quad\Longrightarrow \langle \Psi, H \Psi \rangle - \sum_i (v_i-v'_i)\rho_i \leq \langle \Psi', H' \Psi' \rangle
    \label{eq:HK-proof-1}
\end{align}
and thus
\begin{equation}
\begin{aligned}
    \label{eq:HK-proof-2}
    E' &= \langle \Psi', H' \Psi' \rangle \leq \langle \Psi, H' \Psi \rangle \\&= \langle \Psi, H \Psi \rangle - \sum_i (v_i-v'_i)\rho_i \leq E'.
\end{aligned}
\end{equation}
Hence $\langle \Psi, H' \Psi \rangle = E'$ and thus $\Psi$ is also a ground state of $H'$. The other direction and the proof for ensemble ground states works analogously.
\end{proof}

Since by this first part of the HK theorem it follows $H'\Psi = E'\Psi$, this together with $H\Psi=E\Psi$ leads to
\begin{equation}
    (H-H')\Psi=(V-V')\Psi=(E-E')\Psi.
\end{equation}
The second part of the HK theorem for continuum systems then relies on the assumption that $\Psi\neq 0$ almost everywhere to be able to conclude $v-v'=\text{const}$.
Recently this issue, and thereby the validity of the HK theorem for continuum systems, has been settled in the work of \citet{Garrigue2018} who used the proof structure of \citet{pino2007} and proved the necessary unique continuation property (UCP) for the many-particle Hamiltonians under consideration.
In essence, the UCP assures that if a wave function $\Psi$ vanishes sufficiently quickly at one point and solves the Schr\"odinger equation then $\Psi=0$ almost everywhere.
This then implies that an eigenfunction cannot vanish on a set of positive measure or else it would be identically zero \cite{FigueiredoGossez,Regbaoui}, a property that is crucially used in the proof of the HK theorem.

The UCP or a comparable property does not hold on finite graphs \cite{Isozaki2014-Rellich-no-UCP,Peyerimhoff2019-no-UCP}, since graph Hamiltonian eigenstates can vanish at many adjacent vertices (see \eqref{eq:complete-graph:localized-phi} for an example).
This also implies that, although the HK theorem is settled in the continuum case, it is not so in the discrete case. This can lead to situations where systems with different potentials (modulo constants) have the same ground state and consequently also the same ground-state density. On the other hand, because of the first part of the HK theorem, if two different potentials produce the same ground-state density, they must also share the same ground state. Hence, the property of ``unique $v$-representability'' (or its negation) is always shared by states and their densities. For future use it is therefore possible and useful to introduce the following terminology:

\begin{definition}
If $\Psi$ ($\Gamma$) is the (ensemble) ground state of a Hamiltonian $H_0+V$ with external potential $v$ and if there is no other potential $v'$ differing more than a constant from $v$ such that $\Psi$ ($\Gamma$) is also the (ensemble) ground state of $H_0+V'$, we say that $\Psi$ ($\Gamma$) and its corresponding ground-state density $\rho$ are \emph{uniquely $v$-representable} (uv).
\end{definition}

The validity of the complete HK theorem is then equivalent to the statement that all densities of interest are uv-ground-state densities.
In the next section we will establish sufficient conditions
for unique $v$-representability, even for general eigenstates, which is in line with the continuum case in which the UCP applies to any eigenstate \cite{Garrigue2018}.

It is readily seen that violations of unique $v$-representability and thus of the HK theorem occur in at least two cases, namely when the density $\rho_i$ on a vertex is either zero or one, i.e., the density lies on the boundary of $\mathcal{P}_{M,N}$.
In the case that the density is zero on vertex $i$ we can add an arbitrary repulsive potential $v_i$ at vertex $i$ without any effect on the potential energy $\sum_i v_i\rho_i$ and thus without changing the state $\Psi$ and its designation as a ground state.
On the other hand, in the case $\rho_i=1$ the density is fully saturated at that vertex and by applying an additional attractive potential on vertex $i$ we can make sure that the state $\Psi$ is unchanged and still remains the ground state.
These two situations can be ruled out if the ground-state wave function has enough non-zero coefficients in its expansion
of the form \eqref{Psi_expansion}.
To show this we first define 
\be
m= \binom{ M-1 }{ N-1}, \quad n= \binom{ M-1 }{ N }, \quad n+m= \binom{ M }{ N },
\label{nm}
\ee
where $n+m=L$ is the dimension of the Hilbert space $\H_N$.
If $\rho_i=0$ then according to \eqref{rho_def} 
at most $n$ coefficients, namely all possible ones not containing $i$ in a multi-index, can be non-zero. On the other hand
when $\rho_i=1$ then all non-zero coefficients must contain
$i$ in a multi-index and there are at most $m$ of such coefficients.
If $M > 2N$ then $n > m$ which means that, in this case, expansion \eqref{Psi_expansion} must contain more than $n$ non-zero coefficients to avoid both zero or full occupancy on some vertex.
On the other hand, if $M < 2N$ then $m > n$ and the expansion
must contain more than $m$ non-zero coefficients to avoid 
the same situation. For the remaining case $M=2N$ we have $m=n=\binom{2N-1}{N}$ and we need more than this number of coefficients to avoid either zero or 
full occupancies.
It is clearly conceivable that the HK theorem may also be violated in other cases than zero or full occupancy, i.e., inside the interior of $\mathcal{P}_{M,N}$, and we will see that this is indeed the case by explicit example.
In the next section we derive precisely how many coefficients in expansion \eqref{Psi_expansion} need to be non-zero to guarantee the validity of the HK theorem.

\subsection{Condition for unique \texorpdfstring{$v$}{v}-representability from Odlyzko's theorem}
\label{sec:Odlyzko}

In this section we give a sufficient condition for a state to be uniquely $v$-representable on graphs.
In the HK theorem the focus is on the ground state but our reasoning applies to any eigenstate
and so we will derive a condition to ensure that any eigenstate is produced by a unique potential (modulo constants). 
Let us assume that two potentials $V$ and $V'$ lead to the same eigenstate $\Psi$, then
\begin{align}
( H_0 + V ) \Psi &= E \Psi \quad\quad \text{and} \\
( H_0 + V' ) \Psi &= E' \Psi.
\end{align}
Then if we define $U= (V-E)- (V'-E')$ it follows from above that $U \Psi=0$. Written out more explicitly this gives for every coefficient from the expansion \eqref{Psi_expansion} that
\be\label{eq:u-condition}
(u_{i_1} + \ldots + u_{i_N} ) \, \Psi_I = 0
\ee
with $u_i = v_i-v'_i - \tfrac{1}{N}(E-E')$
for all $i\in I=(i_1,\ldots,i_N)$ with $i_1<\ldots<i_N$, which represents a selection of $N$ distinct integers from the set of vertices $X=\{1,\ldots,M \}$.
We would like to conclude that $u_i=0$ for all $i \in X$ which leads to the usual HK result that
$v = v' + \tfrac{1}{N}(E-E')$. To ensure this conclusion we need enough multi-indices $I$ for which the wave-function coefficient $\Psi_I$ does not vanish. Let us denote the collection of all such multi-indices for $\Psi$ by $C[\Psi]$, so $\Psi_I \neq 0$ for $I \in C[\Psi]$. Note that there is a finite number of possibilities for such collections because every single one cannot include more than $L=\binom{M}{N}$ elements. Furthermore, every state $\Psi \in I_N$ can be allocated to one of those classes, namely $C[\Psi]$, so they offer a finite classification of all states in terms of which coefficient are non-zero.
We then have from \eqref{eq:u-condition} the set of equations
\be
\sum_{i \in I} u_{i} = 0 \quad \quad \text{for all $I \in C[\Psi]$.}
\label{u_zero}
\ee
The key question is now how large the collection $C[\Psi]$ must be in order to guarantee that $u_i=0$ for all $i \in X$ is the only solution, which would give us the full HK theorem for states in this class. This can be translated into a question about the rank of certain types of matrices and can be best illustrated with an example that we will meet again in \eqref{Psi_A} below. If
$M=4$ and $N=2$ and the wave-function coefficient $\Psi_I$ is non-vanishing for $I \in C[\Psi]=\{(1,3),(1,4),(2,3),(2,4)\}$ then the system of equations (\ref{u_zero}) is equivalent to the matrix equation
\be\label{eq:rank-3}
\Upsilon[\Psi]\cdot u = 
\begin{pmatrix*}
1 & 0 & 1 & 0 \\
1 & 0 & 0 & 1 \\
0 & 1 & 1 & 0 \\
0 & 1 & 0 & 1 
\end{pmatrix*} \cdot
\begin{pmatrix*} u_1 \\ u_2 \\ u_3 \\ u_4
\end{pmatrix*} = 0.
\ee 
The rows of the coefficient matrix $\Upsilon[\Psi]$ are thus formed by the non-zero coefficients $\Psi_I$ of $\Psi$. In each row just the entries that correspond to the two indices in $I$ are one, while all other entries remain zero. Equivalently, the rows of $\Upsilon[\Psi]$ are the extreme densities $E_I$ with $I\in C[\Psi]$.
The rank of $\Upsilon[\Psi]$ in this example is $3$ which is less than the number of unknowns and consequently there is an infinite number of solutions, which in this case are of the form $u=(t,t,-t,-t)$ with $t \in \mathbb{R}$.
If we would now take $\Psi_I$ also to be non-zero for $I=(3,4)$, we can add the row $(0,0,1,1)$ to the matrix above and we find that the rank of the matrix becomes $4$ in which case $u=0$ is the only solution.

In the general case of \eqref{u_zero},
the matrix $\Upsilon[\Psi]$ that we need to consider always contains only the elements $0$ and $1$, a so-called $(0,1)$-matrix, has $M$ columns, and its row sum is always $N$. To guarantee uniqueness, we can therefore ask how many different rows a $(0,1)$-matrix with $M$ columns and row sum $N$ needs to have in order to guarantee that its rank is $M$.
This problem was addressed by \citet{Longstaff} and \citet{Odlyzko} in which the latter gave the complete answer that we repeat here as a theorem.

\begin{theorem}[Odlyzko]
The number $g(M,N)$, $1 \leq N < M$, defined by
\be
g(M,N) = \left\{ \begin{array}{cc}  \binom{M-1}{N} & \text{if  $M > 2N$}\\ \\
2 \, \binom{M-2}{N-1} & \text{if $M=2N$}  \\ \\
\binom{M-1}{N-1} & \text{if $M < 2N$}
\end{array} \right.
\label{odlyzko} 
\ee
is the smallest integer so that every $(0,1)$-matrix with $M$ columns, row sum $N$, and strictly more than $g(M,N)$ different rows has rank $M$.
\end{theorem}

In our example we have $g(4,2)=4$ such that we need at least $g(4,2)+1=5$ distinct rows, and therefore $5$ non-vanishing coefficients $\Psi_I$ to guarantee the unique solution $u=0$. Note that Odlyzko's theorem still allows for the existence of
rank $M$ matrices with a number of rows smaller or equal than $g(M,N)$, it just asserts that not \emph{every} matrix with that number of different rows will have rank $M$. For example, we can check for $M=4,N=2$ that a wave function with non-zero coefficients $\Psi_I$ with $I\in C[\Psi]=\{(1,2),(1,3),(1,4),(3,4)\}$ corresponds to a $(0,1)$-matrix $\Upsilon[\Psi]$ of rank $4$ even if this matrix has only $4$ different rows.

If we now apply the theorem of Odlyzko to the system of equations \eqref{u_zero} we conclude that
$u=0$ is the only possible solution in case the number of coefficients $|C[\Psi]|$ for which $\Psi_I$ does not vanish is strictly larger than $g(M,N)$. We formulate this as the following corollary.

\begin{corollary}
\label{cor:HK-necessary}
Let $\Psi$ be an $N$-particle eigenstate of $H_0+V$ on a graph with $M$ vertices. If the number of distinct multi-indices $I$ for which $\Psi_I$ in expansion \eqref{Psi_expansion} does not vanish is strictly larger than the Odlyzko number $g(M,N)$ given by \eqref{odlyzko}, then $\Psi$ is not an eigenvector of any other Hamiltonian $H_0+V'$ where $V'$ differs from $V$ by more than a constant.
\end{corollary}

Let us now give a discussion of these results.
Referring back to $\eqref{nm}$ we see that for $M > 2N$ we have $g(M,N)=n$ and for $M < 2N$ we indeed have $g(M,N)=m$ such that the conditions in Corollary~\ref{cor:HK-necessary} 
for the case $M \neq 2N$ are equivalent to the conditions that exclude the $\rho_i=0$ and $\rho_i=1$ occupation cases that were discussed in the previous section. 
For the case $M=2N$ we have $m=n$ and $g(M,N)>m$ as can be checked by explicit calculation. 
In this case the Odlyzko condition is more strict than the one obtained from the reasoning in the previous section, which indicates the possibility of a violation of HK without zero or full occupancy as we will indeed find by explicit example in the next section.
We stress, however, that this does not necessarily imply HK violation for states where the number of non-zero coefficients is smaller or equal than $g(M,N)$.

\subsection{Examples for non-unique \texorpdfstring{$v$}{v}-representability}
\label{sec:ex-non-unique}

Here we present two explicit examples for non-unique $v$-representability by studying two systems of non-interacting fermions on graphs that are already familiar to us from Section~\ref{sec:graph-examples}. 
As a first example we again take two fermions on a triangle graph, i.e., we have the case
$N=2$ and $M=3$ for which we get $g(3,2)=2$ and therefore we need at least three non-zero coefficients (which means all of them since $L=\binom{M}{N}=3$ too) to guarantee unique $v$-representability by Corollary~\ref{cor:HK-necessary}.
For later reference we give the Hamiltonian for a general external potential $v$, although for our example it suffices to put the potential to zero on all but one vertex.
The one-particle Hamiltonian $h$ is then given by
\be
h(v) = 2 I +
\begin{pmatrix*}[r]
v_1 & -1 & -1  \\
-1 & v_2 & -1  \\
-1 & -1 & v_3  
\end{pmatrix*},
\ee
where $I$ denotes the identity matrix here, while the two-particle Hamiltonian is given by
\be
H(v) = 4 I +
\begin{pmatrix*}[r]
v_1+v_2 & -1 & 1  \\
-1 & v_1+v_3 & -1  \\
1 & -1 & v_2+v_3  
\end{pmatrix*}
\label{eq:H-triangle}
\ee
with respect to the basis $\{e_I\}_I = \{e_1 \wedge e_2 , e_1 \wedge e_3,e_2 \wedge e_3\}$. For the case $v_1=v_2=v_3=0$ the ground state is degenerate
and the ground-state subspace is spanned by the two orthonormal states
\be
 \Psi_A = \frac{1}{\sqrt{2}} (1,0,-1) \quad\quad \text{and} \quad\quad \Psi_B =\frac{1}{\sqrt{6}} (1,2,1)
 \label{PsiAB}
\ee
with eigenvalue $E_A=E_B= 3$ and corresponding densities calculated from \eqref{rho_def},
\be
\rho_A = \tfrac{1}{2}( 1,2,1 ) \quad\quad \text{and} \quad\quad \rho_B = \tfrac{1}{6}( 5,2,5 ).
\ee
We note that $\Psi_A$ has one zero coefficient which is responsible for the full occupancy on vertex $2$ in the density $\rho_A$ that consequently can be found at the boundary of $\mathcal{P}_{3,2}$.
Through applying an attractive potential $(0,-t,0)$ with $t >0$ the state $\Psi_A$ thus remains unchanged while the state $\Psi_B$ is modified. The 
degeneracy is lifted with $E_A=3-t$ and $E_B= \tfrac{1}{2}(9-t-\sqrt{9-2t+t^2})$ such that $\Psi_A$ becomes a non-degenerate ground state
which is insensitive to the value of $t$ as long as $t>0$. 
This is a clear counterexample to the HK theorem which finds its origin in a full occupation of one of the vertices. It is also a refutation of the statement of \citet{CCR1985} that ``no potential which is finite at a particular site can attract particles so strongly to that site that it always contains the maximum number of particles.''

As a second example we take $N=2$ fermions on a square graph with $M=4$ vertices for which $g(4,2)=4$ while $m=n=3$.
We therefore need at least $4$ non-zero coefficients to avoid zero or full occupancy, while the Odlyzko condition implies that we need at least $5$ non-zero coefficients
to guarantee unique $v$-representability. We shall see that for this example we can find a ground state with only $4$ non-zero coefficients and that HK is violated 
in the absence of full or zero occupancy. 
For the square graph we take the one-particle Hamiltonian
\begin{align}\label{eq:square-h}
h= 2 I +
\begin{pmatrix*}[r]
s+t & -1 & 0 & -1 \\
-1 & -s+t & -1 & 0 \\
0 & -1 & -s-t & -1 \\
-1 & 0 & -1 & s-t
\end{pmatrix*}
\end{align}
in which we chose the external potential in a particular form with $s,t\in\R$.
The system can be viewed as one in which we apply a combination of two potentials: a potential of the form $(s,-s,-s,s)$ with opposite values on 
the left-hand and right-hand side of the square and a 90° rotated one $(t,t,-t,-t)$ with opposite values on the top and bottom of the square.
The two-particle Hamiltonian for this system is given by
\begin{align}
H
= 4 I +
\begin{pmatrix*}[r]
2t & -1 & 0 & 0 & 1 & 0 \\
-1 & 0 & -1 & -1 & 0 & 1 \\
0 & -1 & 2s & 0 & -1 & 0 \\
0 & -1 & 0 & -2s & -1 & 0 \\
1 & 0 & -1 & -1 & 0 & -1 \\
0 & 1 & 0 & 0 & - 1 & - 2t
\label{Ham_uv}
\end{pmatrix*}
\end{align}
with respect to the basis
\be
\{ e_I \}_I = \{ e_1 \wedge e_2, e_1 \wedge e_3, e_1 \wedge e_4, e_2 \wedge e_3, e_2 \wedge e_4 , e_3 \wedge e_4 \}.
\ee
For all values $s,t\in\R$ the two normalized eigenstates and eigenvalues of $H$ corresponding to the two lowest energies 
are given by 
\begin{align}
\label{Psi_A}
\Psi_A = \; &\frac{1}{1+\alpha^2 (s)}(0, \alpha (s), \alpha^2 (s), 1, \alpha (s),0 ) \\& \text{with}\quad E_A=4- 2\sqrt{1+s^2} \quad\text{and} \nonumber\\
\label{Psi_B} 
\Psi_B = \; &\frac{1}{1+\alpha^2 (t)} (\alpha^2 (t), \alpha (t), 0,0,-\alpha(t), -1) \\& \text{with}\quad E_B= 4- 2\sqrt{1+t^2},\nonumber
\end{align}
where we defined $\alpha(t) = -t + \sqrt{1+t^2}$, and the corresponding densities are
\begin{align}
\rho_A &= \frac{1}{1+\alpha^2 (s)} ( \alpha^2 (s),1,1,\alpha^2 (s) ) \quad\quad\text{and} \\
\rho_B &= \frac{1}{1+\alpha^2 (t)} ( \alpha^2 (t),\alpha^2 (t),1,1 ) .
\end{align}
The densities can be rewritten as
\begin{align}
\rho_A &= \beta (s) ( 0,1,1,0 ) +
(1-\beta (s)) (1,0,0,1) 
\label{eq:rhoA}\quad\text{and} \\
\rho_B &= \beta (t) (0,0,1,1) + (1-\beta(t)) ( 1,1,0,0 ),
\label{eq:rhoB}
\end{align}
where we defined $\beta (t) = 1 /(1 + \alpha^2 (t)) \in (0,1)$.
In this case they do \emph{not} lie on the boundary of the density domain because $\alpha (t)$ does not become zero but approaches four extreme points in the density domain $\mathcal{P}_{4,2}$ (an octahedron) arbitrarily close. The densities, parameterized by $s$ and $t$, form the diagonals in the middle plane of the octahedron, a situation displayed in Figure~\ref{fig:octahedron}.
We further see that the densities $\rho_A,\rho_B$ clearly reflect the symmetry of the applied potentials. For example, for very large positive $s$ the function $\beta (s)$ approaches one and two particles localize on vertices $2$ and $3$ where the potential becomes large and negative, while for large negative $s$ the coefficient $\beta(s)$ approaches zero and the two particles localize on the opposite sites of the square. The states $\Psi_A$ and $\Psi_B$ have both $4$ non-zero coefficients so that we cannot have a zero or full occupancy on any vertex for both of them. However, the Odlyzko condition for these states is not satisfied, which opens up the possibility for a HK violation.
Indeed, if $|s| > |t|$ then $\Psi_A$ is a non-degenerate ground state which is insensitive to the 
value of $t$ as long as $|t|$ is smaller than $|s|$. For $|s|=|t|$ the two states become degenerate and if $|s|<|t|$ then $\Psi_B$ becomes the ground state which is independent of $s$ as long as 
$|s|$ remains smaller than $|t|$. This again clearly demonstrates the absence of unique $v$-representability.
Interestingly, this happens for a non-degenerate ground state and for a density that has neither zero nor full occupancy anywhere. It therefore presents a counterexample to 
a basic assumption in the
work of \citet{KohnPRL} who stated that any non-constant perturbation of a non-degenerate state on a lattice will always change that state.
It also constitutes a counterexample to the statement by \citet{ullrich2002} that ``the density of a non-degenerate ground state uniquely determines [the potential].''

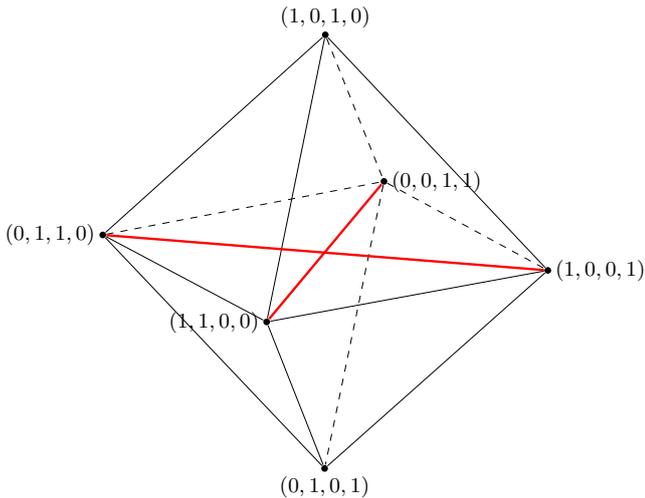
\begin{figure}[ht]
  \centering
  \resizebox{\columnwidth}{!}{%
  \begin{tikzpicture}
    \coordinate (C1) at (3.37,-0.13);
    \coordinate (C2) at (0.07,-3.12);
    \coordinate (C3) at (4.24,-2.32);
    \coordinate (C4) at (6.67,-3.65);
    \coordinate (C5) at (2.50,-4.42);
    \coordinate (C6) at (3.36,-6.61);
    \foreach \m in {1,...,6}
    	\node[circle,fill,inner sep=1.0pt] (N\m) at (C\m) {};
    \draw (N1) -- (N2);
    \draw[dashed] (N1) -- (N3);
    \draw (N1) -- (N4);
    \draw (N1) -- (N5);
    \draw[dashed] (N2) -- (N3);
    \draw[dashed] (N3) -- (N4);
    \draw (N2) -- (N5);
    \draw (N4) -- (N5);
    \draw (N2) -- (N6);
    \draw[dashed] (N3) -- (N6);
    \draw (N5) -- (N6);
    \draw (N4) -- (N6);
    \node [above] at (N1) {$(1,0,1,0)$};
    \node [left] at (N2) {$(0,1,1,0)$};
    \node [right] at (N4) {$(1,0,0,1)$};
    \node [right] at (N3) {$(0,0,1,1)$};
    \node [left] at (N5) {$(1,1,0,0)$};
    \node [below] at (N6) {$(0,1,0,1)$};
    \draw[-,color=red,line width=1pt] (N2) -- (N4);
    \draw[-,color=red,line width=1pt] (N3) -- (N5);
  \end{tikzpicture}
  }
\caption{Density domain $\mathcal{P}_{4,2}$ of the square-graph example with the non-uv densities $\rho_A,\rho_B$ from \eqref{eq:rhoA}-\eqref{eq:rhoB} in red.}
\label{fig:octahedron}
\end{figure}

Note that there is a qualitative difference in the range of potentials that allow for HK violation. The collection of potentials in the triangle graph that keep the ground state
unchanged corresponds to an unbounded domain in potential space since $t \in [0,\infty)$ in our example. On the other hand, in the square graph the corresponding 
domain in potential space is bounded, $t \in [-|s|, |s| \,]$ for a given $s$ in the example. As we will see, this is a general feature that we will prove as Corollary~\ref{cor:pot-bounded} later: 
If $0 < \rho_i <1$ then HK violation can only occur on a bounded domain in potential space. The potentials that lead to HK violations are also `rare', as we will argue in Section~\ref{sec:degen-non-uv}. But first we will establish that also in the density domain, non-uv examples are rare.

\subsection{Still almost all densities are uniquely \texorpdfstring{$v$}{v}-representable}
\label{sec:almost-all-uv-dens}

In Section \ref{sec:v-rep} we will prove that all densities with $0 < \rho_i < 1$ are ensemble $v$-representable.  The question remains whether they are \emph{uniquely} $v$-representable.
This is partially answered by the following theorem.

\begin{theorem}\label{th:non-uv-dens-measure-zero}
The set of non-uv densities is a subset of a finite union of linear subspaces of dimension $<M$ in $\R^M$ and thus has measure zero on the affine set of all physical densities $\mathcal{P}_{M,N}$ with $\sum_i\rho_i=N$.
\end{theorem}

\begin{proof}
First, let us recall the definition of the matrix $\Upsilon[\Psi]$ from Section~\ref{sec:Odlyzko}. We have $\Upsilon[\Psi]_{I,i} = 1$ if $i\in I$ and zero otherwise, where the row index $I \in C[\Psi]$, the set of multi-indices of all non-zero coefficients of $\Psi$ in expansion \eqref{Psi_expansion}.\\
Now, the density $\rho$ for a given (ensemble) state is a finite convex combination of pure states $\Psi_n$ with densities $\rho_n$ of the form
\be
    \rho = \sum_n \lambda_n \rho_n.
\ee
The pure-state density $\rho_n$ is given by
\be\label{eq:rho-from-T}
    \rho_{n,i} =   \sum_{I \ni i} |\Psi_{n,I} |^2 = \sum_{I} T[\Psi_n]_{i,I} |\Psi_{n,I} |^2,
\ee
which we rewrote as the multiplication with matrices $T[\Psi_n]$ where $T[\Psi_n]_{i,I}=1$ if $i\in I \in C[\Psi_n]$ and zero otherwise. However, this means that $T[\Psi_n]$ is exactly the transpose of $\Upsilon[\Psi_n]$, an observation that will become crucial in the course of this proof. If the given density $\rho$ is assumed non-uv then two different potentials $v$ and $v'$ differing by more than a constant not only give rise to the same density $\rho$ but also share the same (ensemble) ground state by the first part of the HK theorem (Theorem~\ref{th:HK1}). 
This means that all $\Psi_n$ in the ensemble are non-uv ground states of both $H(v)$ and $H(v')$ and fulfil \eqref{eq:u-condition} with the same non-zero $u=v-v'-\tfrac{1}{N}(E-E')$, which implies that the intersection $W=\bigcap_n \ker \Upsilon[\Psi_n]$ is non-zero and is a linear subspace of
the kernel of each of the $\Upsilon [\Psi_n]$. 
Since the image of the transpose of a matrix is the orthogonal complement of its kernel
(for a proof for general non-square matrices see \citet{prosalov-book}),
it follows that the image of $T[\Psi_n]$ is orthogonal to $W \subseteq \ker{\Upsilon[\Psi_n]}$ and hence $\rho_n \in W^\perp$ for all $n$ by \eqref{eq:rho-from-T}. It therefore follows that also
$\rho = \sum_n \lambda_n \rho_n \in W^\perp$ where $W^\perp$ has dimension of at most $M-1$
since $W$ has dimension of at least one. 
Every non-uv density matrix belongs to a class characterized by a finite collection of $\Upsilon[\Psi_n]$ matrices specified by the non-uv states in the ensemble and
as there is only a finite number of such collections
there exists only a finite number of subspaces of the type $W^\perp$ that we just discussed.
Finally, we arrive at our conclusion by noticing that the intersection of all those linear subspaces of dimension $<M$ that include non-uv densities $\rho$ with the affine subspace $\{\rho \in \R^M \mid \sum_i \rho_i=N \}$ that contains $\mathcal{P}_{M,N}$ will always lead to a measure-zero set.
%
\end{proof}

We will illustrate this result with the examples from the previous section. For the triangle graph the non-uv state $\Psi_A$ has
\be
    \Upsilon[\Psi_A] = \left( \begin{array}{ccc} 1 & 1 & 0 \\ 0 & 1 & 1 \end{array} \right)
\ee
with rank 2 and we saw in the proof above that the non-uv density that can result from this state must be included in the image of the transpose of $\Upsilon[\Psi_A]$,
\be\label{eq:T-image-1}
    \left( \begin{array}{cc} 1 & 0 \\ 1 & 1 \\ 0 & 1 \end{array} \right) \cdot \left( \begin{array}{c} |\Psi_{(1,2)}|^2 \\ |\Psi_{(2,3)}|^2 \end{array} \right) = \left( \begin{array}{c} \rho_1 \\ 1 \\ 1-\rho_1 \end{array} \right),
\ee
where we used the normalization $|\Psi_{(1,2)}|^2+|\Psi_{(2,3)}|^2=1$ for states that have $\Psi_{(1,3)}=0$.
The resulting density in \eqref{eq:T-image-1} has full occupancy $\rho_2=1$ and is thus on the border of $\mathcal{P}_{3,2}$, a set of measure zero.

This makes it interesting to check the second example of the square graph as well, where the non-uv density did not contain any zero or full occupancy. There,
\be
    \Upsilon[\Psi_A] = \left( \begin{array}{cccc} 1 & 0 & 1 & 0 \\ 1 & 0 & 0 & 1 \\ 0 & 1 & 1 & 0 \\ 0 & 1 & 0 & 1 \end{array} \right)
\ee
with rank 3
and thus for the image, using the
normalization $|\Psi_{(1,3)}|^2+|\Psi_{(1,4)}|^2+|\Psi_{(2,3)}|^2+|\Psi_{(2,4)}|^2 =1$, we find
\be
    \left( \begin{array}{cccc} 1 & 1 & 0 & 0 \\ 0 & 0 & 1 & 1 \\ 1 & 0 & 1 & 0 \\ 0 & 1 & 0 & 1 \end{array} \right) \cdot \left( \begin{array}{c} |\Psi_{(1,3)}|^2 \\ |\Psi_{(1,4)}|^2 \\ |\Psi_{(2,3)}|^2 \\ |\Psi_{(2,4)}|^2 \end{array} \right) = 
\begin{pmatrix*}
\rho_1 \\ 1-\rho_1 \\ \rho_3 \\ 1 -\rho_3 
\end{pmatrix*},
\ee
which represents a 2-dimensional subset of the 3-dimensional affine space $\{\rho\in \R^4 \mid \rho_1+\rho_2 + \rho_3+\rho_4=2 \}$ that does not lie on the border of $\mathcal{P}_{4,2}$, except when $\rho_1$ or $\rho_3 \in \{0,1\}$. The case $\rho_1=1-\rho_3$ with $\rho_3 = \beta (s)$ is just the one that corresponds to the known non-uv density $\rho_A$ in \eqref{eq:rhoA}.

\section{Positive results from the Rellich theorem}
\label{sec:main:Rellich}

\subsection{Openness of the set of non-degenerate, uniquely \texorpdfstring{$v$}{v}-representable ground-state densities}

Rellich's theorem \cite{rellich1937} for the analytic dependency of eigenvalues and eigenvectors on small perturbations is the following (see also the book of \citet[Ch.~1, §1]{rellich-book} and other textbooks \cite{Kato-book,Baumgaertel-book}).

\begin{theorem}[Rellich]\label{th:Rellich}
Let $A_{mn} (\lambda)$, $\lambda \in \R$, be a Hermitian matrix whose components $A_{mn}(\lambda)$ are analytic in $\lambda$ in a neighborhood of $\lambda=0$.
Then there exists a neighborhood of $\lambda=0$ in which the eigenvalues are analytic functions of $\lambda$ and in which we can furthermore choose an orthonormal basis of eigenvectors that are analytic in $\lambda$ as well.
\end{theorem}

The indicated ``choice'' refers to a few circumstances which are related to degeneracy and non-uniqueness.
In the case that the eigenvectors at $\lambda=0$ are degenerate, we cannot freely choose them at $\lambda=0$, but the
requirement of analyticity picks out specific vectors from the unperturbed degenerate manifold (see \citet[Ch.~1, §1]{rellich-book} for an insightful discussion). The correct choice is obtained from the eigenvectors at finite $\lambda$ by letting $\lambda \rightarrow 0$.
If the eigenvalue is non-degenerate, the choice of the eigenvector is unique up to a phase factor, which we can choose to be any analytic function. 
Moreover, if we insist on ordering the states by their eigenvalues then due to possible eigenvalue crossings the eigenvalues as functions of $\lambda$ may display a non-differentiable kink and the eigenstates may jump discontinuously. The latter is illustrated by states \eqref{Psi_A} and \eqref{Psi_B} in our square-graph example, which suddenly switch their role as ground states at $|s|=|t|$ where we take the parameter in Rellich's theorem to be $\lambda=t$ for a fixed value of $s$.
However, since these crossings happen at points with a finite value of $\lambda$, an order-preserving analytic choice of eigenvalues is always possible for sufficiently small $\lambda$.

Let us give an application of Rellich's theorem that serves the topological study of the density-potential mapping. We will establish that the set of uv-ground-state densities coming from a non-degenerate ground state forms an open set. 
The proof corrects and simplifies a proof by \citet{KohnPRL}, who incorrectly assumed that any non-degenerate ground state must be uniquely $v$-representable, a claim refuted in Section~\ref{sec:ex-non-unique} above by explicit example.

\begin{theorem}\label{th:U-open}
The set $\mathcal{U}_{M,N} \subseteq \mathcal{P}_{M,N}$ of uv densities from non-degenerate ground states is an open set, i.e., if $\rho \in \mathcal{U}_{M,N}$ then there is an $\eps > 0$ such that $\rho' \in \mathcal{U}_{M,N}$ whenever $\| \rho'  - \rho \| < \eps$.
\end{theorem}

\begin{proof}
As always, we consider a graph with $M$ vertices. 
Since all potentials are only fixed up to an additive constant, we can always choose $v_M=0$, leaving $v_i$ for $i=1,\ldots,M-1$ as free variables.
Furthermore, since $\rho_M=N- \sum_{i <M} \rho_i$, we only need to consider $\rho_i$ for $i=1, \ldots,M-1$ as well.
We choose a density $\rho$ which belongs to a non-degenerate uv-ground state with eigenvalue $E_0$ of $H(v)$ and apply a small perturbation $\lambda u$ for a given non-zero potential $u$ with $u_M=0$ again, i.e.,
we consider the external potential $v+\lambda u$.
According to Rellich's theorem (Theorem~\ref{th:Rellich}) we can choose the eigenstate $\Psi(\lambda)$ of the perturbed system in an analytic fashion such that it connects smoothly to 
a ground state of the unperturbed system at $\lambda=0$ (see the remarks after Theorem \ref{th:Rellich}).
Since the ground state was presumed to be non-degenerate,
this state is unique up to a phase factor and will yield the prescribed density $\rho$.
Moreover, for small enough $\lambda$
there will be no eigenvalue crossing such that $\Psi (\lambda)$ remains a non-degenerate ground state itself.
Our Hamiltonian is then given by
\be
( H_0 + V + \lambda U ) \Psi (\lambda) = E (\lambda) \Psi (\lambda)
\label{H_lambda}
\ee
and we set $\Psi_0=\Psi(0)$ and $E_0=E(0)$. The matrix $U$ (and similarly $V$) is diagonal in the $\{e_I\}_I$ basis with
diagonal elements $U_I = \sum_{i \in I} u_i$. Now since $\Psi_0$ is a uv-state we cannot have $U\Psi_0=0$ since that would imply that $v$ and $v+u$ represent this state. Consequently, $U\Psi_0 \neq 0$ and by continuity from Rellich's theorem we also have
\be
U \Psi (\lambda) \neq 0 
\label{nonzero}
\ee
for $\lambda$ small enough.
For later reference we further define
\be
\langle U\rangle (\lambda) = \langle \Psi (\lambda), U \Psi (\lambda) \rangle = \sum_{i=1}^{M-1} u_i \rho_i (v+\lambda u),
\label{Udef}
\ee
where $\rho(v+\lambda u)$ is the density corresponding to $\Psi(\lambda)$.
From (\ref{H_lambda}) it follows by differentiation that
\be
[ H - E_0 ] \frac{\d\Psi}{\d\lambda}\Big|_{\lambda=0} = \left[ \left.\frac{\d E}{\d \lambda}\right|_{\lambda=0}  - U\right] \Psi_0
= U^\prime \Psi_0 \neq 0,
\label{dPdl}
\ee
where we defined $H=H(v)=H_0+V$ and furthermore $U'= c\, I -U$ or equivalently $u_i'=c/N-u_i$ with constant $c=\d E/\d\lambda|_{\lambda=0}$. The last inequality in \eqref{dPdl} follows since $\Psi_0$ is 
an uv-state and the only way for $U' \Psi_0$ to vanish is by the same argument as above that $u_i'=0$, which implies that
$u_i=c/N$ is a constant. However, since we put $u_M=0$ this can only be the case of $u=0$
which contradicts our choice of a non-zero $u$.\\
Since the right hand side of \eqref{dPdl} does not vanish, it follows that $\d\Psi/\d\lambda|_{\lambda=0}$ cannot be a ground state of $H$. We now take
the inner product with $\d\Psi/\d\lambda|_{\lambda=0}$ on the first equation of \eqref{dPdl} and then add the complex conjugate of the equation obtained. This yields
\begin{align}
\begin{aligned}
0 &< 2 \left\langle  \frac{\d\Psi}{\d\lambda}\Big|_{\lambda=0}, [ H - E_0 ] \frac{\d\Psi}{\d\lambda}\Big|_{\lambda=0} \right\rangle \\ &= -\left\langle \frac{\d\Psi}{\d\lambda}\Big|_{\lambda=0} , U \Psi_0 \right\rangle -
\left\langle  \Psi_0 , U \frac{\d\Psi}{\d\lambda}\Big|_{\lambda=0} \right\rangle 
\\&\quad + \frac{\d E}{\d \lambda}\Big|_{\lambda=0} \cdot \frac{\d}{\d\lambda} \| \Psi (\lambda) \|^2 \Big|_{\lambda=0} ,
\end{aligned}
\end{align}
where for the inequality we used the Rayleigh--Ritz principle and the fact that
$\d\Psi/\d\lambda|_{\lambda=0}$ is not a ground state of $H$.
Moreover, the last term on the right hand side vanishes since $\| \Psi (\lambda ) \| =1$
and we therefore find, using the definition (\ref{Udef}), that
\be
\frac{\d\langle U\rangle}{\d\lambda} \Big|_{\lambda=0} = \left\langle \frac{\d\Psi}{\d\lambda} \Big|_{\lambda=0} , U \Psi_0 \right\rangle + \left\langle  \Psi_0 , U \frac{\d\Psi}{\d\lambda} \Big|_{\lambda=0} \right\rangle  < 0.
\label{U_pos}
\ee
On the other hand
\be\begin{aligned}
\frac{\d\langle U\rangle}{\d\lambda}\Big|_{\lambda=0} &= \sum_{i=1}^{M-1} u_i \frac{\d\rho_i }{\d\lambda} (v+\lambda u) \Big|_{\lambda=0} \\&=
 \sum_{i,j=1}^{M-1} u_i \frac{\partial \rho_i}{\partial v_j} (v) u_j = \sum_{i,j=1}^{M-1} u_i  J_{ij} u_j,
\end{aligned}\ee
where we defined $J_{ij} = \partial \rho_i/\partial v_j (v)$ to be the Jacobian matrix of partial derivatives of $\rho_i$ with respect to $v_j$ evaluated at $v$. Since
\eqref{U_pos} establishes that this matrix is negative definite it follows
that $\det(J) \neq 0$. The inverse function theorem then guarantees that in a small enough neighbourhood 
of $\rho$ the inverse one-to-one mapping from densities to potentials exists. This means that each density in this neighbourhood is uniquely $v$-representable and belongs to a non-degenerate ground state, which proves the theorem.
\end{proof}

The given examples suggest that the statement of openness could be extended to all uv densities, not only those coming from non-degenerate states. Since we have no strict proof for this at the moment, it should be left standing as a conjecture.

\subsection{Degeneracy generates non-unique \texorpdfstring{$v$}{v}-representability}
\label{sec:degen-non-uv}

In this section we examine (non-)unique \texorpdfstring{$v$}{v}-representability from the side of the potentials. The questions put forward are if potentials that lead to non-uv densities (non-uv potentials) are rare as well and how those potentials are located within the space of all potentials. In this sense this section is complementary to Section~\ref{sec:almost-all-uv-dens}, where almost all densities were proven to be uv. Contrary to the other results in this work, the statements in this section will not have the character of hard theorems, but be more in the style of an explorative outlook.

For this, we start from any potential $v_0$ that results in a uv ground state without any degeneracy, a situation that seems to be by far the most prevalent one, judging from the examples given in Section~\ref{sec:ex-non-unique}. Since the ground state for $v_0$ is uv, it has just the right coefficients in the usual expansion \eqref{Psi_expansion} non-zero to allow for a unique zero solution for \eqref{u_zero}. Now, if we follow an arbitrary path $v(\lambda)\in\R^M$ parameterized by $\lambda$ starting from $v(0)=v_0$, we know by invoking Rellich's theorem (Theorem~\ref{th:Rellich}) that the coefficients of the ground state can be chosen analytic in $\lambda$---as long as no level crossings occur, degeneracy arises, and a previously excited state takes over the role as the new ground state. However, analyticity in $\lambda$ means that the coefficients cannot be zero over any $\lambda$-interval if they have been non-zero before, so the ground state for $v(\lambda)$ will be uv almost everywhere along the path as long as we do not hit any degeneracy region. However, if we do, a new ground state that is not uv can arise where the number of non-zero coefficients is below the Odlyzko condition. Then, along a different path in potential space, these coefficients can (analytically) stay equal to zero---until we hit another degeneracy region. Therefore, sets in potential space that lead to degeneracies form the boundaries of sets of possible non-uv potentials. Now, the work of \citet{ullrich2002} includes the interesting observation that any potential manifold that leads to a $g$-fold degeneracy of the ground state must be of dimensionality $D_g=M-1-\tfrac{1}{2}(g-1)(g+2)$ (where the free additive constant for the potential has already been fixed). Now clearly the lowest possible degeneracy is $g=2$ and gives the highest possible dimension $D_2=M-3$. If the boundary of a set of non-uv potentials has dimension $M-3$ or less then this set itself can only have maximal dimension $D_2+1=M-2$. This tells us that in the space of all possible potentials with a fixed additive constant, which has dimension $M-1$ itself, the set of non-uv potentials actually has measure zero.

This picture of regions of non-uv potentials emanating from potentials that have degenerate ground states will be illustrated with the examples from Section~\ref{sec:ex-non-unique}. In the triangle case we found the single potential $v=0$ to create a two-fold degeneracy of the ground state. Since $g=2$ and $M=3$, we have $D_g=0$, the dimensionality of a point. Then the potential set $\{(0,-t,0)\mid t>0\}$ of dimension $D_g+1=1$ was noted to lead to non-uv densities, by symmetry the same is true for $\{(-t,0,0)\mid t>0\}$ and $\{(0,0,-t)\mid t>0\}$. Hence, in potential space, we have three rays all originating in $v=0$, the point of degeneracy, as non-uv potentials. The corresponding three densities all lie on the border of the density domain $\mathcal{P}_{3,2}$ since they have full occupancy on one vertex. The situation will be discussed in much more detail in Section~\ref{sec:triangle-v-rep}.

The other example was the square graph where we found the condition $|s|=|t|$ for the potential in \eqref{eq:square-h} that leads to degeneracy, again with multiplicity $g=2$. Since now $M=4$ we have $D_g=1$, the dimension of the two lines $s=t$ and $s=-t$ in potential space. For a fixed $s$ the range of potentials connecting two degeneracy points $t=\pm s$, parameterized by $t\in [-|s|,|s|]$, is then non-uv, while the same holds conversely for a fixed $t$. This situation is shown in Figure~\ref{fig:square-v-plot}. Hence, the non-uv region in potential space indeed has dimension $D_g+1=2$ and is spanned by both parameters $s,t$, just the resulting non-uv density is not always the same and switches between $\rho_A$ for $|s|>|t|$ and $\rho_B$ for $|s|<|t|$ while both densities (and many more) are possible ground-state densities in the degeneracy case at $|s|=|t|$. Note that here the $\rho_A$ still depends on $s$ while being insensitive to $t$, and $\rho_B$ reacts the other way around. For $s\to\infty$, $\rho_A\to (0,1,1,0)$, and for $s\to -\infty$, $\rho_A\to (1,0,0,1)$. On the other hand, for $t\to\infty$, $\rho_B\to (0,0,1,1)$, and for $t\to -\infty$, $\rho_B\to (1,1,0,0)$. Therefore, with a large potential on the square graph, one approaches extreme points instead of instantly hitting the border of the density domain like in the triangle case. This situation already found a graphical representation in Figure~\ref{fig:octahedron}, where the $(s,t)$-plane in potential space gets mapped to the middle plane of the octahedron.

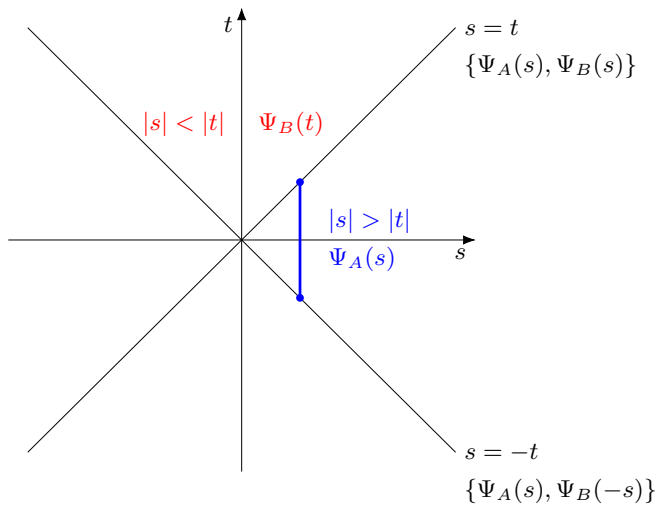
\begin{figure}[ht]
	\centering
	\resizebox{\columnwidth}{!}{%
	\begin{tikzpicture}[every node/.style={scale=2}] 
	    \draw[-{Latex[length=3mm]}] (-6,0) -- (6,0) node[below left] {$s$};
        \draw[-{Latex[length=3mm]}] (0,-6) -- (0,6) node[below left] {$t$};
        \draw[-] (-5.5,-5.5) -- (5.5,5.5) node[right] {$s=t$};
        \draw[-] (-5.5,5.5) -- (5.5,-5.5) node[right] {$s=-t$};
        \draw[pattern=north west lines, pattern color=blue, opacity=0.0, fill opacity=0.2] (0,0) -- (5,5) -- (5,-5) -- (0,0);
        \draw[pattern=north west lines, pattern color=blue, opacity=0.0, fill opacity=0.2] (0,0) -- (-5,5) -- (-5,-5) -- (0,0);
        \draw[pattern=north west lines, pattern color=red, opacity=0.0, fill opacity=0.2] (0,0) -- (5,5) -- (-5,5) -- (0,0);
        \draw[pattern=north west lines, pattern color=red, opacity=0.0, fill opacity=0.2] (0,0) -- (5,-5) -- (-5,-5) -- (0,0);
        \node[right,color=blue] at (2,0.5) {$|s|>|t|$};
        \node[right,color=blue] at (2,-0.5) {$\Psi_A(s)$};
        \node[left,color=red] at (-0.2,3) {$|s|<|t|$};
        \node[right,color=red] at (0.2,3) {$\Psi_B(t)$};
        \node[right] at (5.5,4.5) {$\{\Psi_A(s),\Psi_B(s)\}$};
        \node[right] at (5.5,-6.5) {$\{\Psi_A(s),\Psi_B(-s)\}$};
        \draw[-,color=blue,line width=2pt] (1.5,-1.5) -- (1.5,1.5);
        \node at (1.5,-1.5) [circle,fill,inner sep=1.0pt,color=blue] {};
        \node at (1.5,1.5) [circle,fill,inner sep=1.0pt,color=blue] {};
	\end{tikzpicture}
	}
	\caption{The potential space for the square graph example from \eqref{eq:square-h} in the $(s,t)$-parameter space. Along the diagonals 2-fold degeneracy arises, while the hatched regions do not have degeneracy but are non-uv. A non-uv potential set that gives a fixed density $\rho_A(s)$ is displayed as a blue line that links two degeneracy regions.}
	\label{fig:square-v-plot}
\end{figure}

Following the statement of \citet{ullrich2002} that potentials which create degeneracies are rare, we have argued in this section that also those that lead to non-uv ground states are uncommon. By consistently adapting their argument to the graph setting one should be able to make this statement more rigorous and give a full picture of the topology of the density-potential mapping.

\section{Special results from the Perron--Frobenius theorem}
\label{sec:main:PF}

\subsection{Positivity of the non-interacting ground-state density}
\label{sec:pos-dens}

In Section~\ref{sec:uv} we noted that one cause of HK violation is the possibility of having zero density $\rho_i=0$ on a vertex.
In this section we will show that this situation is excluded for the non-interacting case and small perturbations.
These results are based on the Perron--Frobenius (PF) and Rellich theorems. Nevertheless, having a strictly positive density does not necessarily help with the HK theorem, as it does in the continuum case \cite{Lammert2018}, because still a critical number of wave-function components above the Odlyzko condition might be zero.

We start by explaining the basic setting for the PF theorem. Therein one considers a real symmetric matrix $A_{ij}=A_{ji} \in \mathbb{R}$
with $A_{ij} \geq 0$. To this matrix we can assign a graph $G(A)$ by the adjacency relation $i \sim j$ whenever $A_{ij} \neq 0$ for $i \neq j$ and $i \not\sim j$ otherwise. If the graph is connected then for any choice of two vertices $i$ and $j$ there is a positive integer $n$ such that 
\be
(A^n)_{ij} > 0,
\ee
which graphically means that there is a path of length $n$, i.e., having $n$ edges, connecting the vertices $i$ and $j$.
A matrix with this property is called {\em irreducible}. For such matrices we have the following theorem about the eigenvector with maximal eigenvalue \cite{minc-book,meyer-book,prosalov-book} that has already been employed by \citet{englisch1983vrep-PF} and \citet{CCR1985} for the case of bosonic lattice systems.

\begin{theorem}[Perron--Frobenius]\label{th:PF}
Let a real, non-negative $M \times M$ matrix $A$ be symmetric and irreducible. Then the following statements hold:

\begin{enumerate}
\item $A$ has a non-degenerate, maximal eigenvalue $\lambda_{\mathrm{max}} \in \mathbb{R}$ such that all other eigenvalues are strictly smaller.
\item The eigenspace belonging to $\lambda_{\mathrm{max}}$ is one-dimensional and the corresponding eigenvector 
$x=(x_1, \ldots,x_M)$ can be chosen in such a way that all its components are strictly positive, i.e., $x_j >0$ for $j=1, \ldots, M$.
\item There are no other non-negative eigenvectors corresponding to the lower eigenvalues, i.e., all other eigenvectors must have at least one negative or complex component.
\end{enumerate}
\end{theorem}

This theorem has an immediate application to graph Hamiltonians. We start by considering a real, one-particle Hamiltonian $h_{ij} = h_{ji} \in \mathbb{R}$
corresponding to a connected graph $G(h)$. 
If $h_{ij} \leq 0$ for $i \neq j$ (as it is the case for the negative graph Laplacian) then by a suitable shift of the diagonal of $h$ by adding $-c \,I$ with $c > 0$ we can assure that
$A_{ij} = - h_{ij} + c\, \delta_{ij} \geq 0$ and $A$ thus satisfies the requirements of the PF theorem. The eigenvector with maximal eigenvalue of $A$ then exactly corresponds to the ground-state eigenvector of $h$ that has minimal energy. If we normalize this vector and call it $\phi_0$, we can guarantee $\phi_{0,i} > 0$ and thus have of course also a positive density for the one-particle ground state.
The theorem can, however, not in general be applied to the $N$-particle Hamiltonians such as in the examples of Section~\ref{sec:graph-examples}, because the corresponding Hamiltonians also include $+1$ in the off-diagonal entries.
Yet, in the case of \emph{non-interacting} fermions on the graph $G(h)$ any 
$N$-particle ground state
is of the form
\be
\Psi= \phi_0 \wedge \phi_1 \wedge \ldots \wedge \phi_{N-1},
\ee
where $\phi_j$ for $j=1,\ldots, N-1$ are other orthonormal eigenstates of $h$ chosen according to the aufbau principle. If the highest excited state $\phi_{N-1}$ in $\Psi$ above belongs to a degenerate multiplet then the choice according to the aufbau principle is not unique and we also get degeneracy in the $N$-particle ground state. However, every such choice always contains the non-degenerate, strictly positive state $\phi_0$ and therefore for any pure-state ground-state density $\rho$ we surely have
\be\label{eq:rho-from-Slater}
\rho_i = \phi_{0,i}^2 + | \phi_{1,i} |^2+ \ldots +  | \phi_{N-1,i} |^2 >0.
\ee
We therefore conclude that any (ensemble) ground-state 
density for $N$ non-interacting electrons cannot vanish on any of the vertices as a consequence of the Perron--Frobenius theorem applied on the one-particle Hamiltonian.
We collect our results in the following corollary.

\begin{corollary}\label{cor:PF-one-particle-H}
If $h$ is a real, symmetric one-particle Hamiltonian corresponding to a connected graph $G(h)$ and $h_{ij} \leq 0$ for $i \neq j$ then the ground state $\phi_0$ of $h$
is non-degenerate and can be chosen to be strictly positive on every vertex of $G(h)$, i.e., $\phi_{0,i} >0$ for all $i\in X$. Moreover, for any ground-state density $\rho$ of $N$ non-interacting 
fermions on the graph $G(h)$ we have $\rho_i >0$ for all $i \in X$.
\end{corollary}

Now, Rellich's theorem (Theorem~\ref{th:Rellich}) can be used to slightly extend the result of Corollary~\ref{cor:PF-one-particle-H} to small perturbations of the non-interacting Hamiltonian, or to perturbations of any other Hamiltonian that has a non-zero ground-state density.

\begin{corollary}
Let $H$ be an $N$-particle Hamiltonian that has a non-zero density coming from a non-degenerate ground state. Then for any Hermitian matrix $W$ the perturbed Hamiltonian $H+\lambda W$ will also have a non-zero ground-state density for small enough $\lambda$.
\end{corollary}

\subsection{Unique \texorpdfstring{$v$}{v}-representability for a linear chain with interacting fermions}
\label{sec:lin-chain}

As explained in Section~\ref{sec:Odlyzko}, the validity of the HK theorem can be assured if the ground-state expansion in terms of the states $\{e_I\}_I$ contains a sufficient number of
non-zero coefficients determined by the Odlyzko condition. In general it is a difficult task to ensure that the ground state has the required number of non-zero coefficients. However, if 
we restrict ourselves to a certain graph topology then more detailed results can be obtained. We will do precisely this for the case of a linear chain with a real Hamiltonian.

For a linear chain we can label the vertices $1,\ldots , M$ from left to right along the chain. The chain is then represented by a one-particle Hamiltonian
$h_{ij}$ that for $i\neq j$ is only non-zero when $j = i \pm 1$ (except at the start and the end of the chain where the connection is only in one direction) and we will assume that $h_{ij} <0$ in that case (like it would be for the negative graph Laplacian). We can further allow for any external potential already included in $h$ and an arbitrary interaction $W$, as long as it is diagonal in the $\{e_I\}_I$ basis, so the Hamiltonian that we consider is
\be
H = \sum_j\sum_{i = j \pm 1} h_{ij} \,  \hat{a}^\dagger_i \hat{a}_j + W.
\label{Hamil2}
\ee

\begin{theorem}
Let the graph $G(h)$ of a real, one-particle Hamiltonian $h$ with $h_{ij} \leq 0$ for $i \neq j$ be a linear chain and thus consider the Hamiltonian \eqref{Hamil2} where $W$ is any interaction diagonal in the $\{e_I\}_I$ basis. Then the ground state $\Psi = \sum_I \Psi_I e_I$ is non-degenerate and can be chosen such that 
$\Psi_I >0$ for all $I$. In particular it follows that $\Psi$
is uniquely $v$-representable.
\end{theorem}

\begin{proof}
We first demonstrate that the off-diagonal elements $H_{IJ}$ of $H$ in the $\{e_I\}_I$ basis are non-positive, i.e., $H_{IJ} \leq 0$ if $I \neq J$.
Since $W$ by our choice only has diagonal elements we only need to consider the one-particle part of the Hamiltonian.
We therefore act with $h_{ij} \hat{a}^\dagger_i \hat{a}_j $ on a basis vector $e_I$. If $j$ does not occur in the multi-index $I$
or when $i \neq j\pm 1$
then clearly $h_{ij} \hat{a}^\dagger_i \hat{a}_j e_I=0$, so let us assume $j \in I$ and $i = j \pm 1$. We can then write $I = ( \ldots, i', j, i'', \ldots )$ with $i'<j<i''$
and we find $h_{ij} \hat{a}^\dagger_i \hat{a}_j e_I= h_{ij} e_J$ where $J= ( \ldots, i', j\pm1, i'', \ldots )$. Now either $i''=j+1$ or $i'=j-1$, in which case $e_J=0$,
or $i' < j\pm1 <i''$, in which case $e_J$ is one of the ordered basis vectors and gets a negative prefactor $h_{ij} <0$ for $i=j \pm1$. This proves our assertion.\\
Since we have established that $H_{IJ} \leq 0$ for $I \neq J$ and furthermore the corresponding fermionic graph is connected (Lemma~\ref{lem:ferm-graph-connected}), the PF theorem (Theorem~\ref{th:PF}) applies.
We can thus conclude that the many-particle ground state $\Psi$ expanded in the $\{e_I\}_I$ basis is non-degenerate and can be chosen such that all coefficients have $\Psi_I >0$.
By Corollary~\ref{cor:HK-necessary} this state is then uniquely $v$-representable.
\end{proof}

\section{Constrained-search functionals and \texorpdfstring{$v$}{v}-representability}
\label{sec:main:CS-v-rep}

\subsection{General notions}
\label{sec:CS-functionals}

The notion of $v$-representability is clearly a central concept to DFT: Determine the set of all densities for which a potential exists, such that the respective ground state yields just this density. It is therefore the precondition for a well-defined density-potential mapping. This can be generalized into allowing not only pure ground states but also a ground-state ensemble and thus a mixture of degenerate ground states, which yields the respective density. This describes ``ensemble $v$-representability'' in comparison to ``pure-state $v$-representability''.
A necessary condition for densities to be ensemble $v$-representable on an even infinitely large lattice was given by \citet{CCR1985}. In Theorem~\ref{th:v-rep} below we give a simplified proof for the graph case. The sets of densities under consideration are the most general ones: First, the physical densities $\mathcal{P}_{M,N}$ already defined in \eqref{eq:def-P_N} where only $0\leq\rho_i\leq 1$ and the proper normalization are taken into account and which are also all $N$-representable, as demonstrated in Section~\ref{sec:N-rep}, and second, the open set $\mathcal{P}_{M,N}^+$ with strict inequalities.
\begin{align}
	&\mathcal{P}_{M,N} = \left\{ \rho : X \to \R \;\middle|\; 0 \leq \rho_i \leq 1, \sum_{i=1}^M\rho_i = N \right\}\\
	&\mathcal{P}_{M,N}^+ = \left\{ \rho : X \to \R \;\middle|\; 0 < \rho_i < 1, \sum_{i=1}^M\rho_i = N \right\}
\end{align}
Note that both sets are convex and that the set of physical densities $\mathcal{P}_{M,N}$ is the closure of the open set $\mathcal{P}_{M,N}^+$.

In the following we fix $H_0$ as the internal part of the Hamiltonian including possible interactions, while the external potential $v \in \R^M$ couples as usual to the density and acts as an operator $V$ on $\H_N$.
From the Rayleigh--Ritz variational principle we get the ground state or ground-state ensemble for a given potential $v$ as the minimizers of the respective variational energy expressions
\begin{align}
    E(v) &= \inf_{\Psi\in\mathcal{I}_N} \left\{ \langle \Psi, (H_0+V)\Psi \rangle \right\} \\
    &= \inf_{\Psi\in\mathcal{I}_N} \left\{ \langle \Psi, H_0\Psi \rangle + \sum_i v_i\rho[\Psi]_i \right\} \quad\text{and} \nonumber\\
    E(v) &= \inf_{\Gamma\in\mathcal{D}_N} \left\{ \Tr (\Gamma (H_0+V)) \right\} \\
    &= \inf_{\Gamma\in\mathcal{D}_N} \left\{ \Tr (\Gamma H_0) + \sum_i v_i\rho[\Gamma]_i \right\}, \nonumber
\end{align}
where the variation goes over all normalized wave functions or density matrices. The variation can be split into two parts: First vary over all states yielding a fixed density $\rho$, denoted by $\Psi\mapsto\rho$ and $\Gamma\mapsto\rho$, and then over all possible physical densities in $\mathcal{P}_{M,N}$. We then have
\begin{equation}
    \label{eq:E-inf-inf}
    E(v) = \inf_{\rho\in\mathcal{P}_{M,N}} \left\{ \inf_{\Psi\in \mathcal{I}_N} \{ \langle \Psi, H_0\Psi \rangle \mid \Psi \mapsto \rho \} + \sum_i v_i\rho_i \right\}
\end{equation}
and
\begin{equation}
    \label{eq:E-inf-inf-Gamma}
    E(v) = \inf_{\rho\in\mathcal{P}_{M,N}} \left\{ \inf_{\Gamma\in \mathcal{D}_N} \{ \Tr (\Gamma H_0) \mid \Gamma \mapsto \rho \} + \sum_i v_i\rho_i \right\}.
\end{equation}
This variational form of the ground-state energy allows for an alternative and very direct proof of the first part of the HK theorem, stating that two Hamiltonians that differ only in their external potentials and share a common ground-state density $\rho$ also share a ground state with this density.

\begin{proof}(Alternative proof for Theorem~\ref{th:HK1}
\cite{lammert-workshop})
Let the common ground-state density be fixed as $\rho$, then the outer infimum in \eqref{eq:E-inf-inf} is void and we have
\begin{equation}
    E(v) = \inf_{\Psi\in \mathcal{I}_N} \{ \langle \Psi, H_0\Psi \rangle \mid \Psi \mapsto \rho \} + \sum_i v_i\rho_i.
\end{equation}
However, the remaining infimum is entirely independent of $v$ and will have the same value for all potentials. Consequently, the given density alone determines the ground state. Obviously, the proof is the same considering ensemble states. 
\end{proof}

Realizing in the proof above that the universal, $v$-independent part in the expression for $E(v)$ fulfills an important function, the following definitions of the so-called constrained-search density functionals \cite{Levy79,Lieb1983} on an extended domain that is the full vector space $\mathbb{R}^M$ arise.
\begin{align}\label{eq:tilde-F-def}
	&\tilde F(\rho) = \left\{ \begin{array}{ll}
		\inf_{\Psi\in \mathcal{I}_N} \{ \langle \Psi, H_0\Psi \rangle \mid \Psi \mapsto \rho \} \;\; &\rho \in \mathcal{P}_{M,N}  \\
		+\infty &\rho \in \R^M \setminus \mathcal{P}_{M,N}
	\end{array}\right. \\
	\label{eq:F-def}
	&F(\rho) = \left\{ \begin{array}{ll}
		\inf_{\Gamma\in \mathcal{D}_N} \{ \Tr (\Gamma H_0) \mid \Gamma \mapsto \rho \} \;\;\;\, &\rho \in \mathcal{P}_{M,N}  \\
		+\infty &\rho \in \R^M \setminus \mathcal{P}_{M,N}
	\end{array}\right.
\end{align}
From the definition it is evident that $\tilde F(\rho) \geq F(\rho)$ for all $\rho$ because of the larger search space of the second.
Next, observe that the ground-state energy functional $E(v)$ in \eqref{eq:E-inf-inf-Gamma} (or \eqref{eq:E-inf-inf}) is then given as the Legendre--Fenchel transform or convex conjugate \cite[§12]{rockafellar-book} of $F(\rho)$ (or equally of $\tilde F(\rho)$), where we use a slightly different convention than usual that is specifically adapted to DFT \cite{Kvaal2014},
\begin{equation}\label{eq:E-LF}
	E(v) = \inf_{\rho\in\R^M} \left\{ F(\rho) + \sum_i v_i\rho_i \right\}.
\end{equation}
With another Legendre--Fenchel transformation we can transform $E(v)$ back to $F(\rho)$,
\begin{equation}
	F(\rho) = \sup_{v\in\R^M} \left\{ E(v) - \sum_i v_i\rho_i \right\}, \quad \rho \in \R^M.
	\label{eq:LF}
\end{equation}
For this to be possible, the functional $F$ has to fulfil several properties \cite[Cor.~12.2.1]{rockafellar-book} that do in fact hold: $F$ is proper because it is bounded below, closed because it is lower semi-continuous, and convex from the linearity of the trace in \eqref{eq:F-def} (see \citet{Lieb1983} for the last two properties). Since $\tilde F$ in general fails to be convex, as we will explicitly demonstrate later in Section~\ref{sec:cuboctahedron}, it is only $F$ that we get back from the described double Legendre--Fenchel transformation starting from either $\tilde F$ or $F$.
\be
\begin{aligned}
\begin{tikzpicture}
  \matrix (m)
    [
      matrix of math nodes,
      row sep    = 0.2em,
      column sep = 2em
    ]
    {
      \tilde F & & \\
      & E & F \\
      F & & \\
    };
  \path
    (m-1-1) edge [->] node [above] {\scriptsize LF} (m-2-2)
    (m-3-1) edge [->] node [below] {\scriptsize LF} (m-2-2)
    (m-2-2) edge [->] node [above] {\scriptsize LF} (m-2-3);
\end{tikzpicture}
\end{aligned}
\ee
However, this means that $F$ is the convex hull of $\tilde F$.

\begin{proposition}\label{prop:ch}
$F = \ch\, \tilde F$.
\end{proposition}

\begin{proof}
Any density matrix $\Gamma\mapsto\rho$ can be written as a convex combination $\Gamma = \sum_n \lambda_n\Gamma_n$, where the $\Gamma_n$ correspond to pure states. Thus from the definition,
\begin{widetext}
\begin{equation}
\begin{aligned}
		F(\rho) &=
		\min_{\Gamma \in \mathcal{D}_N} \left\{ \sum_n\lambda_n\Tr (\Gamma_n H_0) \;\middle|\; \sum_n\lambda_n\Gamma_n \mapsto \rho, \lambda_n\in[0,1], \sum_n\lambda_n=1, \Gamma_n\;\text{pure} \right\}  \\
		&= \min_{\rho \in \mathcal{P}_{M,N}} \left\{ \sum_n\lambda_n \min_{\Gamma_n \in \mathcal{D}_N} \left\{ \Tr (\Gamma_n H_0) \;\middle|\; \Gamma_n \mapsto \rho_n, \Gamma_n\;\text{pure} \right\} \;\middle|\;  \sum_n\lambda_n\rho_n \mapsto \rho, \lambda_n\in[0,1], \sum_n\lambda_n=1 \right\} \\
		&= \min_{\rho \in \mathcal{P}_{M,N}} \left\{ \sum_n\lambda_n \tilde F(\rho_n) \;\middle|\;  \sum_n\lambda_n\rho_n \mapsto \rho, \lambda_n\in[0,1], \sum_n\lambda_n=1 \right\}.
\end{aligned}
\end{equation}
\end{widetext}
This is exactly the convex hull of $\tilde F$.
\end{proof}

The convexity of $F$ restricted to $\mathcal{P}_{M,N}$ has a notable consequence, since such functionals are automatically differentiable almost everywhere \cite[Th.~25.5]{rockafellar-book}. Furthermore, differentiability of $F$ at a density point $\rho\in\mathcal{P}_{M,N}^+$ means that $F$ has a unique subgradient there \cite[Th.~25.1]{rockafellar-book}, $v=-\nabla F(\rho)$ (modulo a constant), which exactly gives the (unique) minimum in \eqref{eq:E-LF}, so $\rho$ is found to be uv. Consequently, almost all densities are uv, exactly the result we derived in Theorem~\ref{th:non-uv-dens-measure-zero}. The possible non-differentiability of $F$ in continuum DFT has been previously discussed by \citet{Lammert2007}.
The proof of $v$-representability in \citet{englisch2} for finite-dimensional state spaces rests on finding continuous tangent functionals to $F$, a notion closely related to the subgradient, but we will instead present the approach of \citet{CCR1985} in the next section that avoids any reference to differentiability.

\subsection{Results on \texorpdfstring{$v$}{v}-representability}
\label{sec:v-rep}

$E(v)$ in the supremum of \eqref{eq:LF} is continuous in $v$ as a consequence of Rellich's theorem (Theorem~\ref{th:Rellich}) or simply because every finite concave function is continuous. However, the search space $\R^M$ is not compact, so at this point we cannot be sure that the supremum is always attained for $\rho\in\mathcal{P}_{M,N}$. However, if we are able to show that the search space can equivalently be replaced by a compact set then any choice of $\rho$ leads to a corresponding potential $v$ and we just need to make sure that this $\rho$ also comes from a ground state of $H_0+V$ to have $v$-representability. This is the content of the following theorem adapted from \citet{CCR1985} that guarantees $v$-representability by ensembles for all $\rho \in \mathcal{P}_{M,N}^+$, yet not for the whole $\mathcal{P}_{M,N}$.

\begin{theorem}[ensemble $v$-representability]\label{th:v-rep}
For all $\rho \in \mathcal{P}_{M,N}^+$ there is a $v \in \R^M$ such that the Hamiltonian $H_0+V$ has a ground-state ensemble $\Gamma \mapsto \rho$.
\end{theorem}

\begin{proof}
In order to prove ensemble $v$-representability we will show that for all $\rho \in \mathcal{P}_{M,N}^+$ the supremum in \eqref{eq:LF} is actually a maximum. This means that there is $v \in \R^M$ such that $F(\rho) = E(v) - \sum_i v_i\rho_i$. However, rearranging this to $E(v) = F(\rho) + \sum_i v_i\rho_i$ it identically means that $\rho$ is a minimizer in \eqref{eq:E-LF} for the determined $v$. However, since in \eqref{eq:F-def} a continuous function is varied over a compact set, for every $\rho\in\mathcal{P}_{M,N}$ there is a $\Gamma \in \mathcal{D}_N$ that minimizes \eqref{eq:F-def} and this density matrix is then the required ground-state ensemble.\\
To start with, we note that by adding a constant shift to a potential $v$, the associated ground-state density does not change while $E(v)$ is shifted. Therefore, without loss of generality, we always limit the $v$ under consideration to those that give $E(v)=0$.\\
Considering now an arbitrary $\rho \in \mathcal{P}_{M,N}^+$, the supremum in \eqref{eq:LF} means that for every $\eps>0$ we can find a $v\in \R^M$ with $E(v)=0$ such that
\begin{equation}
F(\rho)-\eps \leq - \sum_i v_i\rho_i.
\end{equation}
This estimate leads to 
\begin{equation}
\begin{aligned}
\label{eq:v-estimate-1}
\sum_i v_i\rho_i &= \sum_{\substack{i \\ v_i > 0 }} v_i\rho_i -\sum_{\substack{i \\ v_i < 0 }} |v_i|\rho_i \\&= D^+ - D^- \leq -F(\rho)+\eps \leq \|H_0\|+\eps,
\end{aligned}
\end{equation}
where we separated the positive and negative contributions from the energy of the external potential, $D^+$ and $D^-$, and use the operator norm of $H_0$ to bound $|F(\rho)|$ from above.
On the other hand, from \eqref{eq:E-LF} and \eqref{eq:F-def} we then have for all $\tilde\rho \in \mathcal{P}_{M,N}$ and $v\in\R^M$ with $E(v)=0$ that
\begin{equation}\label{eq:v-rho-inequality}
	-\sum_i v_i\tilde\rho_i \leq F(\tilde\rho) \leq \|H_0\|.
\end{equation}
We now choose a density $\tilde\rho$ that has the value $p\rho_i$ for all $i\in X$ where $v_i \geq 0$ and $q\rho_i$ otherwise. Here $p,q$ are chosen with $0<p<1<q$ such that $\tilde\rho$ is still normalized to $N$ as required. This is clearly possible since $\rho_i<1$ everywhere and so there is room to raise the density by a factor $q>1$ where $v_i<0$ while lowering it with $p<1$ where $v_i>0$ and still having $q\rho_i \leq 1$ and keeping it normalized. The extreme cases where the potential is purely positive or negative will be considered afterwards. Now \eqref{eq:v-rho-inequality} is
\begin{equation}\label{eq:v-estimate-2}
	-pD^+ + qD^- \leq F(\tilde\rho) \leq \|H_0\|
\end{equation}
and the combination of \eqref{eq:v-estimate-1} and \eqref{eq:v-estimate-2} gives
\begin{align}
&(q-p) D^+ \leq (1+q)\|H_0\| + q\eps,\\
&(q-p) D^- \leq (1+p)\|H_0\| + p\eps,\quad\text{and finally}\\
&D^+ + D^- = \sum_i |v_i|\rho_i \leq \frac{p+q+2}{q-p}\|H_0\| + \frac{p+q}{q-p}\eps.\label{eq:v-estimate-3}
\end{align}
In the case $v_i\geq 0$ everywhere we simply have $D^- = 0$ and thus the stronger estimate $D^+ + D^- = D^+ - D^- \leq \|H_0\|+\eps$ from \eqref{eq:v-estimate-1} holds. If $v_i<0$ everywhere then $D^+=0$ and we get $D^+ + D^- = -D^+ + D^- \leq \|H_0\|$ directly from \eqref{eq:v-rho-inequality} by just setting $\tilde\rho=\rho$. This means that in any case we can rely on the estimate \eqref{eq:v-estimate-3} when looking for potentials that yield a ground-state density $\rho$. Since we know $\min_i\rho_i > 0$ we even have the 1-norm bound
\begin{align}
\|v\|_1&=\sum_i |v_i| \leq (\min_i\rho_i)^{-1} \left( \frac{p+q+2}{q-p}\|H_0\| + \frac{p+q}{q-p}\eps \right) \nonumber\\&= R_\rho \label{eq:v-bound}
\end{align}
for the potential, where the r.h.s.\ just depends on the chosen $\rho$ since this also determines the parameters $q$ and $p$.
Overall, we have concluded that in order to find the supremum in \eqref{eq:LF} we can limit the search to potentials inside the closed 1-norm ball with finite radius $R_\rho$, which is clearly compact. Since $E$ is known to be continuous, the extreme value theorem yields a definitive maximum for \eqref{eq:LF} inside this ball.
\end{proof}

Such ensemble $v$-representability was also demonstrated for a coarse-grained version of continuum DFT \cite{lammert2006coarse,lammert2010well}.
On the other hand, it has been stated recently by \citet{RoesslerVerdozzi2018} that there are densities in $\mathcal{P}_{M,N}^+$ on simple lattice systems that are not ensemble $v$-representable, which would be in opposition to the theorem above. However, the claim rests on a numerical reverse-engineering procedure for finding a suitable potential $v$ where uniqueness may have been assumed. Since their plot for the deviation from the reference density shows a kink at its minimum, this rather points to non-uniqueness than non-$v$-representability, such as in our HK-violation examples (Section~\ref{sec:ex-non-unique}).

The proof above suggests a simple corollary that limits the set of possible potentials yielding a given density.

\begin{corollary}\label{cor:pot-bounded}
The set of potentials (modulo a constant) that leads to any given $\rho \in \mathcal{P}_{M,N}^+$ as the ground-state density is a bounded set, as can be directly seen from estimate \eqref{eq:v-bound}.
\end{corollary}

The above result tells us that ensemble representability in $\mathcal{P}_{M,N}^+$ is not only possible but even realized within bounded potential sets determined by the given density.
This leaves open the question of \emph{pure-state} $v$-representability. The following proposition connects this problem with the value of the two constrained-search functionals.

\begin{proposition}\label{prop:pure-state-v-rep}
A density $\rho \in \mathcal{P}_{M,N}^+$ is pure-state $v$-representable if and only if $F(\rho)=\tilde F(\rho)$.
\end{proposition}

\begin{proof}
We first show that from pure-state $v$-representability it follows $F(\rho)=\tilde F(\rho)$.
The overall estimate $F \leq \tilde F$ is clear from the larger search space for the infimum in $F$. Hence, we just need to show $\tilde F(\rho) \leq F(\rho)$ if $\rho$ is pure-state $v$-representable. In such a case there is a $\Psi \in \H_N$ that is a ground state of $H_0+V$ and has $\Psi \mapsto \rho$. Now take any $\Gamma \mapsto \rho$ to be $\Gamma = \sum_n \lambda_n \Gamma_n$ with $\sum_n\lambda_n=1$, $\lambda_n \in [0,1]$, and $\Gamma_n$ corresponding to a pure state $\Psi_n$. We also have from \eqref{eq:tilde-F-def} that $\tilde F(\rho) = \langle \Psi,H_0\Psi \rangle$ and so with the Rayleigh--Ritz variation principle
\begin{align}
	E(v)&=\tilde F(\rho) + \sum_i v_i \rho_i = \sum_n\lambda_n \langle \Psi,(H_0+V)\Psi \rangle \nonumber \\
	&\leq \sum_n\lambda_n \langle \Psi_n,(H_0+V)\Psi_n \rangle \\
	&= \sum_n\lambda_n \Tr(\Gamma_n (H_0+V)) = \Tr(\Gamma H_0) + \sum_i v_i \rho_i. \nonumber
\end{align}
This means we derived $\tilde F(\rho) \leq \Tr(\Gamma H_0)$ for arbitrary density matrices that have density $\rho$. Taking the infimum over all these will result in $\tilde F(\rho) \leq F(\rho)$ and finishes the first part of the proof.\\
For the other implication, assume $F(\rho)=\tilde F(\rho)$. However, this just means that instead of a density matrix one can choose a pure state to achieve the same minimum, thus $\rho$ is $v$-representable by this pure state.
\end{proof}

By finding a single counterexample for pure-state $v$-representability in $\mathcal{P}_{M,N}^+$ we would thus know that $\tilde F \neq F$ and consequently, since $F$ is the convex hull of $\tilde F$, the pure-state constrained-search function $\tilde F$ cannot be convex. The argument for non-pure-state $v$-representability of \citet[Sec.~V]{Levy1982} depends on at least 3-fold degeneracy and finding arbitrarily many equations that must be simultaneously fulfilled (by checking different positions in the continuum) and thus does not work on a graph. \citet[Th.\ 3.4(i)]{Lieb1983}, on the other hand, gave a group theoretical argument that also relies on at least 3-fold degeneracy but is invalid in the discrete setting discussed here as we will show in Section~\ref{sec:complete-graph}.
To reconcile the results of non-pure-state $v$-representability with the graph setting, we give an explicit counterexample in Section~\ref{sec:cuboctahedron}. However, before that we discuss an example where actually every density is pure-state $v$-representable and thus $\tilde F = F$.

\subsection{Triangle graph: Full \texorpdfstring{$v$}{v}-representability and explicit constrained-search functional}
\label{sec:triangle-v-rep}

The two-particle Hamiltonian for the triangle graph with a general potential $v$ was already given in \eqref{eq:H-triangle} together with the 2-fold degenerate ground-state space for $v=0$ spanned by the two orthonormal states
\be
 \Psi_A = \frac{1}{\sqrt{2}} (1,0,-1) \quad \text{and} \quad \Psi_B =\frac{1}{\sqrt{6}} (1,2,1).
\ee
in the $\H_2$-basis $\{e_1\wedge e_2,e_1 \wedge e_3,e_2\wedge e_3\}$ with densities $\rho_A=\frac{1}{2}(1,2,1)$ and $\rho_B=\frac{1}{6}(5,2,5)$.
This multiplicity disappears as soon as a non-constant potential is applied, symmetry is broken, and the degeneracy is lifted. That the Hamiltonian with a non-constant $v$ actually has no multiple eigenvalues can be checked by showing that the discriminant of the characteristic polynomial of $H(v)$ is non-zero for all $v$, which can be done explicitly in this case with some effort. This means that the only densities that can come from ensemble states are the ground-state densities for $v=0$, collected in the set $\mathcal{C}$, all coming from the states $\Psi_A,\Psi_B$ above. Hence, in order to test the possibility of full pure-state $v$-representability, just the mixtures of $\Psi_A$ and $\Psi_B$ have to be scrutinized. However, any convex combination $\lambda_A\rho_A+\lambda_B\rho_B \in \mathcal{C}$ is also simply the density of the linear combination $\sqrt{\lambda_A}\Psi_A+\i\sqrt{\lambda_B}\Psi_B$ with the real $\Psi_A,\Psi_B$, so the two real dimensions of the complex plane that correspond to the double degeneracy save the day. It is thus straightforwardly shown that \emph{any} density in $\mathcal{P}_{3,2}^+$ is also pure-state $v$-representable and consequently $\tilde F=F$ on the triangle graph by Proposition~\ref{prop:pure-state-v-rep}.

Since any convex combination of two densities from the set $\mathcal{C}$ of ensemble ground-state densities for $v=0$ is again in $\mathcal{C}$, this set is convex. Furthermore, it is closed, because it is the image of a compact set $\{c_A\Psi_A+c_B\Psi_B \mid |c_A|^2+|c_B|^2=1\}$ under the continuous mapping $\Psi \mapsto \rho$.
To get a graphical representation of the density-potential map we observe that the set $\mathcal{P}_{3,2}$ forms a triangle itself: The extreme points are the densities $(1,1,0)$, $(1,0,1)$, and $(0,1,1)$, and all other elements of $\mathcal{P}_{3,2}$ are convex combinations of these three points and so it is natural to use barycentric coordinates for their representation in Figure~\ref{fig:triangle-mapping}. The uniform density $\bar\rho=\tfrac{2}{3}(1,1,1)=\tfrac{1}{2}\rho_A + \tfrac{1}{2}\rho_B$ coming from $1/\sqrt{2}\Psi_A + \i/\sqrt{2}\Psi_B$, which is clearly a ground-state solution to $v=0$ already due to symmetry, forms the center. The possible linear combinations $\Psi=c_A\Psi_A+c_B\Psi_B$ with $|c_A|^2+|c_B|^2=1$ that form the degenerate ground-state manifold for $v=0$ all fulfil a certain density constraint: Take without loss of generality $c_A=\alpha$ and $c_B=\beta \e^{\i\varphi}$ with $\alpha,\beta,\varphi\in\R$ and $\alpha^2+\beta^2=1$, then the density of $\Psi$ is
\be\begin{aligned}
 \rho = \frac{1}{6}(& 3\alpha^2+5\beta^2+2\sqrt{3}\alpha\beta\cos\varphi,\\[-0.2em] & 6\alpha^2+2\beta^2,\\ & 3\alpha^2+5\beta^2-2\sqrt{3}\alpha\beta\cos\varphi).
\end{aligned}\ee
We form the Euclidean distance to the center $\bar\rho$ and use $\alpha^2+\beta^2=1$ to get
\begin{align}
\|\rho-\bar\rho\| &= \left( \sum_i (\rho_i-\bar\rho_i)^2 \right)^{\frac{1}{2}} \nonumber\\&= \left( \frac{1}{6}(1 - 2\alpha^2 + 2\alpha^4 + 2\alpha^2(1-\alpha^2)\cos(2\varphi)) \right)^{\frac{1}{2}} \nonumber\\&\leq \frac{1}{\sqrt{6}}
\label{eq:incircle-condition}
\end{align}
if the cosine is estimated by 1. Independent of $\alpha,\beta$ the choice of $\varphi$ that maximizes the above expression will always yield $1/\sqrt{6}$. This means all the possible $\rho \in \mathcal{C}$ from the degenerate ground states for $v=0$ form a closed circular region of radius $1/\sqrt{6}$ around the center $\bar\rho$. However, $1/\sqrt{6}$ is precisely the incircle radius of an equilateral triangle of side length $\|(1,1,0)-(1,0,1)\|=\sqrt{2}$, so the set $\mathcal{C}$ will be the incircle region of $\mathcal{P}_{3,2}$, touching the border of the triangle at three points. All this is displayed in Figure~\ref{fig:triangle-mapping} correspondingly, together with some details on the mapping back to potentials that will be discussed next.

When plotting the associated potentials we will choose the gauge condition $\sum_i v_i=0$ which can always be achieved by subtracting the constant $\tfrac{1}{M}\sum_i v_i$ from a given potential.
Already in Sections~\ref{sec:ex-non-unique} we noted that by applying an attractive potential from $\{(0,-t,0)\mid t>0\}$, or equivalently from $\{(t,-2t,t)\mid t>0\}$, the ground state $\Psi_A$ remains unchanged (because it has full density $\rho_{A,2}=1$ at vertex 2) while degeneracy is lifted and thus all those potentials yield the same exceptional ground-state density $\rho_A=(\tfrac{1}{2},1,\tfrac{1}{2})$ that formed the first counterexample to the HK theorem. By symmetry, the same is also true for the potentials $(-2t,t,t)$ and $(t,t,-2t)$, just with permuted exceptional densities $(1,\tfrac{1}{2},\tfrac{1}{2})$ and $(\tfrac{1}{2},\tfrac{1}{2},1)$. They occur exactly at those three points where the incircle touches the border of the equilateral density triangle, consequently they are in $\mathcal{P}_{3,2}$ but not in $\mathcal{P}_{3,2}^+$ (still they are $v$-representable, even by an infinite number of different potentials). These considerations display the topological richness of the density-potential mapping for this simplest, non-trivial example of a fermionic graph: While the whole incircle region $\mathcal{C}$ including its boundary with the exception of the three exceptional points is mapped many-to-one to the unique potential $v=0$ (as always modulo an additive constant) the three exceptional points are mapped one-to-many to three rays extending straight into infinity and dividing the whole two-dimensional potential space (after removing the additive constant) into three separate regions. The three remaining open regions in the spikes of the density triangle $\mathcal{S}_1,\mathcal{S}_2,$ and $\mathcal{S}_3$ were already shown to come from non-degenerate states that arise from non-constant potentials, and also the density inside the spikes (not on the boundary) does not have zero or full occupancy at any vertex. One can thus argue in the triangle case that none of the coefficients of the corresponding wave function is zero, so by the Odlyzko condition (Corollary~\ref{cor:HK-necessary}) they are all uv-densities and therefore belong to a unique potential. Those densities are consequently mapped one-to-one to the potential space and form the open set $\mathcal{U}_{3,2} = \mathcal{S}_1\cup\mathcal{S}_2\cup\mathcal{S}_3$ of non-degenerate, uniquely $v$-representable ground-state densities from Theorem~\ref{th:U-open}. The remaining border of the density triangle, with the exception of the exceptional points, cannot be reached by any potential but only be approximated by very large potentials. The whole situation is summarized in Figure~\ref{fig:triangle-mapping}.

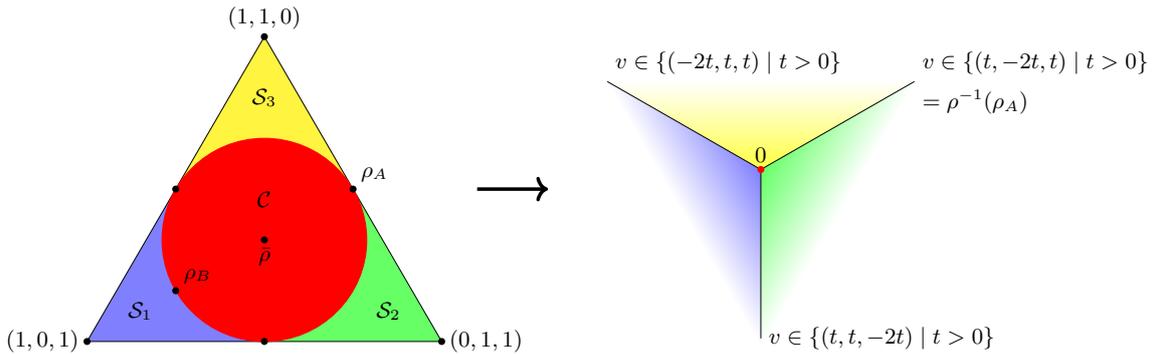
\begin{figure*}[ht]
	\centering
	\resizebox{.85\textwidth}{!}{%
	\begin{tikzpicture}
	\def\ORX{0}; 
	\def\ORY{-1};
	\def\L{5}; 
    \coordinate (R1) at (\ORX,{\ORY+\L*sqrt(3)/2-\L/2/sqrt(3)}); 
    \coordinate (R2) at ({\ORX-\L/2},{\ORY-\L/2/sqrt(3)}); 
    \coordinate (R3) at ({\ORX+\L/2},{\ORY-\L/2/sqrt(3)}); 
	\coordinate (RI) at ({\ORX-\L/4},{\ORY+\L*sqrt(3)/2/2-\L/2/sqrt(3)});
	\coordinate (RII) at ({\ORX+\L/4},{\ORY+\L*sqrt(3)/2/2-\L/2/sqrt(3)});
	\coordinate (RIII) at (\ORX,{\ORY-\L/2/sqrt(3)});
	\coordinate (RB) at ({\ORX-\L/4},{\ORY-\L/4/sqrt(3)});
    \path[fill=yellow,opacity=0.75] (RI) -- (R1) -- (RII) -- (RI);
    \path[fill=green,opacity=0.6] (RII) -- (R3) -- (RIII) -- (RII);
    \path[fill=blue,opacity=0.5] (RIII) -- (R2) -- (RI) -- (RIII);
    \draw (R1) -- (R2) -- (R3) -- (R1);
    \node at (R1) [circle,fill,inner sep=1.0pt] {};
    \node at (R2) [circle,fill,inner sep=1.0pt] {};
    \node at (R3) [circle,fill,inner sep=1.0pt] {};
    \node[above] at (R1) {$(1,1,0)$};
    \node[left] at (R2) {$(1,0,1)$};
    \node[right] at (R3) {$(0,1,1)$};
	\draw [red,fill=red] (\ORX,\ORY) circle ({\L/2/sqrt(3)});
	\node [above left] at (RI) {};
	\node at (RI) [circle,fill,inner sep=1.0pt] {};
	\node [above right] at (RII) {$\rho_{A}$};
	\node at (RII) [circle,fill,inner sep=1.0pt] {};
	\node [below] at (RIII) {};
	\node at (RIII) [circle,fill,inner sep=1.0pt] {};
	\node [below] at (\ORX,\ORY) {$\bar\rho$};
	\node at (\ORX,\ORY) [circle,fill,inner sep=1.0pt] {};
	\node [above right] at (RB) {$\rho_B$};
	\node at (RB) [circle,fill,inner sep=1.0pt] {};
	\node at ({(\ORX-\L/2)*0.7+\ORX*0.3},{(\ORY-\L/2/sqrt(3))*0.7+\ORY*0.3}) {$\mathcal{S}_1$};
	\node at ({(\ORX+\L/2)*0.7+\ORX*0.3},{(\ORY-\L/2/sqrt(3))*0.7+\ORY*0.3}) {$\mathcal{S}_2$};
	\node at ({\ORX*0.7+\ORX*0.3},{(\ORY+\L*sqrt(3)/2-\L/2/sqrt(3))*0.7+\ORY*0.3}) {$\mathcal{S}_3$};
	\node at ({\ORX*0.2+\ORX*0.8},{(\ORY+\L*sqrt(3)/2-\L/2/sqrt(3))*0.2+\ORY*0.8}) {$\mathcal{C}$};
	\draw [->, line width=.5mm] (3,{\ORY+\L*sqrt(3)/2/2-\L/2/sqrt(3)}) -- (4,{\ORY+\L*sqrt(3)/2/2-\L/2/sqrt(3)});
	\def\OPX{7}; 
	\def\OPY{0};
	\def\L{2.5};
	\coordinate (V1) at (\OPX,\OPY-2.4);
	\coordinate (V2) at ({\OPX-\L*sqrt(3)/2},{\OPY+\L/2});
	\coordinate (V3) at ({\OPX+\L*sqrt(3)/2},{\OPY+\L/2});
    \shade[bottom color=yellow!75,top color=white] (V2) -- (\OPX,\OPY) -- (V3) -- (V2);
    \shade[bottom color=green!60,top color=white,transform canvas={rotate around={240:(\OPX,\OPY)]}}] (V2) -- (\OPX,\OPY) -- (V3) -- (V2);
    \shade[bottom color=blue!50,top color=white,transform canvas={rotate around={120:(\OPX,\OPY)]}}] (V2) -- (\OPX,\OPY) -- (V3) -- (V2);
	\draw (\OPX,\OPY) -- (V1);
	\draw (\OPX,\OPY) -- (V2);
	\draw (\OPX,\OPY) -- (V3);
	\node [right] at (V1) {$v \in \{(t,t,-2t)\mid t>0\}$};
	\node [above right] at (V2) {$v \in \{(-2t,t,t)\mid t>0\}$};
	\node [above right] at (V3) {$v \in \{(t,-2t,t)\mid t>0\}$};
	\node [below right] at (V3) {$=\rho^{-1}(\rho_A)$};
	\node at (\OPX,\OPY) [circle,fill=red,inner sep=1.0pt] {}; 
	\node [above] at (\OPX,\OPY) {$0$};
	\end{tikzpicture}
	}
	\caption{Here the topological features of the (multivalued) density--potential mapping in the triangle graph example are displayed. We look upon the density domain $\mathcal{P}_{3,2}$ from the $(1,1,1)$ direction in $\R^3$ space, where it appears as an equilateral triangle. The red circle $\mathcal{C}$ corresponds to the origin $v=(0,0,0)$ in the plane of gauged potentials with $\sum_i v_i=0$. Furthermore, the three exceptional densities, where the incircle touches the triangle, are mapped to the three displayed rays in the potential plane. Lastly, the three colored spikes of the triangle $\mathcal{S}_1,\mathcal{S}_2,\mathcal{S}_3$ are mapped to the three corresponding areas in the potential plane that are separated by the rays.}
	\label{fig:triangle-mapping}
\end{figure*}


\begin{figure}[ht]
    \centering
    \resizebox{\columnwidth}{!}{%
    \includegraphics[width=\columnwidth]{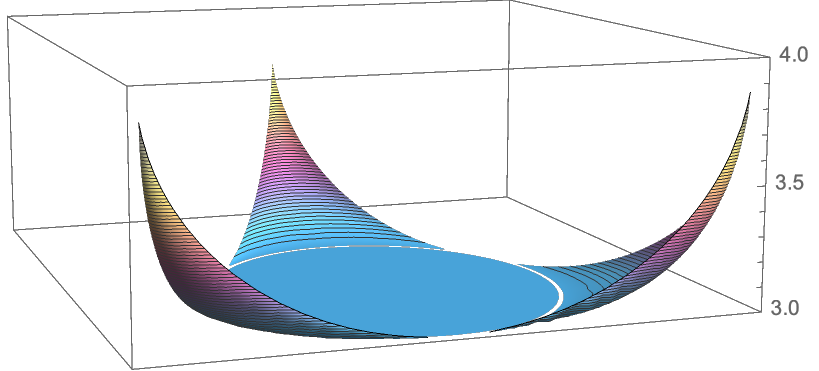}
    }
    \caption{Plot of the functional $F(\rho)=\tilde F(\rho)$ on the density domain $\mathcal{P}_{3,2}$ for the triangle graph example. The central, circular region that attains the minimum $F(\rho)=3$ comes from the two degenerate ground states as $\Psi = c_A\Psi_A + c_B\Psi_B \mapsto \rho$.}
    \label{fig:F-triangle}
\end{figure}

However, we can go much further in this simplest, yet extremely vivid example. It is even possible to give a closed, analytical expression for the functional $\tilde F(\rho)$ by choosing a parameterization for the trial wave function and minimizing \eqref{eq:tilde-F-def}. Using the already introduced division of $\mathcal{P}_{3,2}^+$ including the three exceptional points into the closed circular region $\mathcal{C}$ and the three open spike regions $\mathcal{S}_1,\mathcal{S}_2,\mathcal{S}_3$ (bottom left, bottom right, and top), we get after a quite tedious calculation that is given in Appendix~\ref{sec:app} that
\begin{widetext}
\begin{equation}\label{eq:full-triangle-tildeF}
    \tilde{F}(\rho) = \left\{ \begin{array}{ll}
        3 & \rho\in\mathcal{C} \\[0.5em]
        4 + 2  \left( \sqrt{(1-\rho_1) (1-\rho_3)} - \sqrt{(1-\rho_1) (1-\rho_2)} - \sqrt{(1-\rho_2) (1-\rho_3)}  \right) & \rho \in \mathcal{S}_1  \\[0.5em]
        4+2 \left(-\sqrt{(1-\rho_1) (1-\rho_3)} - \sqrt{(1-\rho_1) (1-\rho_2)} + \sqrt{(1-\rho_2) (1-\rho_3)} \right)\quad & \rho \in \mathcal{S}_2 \\[0.5em]
        4+2\left( -\sqrt{(1-\rho_1) (1-\rho_3)} + \sqrt{(1-\rho_1) (1-\rho_2)} - \sqrt{(1-\rho_2) (1-\rho_3)}  \right) & \rho \in \mathcal{S}_3 .
\end{array}
\right.
\end{equation}
\end{widetext}
The plot of the functional is displayed in Figure~\ref{fig:F-triangle}. It would be possible now to check that $\tilde F$ is convex directly and thus $F=\tilde F$, a result already obtained by arguing that any density in $\mathcal{P}_{3,2}^+$ is pure-state $v$-representable.
The only other example for an explicit form of a functional $\tilde F$ in the literature that we are aware of is by \citet{schonhammer1987discontinuity}, but only for a 2-site model.

Since we have obtained $F = \tilde{F}$ analytically, we can actually find the ground-state energy and density by direct minimization of the energy functional such as in \eqref{eq:E-LF} without any reference to the wave function. We have
\be
E(v) = \inf_\rho \left\{ \tilde F (\rho) + \rho_1 v_1 + \rho_2 v_2 + \rho_3 v_3 \right\}
\ee
where the minimization is over all densities $0 \leq \rho_i \leq 1$ and $\rho_1+\rho_2+\rho_3=2$.
Since we can always put $v_3=0$ by a constant shift of the potential we can write this as
\be
E(v) = \inf_{\rho_1,\rho_2} \left\{ \tilde{F} (\rho_1,\rho_2,2-\rho_1-\rho_2) + \rho_1 v_1 + \rho_2 v_2 \right\},
\ee
where $v_1,v_2$ are specified. The minimizing equations from the differentiation of the functional as already mentioned at the end of Section~\ref{sec:CS-functionals} are
\be
v_j = -\frac{\partial \tilde{F} (\rho_1,\rho_2,2-\rho_1-\rho_2)}{\partial \rho_j} \quad \quad (j=1,2)
\label{eq:Fderiv}
\ee
which determines the density if we specify the potential. 
If the potential $v$ is such that $\rho \in \mathcal{S}_2$ we need to employ
\begin{equation}
\begin{aligned}
\tilde F (\rho_1,\rho_2,2-\rho_1- \rho_2) = 4&+2 \left(  \sqrt{(1-\rho_2)(\rho_1+\rho_2-1)} \right. \\ &-  \sqrt{(1-\rho_1)(\rho_1+\rho_2-1)} \\&\left. - \sqrt{(1-\rho_1)(1-\rho_2)} \right)
\end{aligned}
\end{equation}
for which \eqref{eq:Fderiv} gives the following equations,
\begin{align}
v_1=&-\frac{\partial \tilde F (\rho_1,\rho_2,2-\rho_1-\rho_2)}{\partial \rho_1} \nonumber\\ =&
-\sqrt{ \frac{1-\rho_2}{\rho_1+\rho_2-1 } }
-
\sqrt{ \frac{\rho_1+\rho_2-1 }{1-\rho_1}}
\\&+
\sqrt{ \frac{1-\rho_1}{\rho_1+\rho_2-1 } }
-
\sqrt{ \frac{1-\rho_2}{1-\rho_1 } }, \nonumber
\end{align}
and
\begin{align}
v_2=&-\frac{\partial \tilde F  (\rho_1,\rho_2,2-\rho_1-\rho_2)}{\partial \rho_2} \nonumber\\=&
-\sqrt{ \frac{1-\rho_2}{\rho_1+\rho_2-1 } }
+
\sqrt{ \frac{\rho_1+\rho_2-1 }{1-\rho_2}}
\\&+
\sqrt{ \frac{1-\rho_1}{\rho_1+\rho_2-1 } }
-
\sqrt{ \frac{1-\rho_1}{1-\rho_2 } }. \nonumber
\end{align}
Take, for example, $(v_1,v_2,v_3)=(2,1,0)$, then we find that these equations
have the solution $(\rho_1,\rho_2,\rho_3)=(0.2121,\allowbreak 0.8176,\allowbreak 0.9704)$ which indeed lies inside the green density area $\mathcal{S}_2$ in Figure~\ref{fig:triangle-mapping} and gives
$\tilde F(\rho)=3.0832$ and $E(v)=\tilde F(\rho)+ 2\rho_1 +\rho_2=4.3249$, which can be checked by direct diagonalization of $H(v)$ to be the ground-state eigenvalue of the Hamiltonian \eqref{eq:H-triangle} for our choice of potential. We have therefore shown by explicit example that it is possible to find the ground-state energy without diagonalizing the Hamiltonian if
the analytic form of $F(\rho)$ is known and the functional is found to be differentiable (which it always is almost everywhere, as we remarked earlier).

\subsection{Complete graph: Counterexample to Lieb's non-convexity proof}
\label{sec:complete-graph}

A complete graph consists of $M$ vertices in which every vertex is connected to all other vertices. This is in a sense the opposite situation to the linear chain discussed in Section~\ref{sec:lin-chain} because it is maximally connected, whereas the linear chain is minimally connected.
For example, the complete graph with $M=3$ is the triangle graph that we already studied in detail in the previous section.
Taking the one-particle Hamiltonian $h$ from \eqref{eq:h-with-v}, for a complete graph we have
\be
h_{ii} = (M-1) + v_i, \quad \quad h_{ij} =-1  \quad ( i \neq j).
\ee
By Corollary~\ref{cor:PF-one-particle-H} we can find a non-degenerate one-particle ground state
\be
\phi_0 = (\phi_{0,1},\ldots, \phi_{0,M}) \in \H_1 \quad \text{with all} \quad \phi_{0,i} > 0.
\ee
If $\phi_k\in\H_1$ for $k=1,\ldots, N-1$ are the lowest, orthonormal, excited one-particle eigenstates of $h$ then the non-interacting $N$-particle ground state (or one out of the ground state multiplet, in case of degeneracy) is given by 
\be
\Psi_0 = \phi_0 \wedge \phi_1 \wedge \ldots \wedge \phi_{N-1} \in \H_N.
\ee
The density $\rho$ of the $N$-particle ground-state is the sum of one-particle densities evaluated at vertex $i$ such as in \eqref{eq:rho-from-Slater}, and we get
\be
\rho_i = \sum_{k=0}^{N-1} | \phi_{k,i} |^2 \geq \phi_{0,i}^2 > 0
\ee
since $\phi_{0,i}$ was found to be strictly positive.
This is not special to the complete graph but valid for all non-interacting $N$-particle states on a graph described by a graph Laplacian
(or more generally when we can apply the Perron--Frobenius theorem to the one-particle Hamiltonian, as demonstrated in Section~\ref{sec:pos-dens}).

For $v=0$ we define the uniform one-particle wave function
\be
\phi_0 = \frac{1}{\sqrt{M}} (1,1,\ldots,1)
\ee
for which we get $h\phi_0 = 0$ and this will turn out to be precisely the ground state. The excited states $\phi_k$ for $k=1,\ldots,M-1$ must be orthonormal to the eigenstate $\phi_0$, so they must be of the form
\be
\begin{aligned}
&\phi = (c_1,c_2,\ldots,c_M), 
\\&\langle \phi_0,\phi \rangle = \frac{1}{M}(c_1+c_2+\ldots + c_M)=0.
\end{aligned}
\ee
This in turn already means that any such vector is an eigenvector of $h$ with eigenvalue $M$, since
\be
\begin{aligned}
(h\phi)_i &= \sum_j h_{ij}c_j \\&= h_{ii} c_i + \sum_{j \neq i} h_{ij} c_j \\&= (M-1) c_i - \sum_{j \neq i} c_j \\&= M c_i - \sum_{j} c_j = M c_i.
\end{aligned}
\ee
Hence, the excited states are $(M-1)$-fold degenerate with eigenvalue $M$. One possible choice for such an excited state would be the maximally localized
\be\label{eq:complete-graph:localized-phi}
\phi = \frac{1}{\sqrt{2}} (1,-1,0,\ldots,0)
\ee
that immediately shows that a unique continuation property (UCP) for eigenstates of graph Hamiltonians is not achievable, since most coefficients are zero here.

We can give an alternative construction for the excited states using the following basic relations for the $M$-th roots of unity: Let $\omega= \exp(2 \pi \i/M)$ and define the $M$ vectors $\phi_k$ with components
\be\label{eq:complete-graph:phi-k-excited}
\phi_{k}= \frac{1}{\sqrt{M}} ( \omega^k, \omega^{2k}, \ldots, \omega^{Mk} ),\, k \in \{0,\ldots,M-1\}.
\ee
These vectors correspond to plane waves in the discrete setting and are therefore the basis for the discrete Fourier space.
Then, from a well-known relation for the roots of unity, we have
\be
\langle \phi_k, \phi_l \rangle = \frac{1}{M} \sum_{j=1}^M \, \omega^{(l-k)j} =  \delta_{k,l},
\ee
which establishes the fact that the orbitals $\phi_k$ define an orthonormal basis for $\H_1$ and are indeed the sought-after eigenstates of $h$.
Interestingly, the density of all these orbitals is the same,
\be
|\phi_k |^2 = \frac{1}{M} (1,1,\ldots,1),
\ee
and we have therefore found an orthonormal basis of equal density, which is also an eigenbasis. Finally, it should be noted that this equidensity eigenbasis forms a clear contradiction to the suggested proof of non-convexity of $\tilde F$ in \citet[Theorem~3.4]{Lieb1983}\footnote{Note that the original reference also contains a typo in the formulation of the theorem, where it says $F$, which is always convex as noted here in Section~\ref{sec:CS-functionals}, but actually the pure-state constrained-search functional $\tilde F$ is meant.}. In Lieb's proof it is stated that for a Hamiltonian with a certain rotational symmetry, as we have here with full permutational symmetry in the complete graph when $v=0$, the density of a uniformly mixed state including all orthonormal, degenerate ground-state orbitals will also have the full symmetry. In our case we just take two particles, $N=2$, and the ground states $\Phi_k = \phi_0\wedge\phi_k$, $k=1,\ldots,M-1$. Furthermore, it says that if $\Phi$ is any pure ground state, and hence a linear combination of the $\Phi_k$, then its density will \emph{not} have the same symmetry. However, in our example all $\Phi_k$ have exactly the same fully symmetric density. This means that the described procedure to find a non-pure-state $v$-representable density cannot work in general. However, we will rescue the argument right away in the next section with a more complex pure-state \texorpdfstring{$v$}{v}-representability counterexample.

\subsection{Cuboctahedron graph: Pure-state \texorpdfstring{$v$}{v}-representability counterexample}
\label{sec:cuboctahedron}


Here we give an example of a density $\rho \in \mathcal{P}_{M,N}^+$ which is \emph{not} pure-state $v$-representable. This issue was already addressed by \citet{Levy1982} and \citet{Lieb1983} who pointed out that any such example needs a system with at least three-fold degeneracy.
Here, we only have to show that the chosen ground-state ensemble density cannot come from a pure ground state of the \emph{same} Hamiltonian, which obviously saves us a lot of work.
This is because the first part of the HK theorem (Theorem~\ref{th:HK1}) states that if $H=H_0+V$ has a ground state $\Psi$ with a given density then $\Psi$ will also be the ground state for any other $H'=H_0+V'$ that allows the same ground-state density. Hence, it is always enough to check for just one potential.
The system under consideration is a graph with $M=12$ and the symmetry of a 3D structure that is known as the \emph{cuboctahedron}, see Figure~\ref{fig:gr12}. 
\citet{ullrich2002} used a very similar system in order to demonstrate that in general just very few potentials lead to degeneracy while there are many densities coming from degenerate states, a result that we already used in Section~\ref{sec:degen-non-uv} and that was also observed in the example of Section~\ref{sec:triangle-v-rep}.

\begin{figure}[ht]
\centering
\resizebox{\columnwidth}{!}{%
\begin{tikzpicture}
\def\centerarc[#1](#2)(#3:#4:#5)
    { \draw[#1] ($(#2)+({#5*cos(#3)},{#5*sin(#3)})$) arc (#3:#4:#5);}
  \centerarc[black](0,0)(0:360:3);
  \coordinate (g7) at (3*cos{45},3*sin{45}) ;
  \coordinate (g6) at (3*cos{135},3*sin{135}) ;
  \coordinate (g5) at (3*cos{225},3*sin{225}) ;
  \coordinate (g8) at (3*cos{315},3*sin{315}) ;
  
  \coordinate (g9) at (3*cos{135},0);
  \coordinate (g10) at (0,3*sin{45});
  \coordinate (g11) at (3*cos{45},0);
  \coordinate (g12) at (0,3*sin{315});
  
  \coordinate (g1) at (-1.5*cos{45},- 1.5*cos{45});
  \coordinate (g2) at (-1.5*cos{45},1.5*cos{45});
  \coordinate (g3) at (1.5*cos{45},1.5*cos{45});
  \coordinate (g4) at (1.5*cos{45},-1.5*cos{45});
  
  \draw[black] (g5) -- (g6) -- (g7) -- (g8) -- (g5);
  \draw[black] (g1) -- (g2) -- (g3) -- (g4) -- (g1);
  \draw[black] (g9) -- (g10) -- (g11) -- (g12) -- (g9);

  \node[graphnode] at (g1) { 1 };
   \node[graphnode] at (g2) { 2 };
   \node[graphnode] at (g3) { 3 };
   \node[graphnode] at (g4) { 4 };
   \node[graphnode] at (g5) { 5 };
   \node[graphnode] at (g6) { 6 };
   \node[graphnode] at (g7) { 7 };
   \node[graphnode] at (g8) { 8 };
    \node[graphnode] at (g9) { 9 };
   \node[graphnode] at (g10) {  \scriptsize 10 };
   \node[graphnode] at (g11) { \scriptsize 11 };
   \node[graphnode] at (g12) {  \scriptsize 12 };
 
	\node[inner sep=0pt] at (5.5,0)
		{\includegraphics[width=0.2\textwidth]{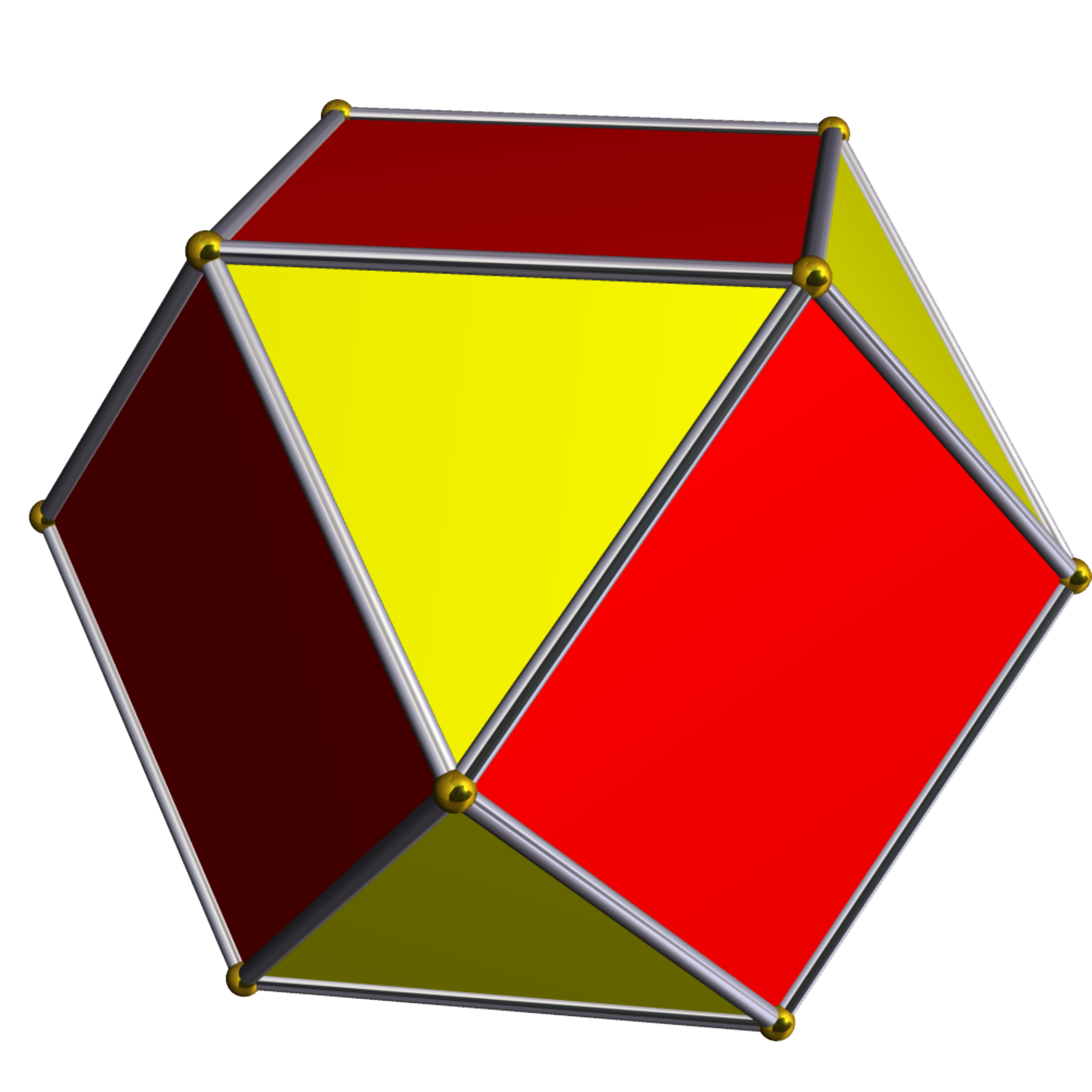}};
\end{tikzpicture}
}
\caption{Cuboctahedron graph. The graph is regular since every vertex has 4 adjacent neighbours. Its shape as an Archimedean solid is also displayed.}
\label{fig:gr12}
\end{figure}
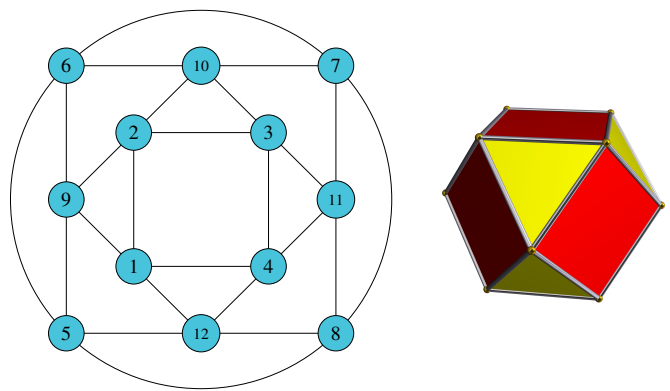

Taking the corresponding graph Laplacian \eqref{eq:graph-laplacian} without any further potentials as the one-particle Hamiltonian $h=-\Delta$, we arrive at the energy eigenvalues $E_0=0$ (non-degenerate), $E_1 = 2$ (3-fold degenerate), $E_2 = 4$ (3-fold degenerate), and $E_3 = 6$ (5-fold degenerate).
The ground state in the basis ordering displayed in Figure~\ref{fig:gr12} is given by
\be
\phi_0 = \frac{1}{\sqrt{12}} (1,1,1,1,1,1,1,1,1,1,1,1),
\ee
and the first three degenerate excited states are given by
\begin{align}
\phi_1 &= \frac{1}{\sqrt{8}} (-1,-1,-1,-1,1,1,1,1,0,0,0,0),\\
\phi_2 &= \frac{1}{4} ( 1,-1,-1,1,1,-1,-1,1,0,-2,0,2),
\end{align}
and
\be
\phi_3 = \frac{1}{4} (  -1,-1,1,1,-1,-1,1,1,-2,0,2,0).
\ee
The densities corresponding to these eigenvectors are given by
\begin{align}
\rho_0 &= \frac{1}{12} (1,1,1,1,1,1,1,1,1,1,1,1), \\
\rho_1 &= \frac{1}{8} (1,1,1,1,1,1,1,1,0,0,0,0 ), \\
\rho_2 &= \frac{1}{16} (1,1,1,1,1,1,1,1,0,4,0,4), \\
\rho_3 &= \frac{1}{16} (1,1,1,1,1,1,1,1,4,0,4,0).
\end{align}
If now $N=2$ non-interacting particles are considered on this graph with $2$-particle Hamiltonian $H$, their ground state is
\begin{equation}\label{eq:gr12-gs}
    \Psi = \sum_{n=1}^3 c_n \phi_0 \wedge \phi_n = \phi_0 \wedge \left(\sum_{n=1}^3 c_n\phi_n\right) \;\text{with}\; \sum_{n=1}^3 |c_n|^2=1
\end{equation}
or any mixture of such states. From the above we see that in this case any pure ground state can be written as a single Slater determinant, unlike in the example of \citet{englisch2}, where they showed using an explicit system with 3 particles with 6-fold degeneracy that even in a non-interacting system not every pure ground state can be written as a single Slater determinant. We take the equally distributed ensemble made from the three $\phi_0\wedge\phi_n$ as our counterexample, leading to to the uniform density
\be
\bar{\rho} = \rho_0 + \frac{1}{3} ( \rho_1 + \rho_2 + \rho_3 )
= \frac{1}{6} (1,1,1,1,1,1,1,1,1,1,1,1).
\ee
We will now demonstrate that this density cannot come from any pure state of the form \eqref{eq:gr12-gs}, which is, as we argued before, the most general form of a pure ground state of the Hamiltonian $H$. The density of \eqref{eq:gr12-gs} is
\be
\rho_i = \rho_{0,i} + \left|\sum_{n=1}^3 c_n\phi_{n,i}\right|^2,
\ee
so in order for this density to be equal to $\bar\rho$ we need to find coefficients $c_n$ that give
\be
\left|\sum_{n=1}^3 c_n\phi_{n,i}\right|^2 = \frac{1}{12} \quad\text{for all}\;i \in \{1,\ldots,12\}.
\ee
From those 12 equations some can be eliminated as duplicates and the following 6 remain,
\begin{align}
    &|\sqrt{2}c_1 + c_2 + c_3|^2 = |\sqrt{2}c_1 + c_2 - c_3|^2 \nonumber\\
    & \;\,\quad\quad\quad\quad\quad\quad\quad= |\sqrt{2}c_1 - c_2 + c_3|^2 \nonumber\\
    \label{eq:cubo-c-equations}
    & \;\,\quad\quad\quad\quad\quad\quad\quad= |\sqrt{2}c_1 - c_2 - c_3|^2 = \tfrac{4}{3},\\&
    |c_2|^2=|c_3|^2=\tfrac{1}{3}.\nonumber
\end{align}
By adding the first two equations we get
\begin{align}
    &|\sqrt{2}c_1 + c_2 + c_3|^2+|\sqrt{2}c_1 + c_2 - c_3|^2 \nonumber\\
    &= 2\left(|\sqrt{2}c_1 + c_2|^2 + |c_3|^2\right) = \tfrac{8}{3} \\&\Rightarrow \quad |\sqrt{2}c_1 + c_2|^2=1 \nonumber
\end{align}
and similarly $|\sqrt{2}c_1 - c_2|^2 = |\sqrt{2}c_1 + c_3|^2 = |\sqrt{2}c_1 - c_3|^2=1$. Using the same trick again we arrive at
\be
\begin{aligned}
    &|\sqrt{2}c_1 + c_2|^2+|\sqrt{2}c_1 - c_2|^2 = 4|c_1|^2+2|c_2|^2 = 2 \\
    & \Rightarrow \quad |c_1|^2 = |c_2|^2 = |c_3|^2 = \tfrac{1}{3}.
\end{aligned}
\ee
However, the equations above also mean that $\sqrt{2}c_1$ has equal distance to $\pm c_2$ and to $\pm c_3$, so in the complex plane it must lie on a line orthogonal to the lines connecting $c_2$ and $-c_2$ as well as $c_3$ and $-c_3$. This can only be fulfilled if $c_1 = \i s c_2 = \i t c_3$ with $s,t \in \mathbb{R}$, found to be $|s|=|t|=1$ from $|c_1|^2 = |c_2|^2 = |c_3|^2$. This means $c_2=\pm c_3$ which is in contradiction to \eqref{eq:cubo-c-equations} where then $c_2$ and $c_3$ would always cancel in two of the equations and thus would not lead to the necessary result $\tfrac{4}{3}$.
This shows that $\bar\rho$ cannot come from a pure ground state and thus is not pure-state $v$-representable. This in turn implies by Proposition~\ref{prop:pure-state-v-rep} that $\tilde F \neq F$, so since $F$ is known to be the convex hull of $\tilde F$ the pure-state constrained-search functional $\tilde F$ cannot be convex in general.

\section{Conclusions and open questions}
\label{sec:main:conclusions}

The main objective of this work was to point towards the problem of possible HK violations in a lattice setting and to seek possible remedies. One big relief is of course the result that almost all densities are uniquely $v$-representable (Theorem~\ref{th:non-uv-dens-measure-zero}). However, this and the many other albeit positive results point to a much larger theoretical complex that promises interesting paths for future investigations, some of which are collected in the following listing:

\begin{itemize}
    \item The considerations in Section~\ref{sec:degen-non-uv} point towards the intriguing possibility of a full \emph{geometrization} of the density-potential mapping. Therein, the potentials that lead to non-uv states are bounded by potentials that lead to degeneracy, while the corresponding degeneracy regions eventually touch the border of the density set or each other at the non-uv density points. The whole situation is therefore reduced to spheres (in different dimensionality) and lines (planes etc.), the basic elements of Euclidean geometry. A first demonstration has been given with the triangle example in Figure~\ref{fig:triangle-mapping}, a speculative depiction of how this might look in a more complex situation is given in Figure~\ref{fig:conjecture-mapping}. A corollary from this could be that almost all potentials lead to uv densities, a statement already put forward in Section~\ref{sec:degen-non-uv}.
    
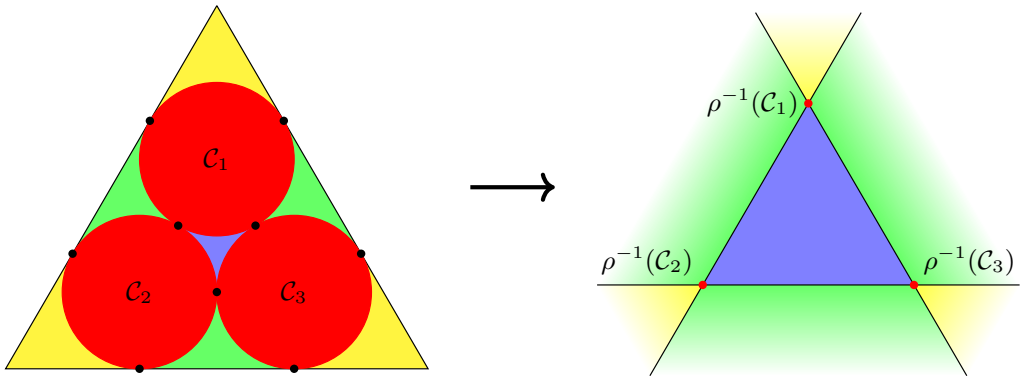
\begin{figure*}[ht] 
	\centering
	\resizebox{.75\textwidth}{!}{%
	\begin{tikzpicture}
	\def\ORX{0}; 
	\def\ORY{-1};
	\def\L{5}; 
	\def\R{0.915}; 
    \coordinate (R1) at (\ORX,{\ORY+\L*sqrt(3)/2-\L/2/sqrt(3)}); 
    \coordinate (R2) at ({\ORX-\L/2},{\ORY-\L/2/sqrt(3)}); 
    \coordinate (R3) at ({\ORX+\L/2},{\ORY-\L/2/sqrt(3)}); 
    \coordinate (P1) at ({\ORX-\R/2},{\ORY+sqrt(3)/6*\R});
    \coordinate (P2) at ({\ORX+\R/2},{\ORY+sqrt(3)/6*\R});
    \coordinate (P3) at (\ORX,{\ORY-1/sqrt(3)*\R});
    \coordinate (P4) at ({\ORX-(1+sqrt(3)/2)*\R},{\ORY-sqrt(3)/6*\L+3/2*\R});
    \coordinate (P5) at ({\ORX+(1+sqrt(3)/2)*\R},{\ORY-sqrt(3)/6*\L+3/2*\R});
    \coordinate (P6) at ({\ORX-sqrt(3)/2*\R},{\ORY+1/sqrt(3)*\L-3/2*\R});
    \coordinate (P7) at ({\ORX+sqrt(3)/2*\R},{\ORY+1/sqrt(3)*\L-3/2*\R});
    \coordinate (P8) at ({\ORX-\R},{\ORY-sqrt(3)/6*\L});
    \coordinate (P9) at ({\ORX+\R},{\ORY-sqrt(3)/6*\L});
	\coordinate (C1) at (\ORX,{\ORY+2/sqrt(3)*\R)});
	\coordinate (C2) at ({\ORX-\R},{\ORY-(sqrt(3)/6*\L-\R)});
	\coordinate (C3) at ({\ORX+\R},{\ORY-(sqrt(3)/6*\L-\R)});
    \path[fill=yellow,opacity=0.75] (R1) -- (P6) -- (P7) -- (R1);
    \path[fill=yellow,opacity=0.75] (R2) -- (P4) -- (P8) -- (R2);
    \path[fill=yellow,opacity=0.75] (R3) -- (P5) -- (P9) -- (R3);
    \path[fill=green,opacity=0.6] (P8) -- (P9) -- (P3) -- (P8);
    \path[fill=green,opacity=0.6] (P4) -- (P6) -- (P1) -- (P4);
    \path[fill=green,opacity=0.6] (P5) -- (P7) -- (P2) -- (P5);
    \path[fill=blue,opacity=0.5] (P1) -- (P2) -- (P3) -- (P1);
    \draw (R1) -- (R2) -- (R3) -- (R1);
    \draw [red,fill=red] (C1) circle (\R);
    \draw [red,fill=red] (C2) circle (\R);
    \draw [red,fill=red] (C3) circle (\R);
    \node at (C1) {$\mathcal{C}_1$};
    \node at (C2) {$\mathcal{C}_2$};
    \node at (C3) {$\mathcal{C}_3$};
    \node at (P1) [circle,fill,inner sep=1.0pt] {};
    \node at (P2) [circle,fill,inner sep=1.0pt] {};
    \node at (P3) [circle,fill,inner sep=1.0pt] {};
    \node at (P4) [circle,fill,inner sep=1.0pt] {};
    \node at (P5) [circle,fill,inner sep=1.0pt] {};
    \node at (P6) [circle,fill,inner sep=1.0pt] {};
    \node at (P7) [circle,fill,inner sep=1.0pt] {};
    \node at (P8) [circle,fill,inner sep=1.0pt] {};
	\node at (P9) [circle,fill,inner sep=1.0pt] {};
	\draw [->, line width=.5mm] (3,{\ORY+\L*sqrt(3)/2/2-\L/2/sqrt(3)}) -- (4,{\ORY+\L*sqrt(3)/2/2-\L/2/sqrt(3)});
	\def\OPX{7}; 
	\def\OPY{0};
	\def\L{2.5};
    \coordinate (V1) at (\OPX,{\OPY+\L/2/sqrt(3)}); 
    \coordinate (V2) at ({\OPX-\L/2},{\OPY-\L*sqrt(3)/2+\L/2/sqrt(3)}); 
    \coordinate (V3) at ({\OPX+\L/2},{\OPY-\L*sqrt(3)/2+\L/2/sqrt(3)}); 
    \coordinate (V11) at ({\OPX-\L/4},{\OPY+\L/2/sqrt(3)+sqrt(3)/4*\L});
    \coordinate (V12) at ({\OPX+\L/4},{\OPY+\L/2/sqrt(3)+sqrt(3)/4*\L});
	\coordinate (V21) at ({\OPX-\L},{\OPY-\L*sqrt(3)/2+\L/2/sqrt(3)});
    \coordinate (V22) at ({\OPX-3*\L/4},{\OPY-\L*sqrt(3)/2+\L/2/sqrt(3)-sqrt(3)/4*\L});
    \coordinate (V31) at ({\OPX+\L},{\OPY-\L*sqrt(3)/2+\L/2/sqrt(3)});
    \coordinate (V32) at ({\OPX+3*\L/4},{\OPY-\L*sqrt(3)/2+\L/2/sqrt(3)-sqrt(3)/4*\L});
	\path[fill=blue,opacity=0.5] (V1) -- (V2) -- (V3) -- (V1);
    \shade[bottom color=yellow!75,top color=white] (V1) -- (V11) -- (V12) -- (V1);
    \shade[bottom color=yellow!75,top color=white,transform canvas={rotate around={240:(\OPX,{\OPY-sqrt(3)/6*\L})}}] (V1) -- (V11) -- (V12) -- (V1);
    \shade[bottom color=yellow!75,top color=white,transform canvas={rotate around={120:(\OPX,{\OPY-sqrt(3)/6*\L})}}] (V1) -- (V11) -- (V12) -- (V1);
    \shade[top color=green!60,bottom color=white] (V22) -- (V2) -- (V3) -- (V32);
    \shade[top color=green!60,bottom color=white,transform canvas={rotate around={240:(\OPX,{\OPY-sqrt(3)/6*\L})]}}] (V22) -- (V2) -- (V3) -- (V32);
    \shade[top color=green!60,bottom color=white,transform canvas={rotate around={120:(\OPX,{\OPY-sqrt(3)/6*\L})]}}] (V22) -- (V2) -- (V3) -- (V32);
	\draw (V11) -- (V32);
	\draw (V12) -- (V22);
	\draw (V21) -- (V31);
	\node at (V1) [circle,fill=red,inner sep=1.0pt] {};
	\node [left] at (V1) {$\rho^{-1}(\mathcal{C}_1)$};
	\node at (V2) [circle,fill=red,inner sep=1.0pt] {};
	\node [above left] at (V2) {$\rho^{-1}(\mathcal{C}_2)$};
	\node at (V3) [circle,fill=red,inner sep=1.0pt] {};
	\node [above right] at (V3) {$\rho^{-1}(\mathcal{C}_3)$};
	\end{tikzpicture}
	}
	\caption{Symbolic depiction of a the possible topological structure of the density--potential mapping. The three red regions of densities from degenerate ground states $\mathcal{C}_1,\mathcal{C}_2,\mathcal{C}_3$ correspond to the three potential points in red. The single, marked density points are non-uv and get mapped to lines of potentials that extend to infinity when the density lies on the border of the density set. The graphs have to be understood as a 2D cut through the $(M-1)$-dimensional density and potential spaces.}
	\label{fig:conjecture-mapping}
\end{figure*}

    \item By presenting simple examples of graphs on which the HK theorem is violated (Section~\ref{sec:ex-non-unique}) and where it still holds (Section~\ref{sec:lin-chain}) the following task appears naturally: Give a maximal classification of graph topologies, including the allowed interactions, such that the HK theorem indeed holds for all ground states or even all eigenstates by fulfilling the Odlyzko condition from Corollary~\ref{cor:HK-necessary}. From the examples it seems that it is beneficial if, for many-particle systems, they are not \emph{too connected}. This could be checked by studying regular graphs (where every vertex has equal degree, i.e., the same number of connections; in this aspect they are closer to the continuum setting with a flat geometry), trees (which do not contain any loops), or graphs with a certain degree limit.
    
    \item We noted that by Odlyzko's theorem the number of zeroes in the wave function is directly related to unique $v$-representability. A closely connected issue is how these zeroes are distributed on a fermionic graph. A similar question has been the focus of considerable research in connection with the Courant--Hilbert nodal domain theorem for graph Laplacians \cite{graph-eigenvectors-book}. A natural question to ask in our context is whether a similar theorem holds for fermionic Hamiltonians. A significant open problem related to this situation is the so-called nodal-domain conjecture \cite{ceperley1991,bajdich2005,mitas2006} which states that interacting spin-$\tfrac{1}{2}$ fermions can be described by a spatial wave function with two connected domains, one on which the wave function is positive and one on which it is negative, separated by a nodal boundary. The resolution of this issue has important applications with regard to the so-called sign problem in Quantum Monte Carlo methodology.
    
    \item At the core of practical DFT lie, of course, the various approximations to the exchange-correlation potential and it would be interesting to see what can be learned about them and about possible exact conditions within the theory presented here.
    
    \item The omission of spin in this work comes as no restriction per se, since a spin degree of freedom can always be taken into account by adding more vertices. Yet, there will be usually no separate external potential acting on this internal coordinate and the graph will be disconnected into separate spin components. This leads to non-uniqueness of potentials that was already noted and discussed in both the continuous \cite{capelle2001nonuniqueness} and lattice setting \cite{ullrich2005nonuniqueness-spin}. Take for example $N=2$ spin-$\tfrac{1}{2}$ particles on the triangle graph then the one-particle Hamiltonian simply duplicates for the spin up and spin down sectors, represented by two disconnected graphs. The resulting fermionic graph displayed in Figure~\ref{fig:spin-graph} has three connected components of equal total spin polarization. In the general case of $N$ spin-$s$ particles we get ${N+2s \choose N}$ connected components in the fermionic graph. This immediately gives non-uv states due to disconnectedness, while the possibility of additional non-uniqueness of the type discussed in this paper can still appear in each separate component. If one now adds general magnetic fields they will appear as non-diagonal terms in the one-particle Hamiltonian and couple the different spin components. 

\begin{figure*}[ht]
\begin{subfigure}{.45\textwidth}
  \centering
  \begin{tikzpicture}[scale=2.0]
    \coordinate (CU1) at (-1,0);
    \coordinate (CU2) at (0,{sqrt(3)});
    \coordinate (CU3) at (1,0);
    \def\dY{-0.7};
    \coordinate (CD1) at (-1,{0+\dY});
    \coordinate (CD2) at (0,{sqrt(3)+\dY});
    \coordinate (CD3) at (1,{0+\dY});
    \foreach \m in {1,...,3}
    	\node[graphnode] (NU\m) at (CU\m) {$u_\m$};
    \foreach \m in {1,...,3}
    	\node[graphnode] (ND\m) at (CD\m) {$d_\m$};
    \draw (NU1) -- (NU2) -- (NU3) -- (NU1);
    \draw (ND1) -- (ND2) -- (ND3) -- (ND1);
  \end{tikzpicture}
  \caption{Triangle graph for a spin-$\tfrac{1}{2}$ particle}
\end{subfigure}
\hfill 
\begin{subfigure}{.45\textwidth}
  \centering
  \begin{tikzpicture}[scale=1.5]
    \node[draw,ellipse,minimum height=1cm,minimum width=4cm] (E1) at (0,0) {};
    \node[draw,ellipse,minimum height=1cm,minimum width=4cm] (E2) at (0,1) {};
    \node[draw,ellipse,minimum height=1cm,minimum width=4cm] (E3) at (0,2) {};
    \node[draw,ellipse,minimum height=1cm,minimum width=4cm] (EU) at (0,3.2) {};
    \node[draw,ellipse,minimum height=1cm,minimum width=4cm] (ED) at (0,-1.2) {};
   	\coordinate (E1C1) at (E1.270);
   	\coordinate (E1C2) at (E1.180);
   	\coordinate (E1C3) at (E1.0);
   	\coordinate (E2C1) at (E2.270);
   	\coordinate (E2C2) at (E2.180);
   	\coordinate (E2C3) at (E2.0);
   	\coordinate (E3C1) at (E3.270);
   	\coordinate (E3C2) at (E3.180);
   	\coordinate (E3C3) at (E3.0);
   	\coordinate (EUC1) at (EU.270);
   	\coordinate (EUC2) at (EU.180);
   	\coordinate (EUC3) at (EU.0);
   	\coordinate (EDC1) at (ED.270);
   	\coordinate (EDC2) at (ED.180);
   	\coordinate (EDC3) at (ED.0);
   	\draw plot [smooth cycle] coordinates {(E1C1) (E2C1) (E3C1) ($ (E3C1) + (-0.2,0.3) $) ($ (E3C1)!.5!(E2C1) + (-0.4,0) $) ($ (E2C1)!.5!(E1C1) + (-0.4,0) $) ($ (E1C1) + (-0.2,-0.3) $)};
   	\draw plot [smooth cycle] coordinates {(E1C2) (E2C2) (E3C2) ($ (E3C2) + (-0.2,0.3) $) ($ (E3C2)!.5!(E2C2) + (-0.4,0) $) ($ (E2C2)!.5!(E1C2) + (-0.4,0) $) ($ (E1C2) + (-0.2,-0.3) $)};
   	\draw plot [smooth cycle] coordinates {(E1C3) (E2C3) (E3C3) ($ (E3C3) + (0.2,0.3) $) ($ (E3C3)!.5!(E2C3) + (0.4,0) $) ($ (E2C3)!.5!(E1C3) + (0.4,0) $) ($ (E1C3) + (0.2,-0.3) $)};
    \foreach \m in {1,...,3}
    	\node[graphnode] (E1N\m) at (E1C\m) {};
    \foreach \m in {1,...,3}
    	\node[graphnode] (E2N\m) at (E2C\m) {};
    \foreach \m in {1,...,3}
    	\node[graphnode] (E3N\m) at (E3C\m) {};
    \foreach \m in {1,...,3}
    	\node[graphnode] (EUN\m) at (EUC\m) {};
    \foreach \m in {1,...,3}
    	\node[graphnode] (EDN\m) at (EDC\m) {};
    \node[below=5pt] at (EUC1) {$u_1\wedge u_2$};
    \node[right=5pt] at (EUC2) {$u_1\wedge u_3$};
    \node[left=5pt] at (EUC3) {$u_2\wedge u_3$};
    \node[below=5pt] at (EDC1) {$d_1\wedge d_2$};
    \node[right=5pt] at (EDC2) {$d_1\wedge d_3$};
    \node[left=5pt] at (EDC3) {$d_2\wedge d_3$};
    \node[below=4pt,right=3pt] at (E3C1) {$u_1\wedge d_1$};
    \node[right=5pt] at (E3C2) {$u_1\wedge d_2$};
    \node[left=5pt] at (E3C3) {$u_1\wedge d_3$};
    \node[below=4pt,right=3pt] at (E2C1) {$u_2\wedge d_1$};
    \node[right=5pt] at (E2C2) {$u_2\wedge d_2$};
    \node[left=5pt] at (E2C3) {$u_2\wedge d_3$};
    \node[below=4pt,right=3pt] at (E1C1) {$u_3\wedge d_1$};
    \node[right=5pt] at (E1C2) {$u_3\wedge d_2$};
    \node[left=5pt] at (E1C3) {$u_3\wedge d_3$};
    \node at (-2.5,3.2) {$S_z=+1$};
    \node at (-2.5,-1.2) {$S_z=-1$};
    \node at (-2.5,1) {$S_z=0$};
  \end{tikzpicture}
  \caption{Fermionic triangle graph for two spin-$\tfrac{1}{2}$ particles where the spin-compensated part turns out to be a discrete torus}
\end{subfigure}
\caption{Example for considering spin.}
\label{fig:spin-graph}
\end{figure*}
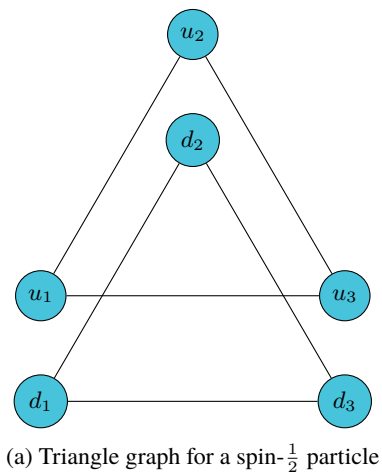
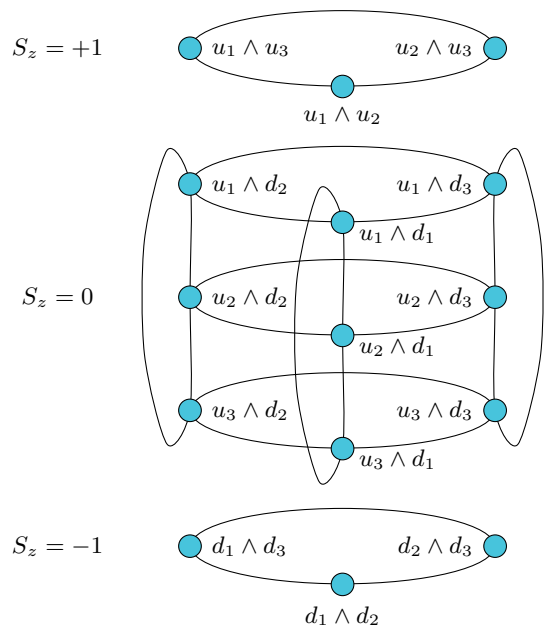
    
    \item In Section~\ref{sec:main:Rellich} we put forward the conjecture that the set of all uv densities in $\mathcal{P}_{M,N}$ is always open and the given examples seem to support this claim. Conversely, this would mean that the non-uv densities, which are known to have measure zero (Theorem~\ref{th:non-uv-dens-measure-zero}), also form a closed set.
    
    \item The field of graph theory has barely been touched in this work, so it is possible that some known results might prove useful for DFT of graphs. For example there is a theory about the spectra of signed graphs \cite{zaslavsky1982,belardo2019signed-graphs} that have a close connection to the fermionic Hamiltonians with $\pm 1$ entries in the off-diagonal appearing here.
    
    \item One further remaining question would be on how to achieve an appropriate continuum limit that also restores the full HK theorem. In accordance with \citet{lammert2010well} and contrary to \citet{ullrich2002} we do not believe that ``[t]he transition from a dense discrete lattice to a continuous variable does not appear to offer difficulties.''
    
    \item Going beyond the ground-state situation there is the question whether some of the insights gained here can be used to study the time-dependent Schr\"odinger equation for fermions on graphs.
    The discrete setting has been useful to study fundamental issues in time-dependent density functional theory \cite{li2008time,tokatly2011time,farzanehpour2012time}, and has important implications for quantum transport applications \cite{uimonen2011,stefanucci2015,jacob2020} or quantum dynamics of many-particle systems on lattices as in the Hubbard model \cite{karlsson2011} or in atomic traps \cite{kartsev2013}. In the mathematics literature there further exist many results on diffusion on graphs \cite{spectral-graph-book}. A significant advantage of the discrete setting is that path integrals become well-defined objects.
    
\end{itemize}

\section*{Data Availability Statement}

In order to explore some properties of fermionic systems on graphs numerically and display the graphs alongside their fermionic counterparts, we created a small Python script that is part of the public domain and can be found at {\small\url{https://mage.uber.space/dokuwiki/material/fermion-graph}}. All further data that support the findings of this study are available from the corresponding author upon reasonable request.

\begin{acknowledgments}
The authors express their gratitude towards Michael Ruggenthaler who helped them with the final draft of the manuscript during a joint stay at the Oberwolfach Research Institute for Mathematics.
M.~P.\ further acknowledges helpful discussions with Alexander Steinicke, the hospitality provided by Aisha during the final stage of this work, and financial support by the Erwin Schrödinger Fellowship J 4107-N27 of the FWF (Austrian Science Fund). R.~v.~L.\ acknowledges the Academy of Finland for support under project no.\ 317139.
\end{acknowledgments}

\appendix 

\section{Expression for the constrained-search functional in the triangle example}
\label{sec:app}

Here we give the full derivation of the closed, analytical expression for the constrained-search density functional $\tilde F$ for a non-interacting, two-particle system on the triangle graph with Hamiltonian \eqref{eq:H-triangle}
that is referenced in Section~\ref{sec:triangle-v-rep}. A general trial wave function $\Psi$ given in the $\{e_I\}_I$ basis of $\H_2$ is 
\be
    \Psi = c_1\, e_1\wedge e_2 + c_2\, e_1 \wedge e_3 + c_3\, e_2 \wedge e_3
\ee
and we want to minimize \eqref{eq:tilde-F-def} under the constraint $|c_1|^2+|c_2|^2+|c_3|^2=1$. Now since
\be\label{eq:app:constraints}
    \rho_1 = |c_1|^2+|c_2|^2, \rho_2 = |c_1|^2+|c_3|^2, \rho_3 = |c_2|^2+|c_3|^2,
\ee
and $\rho_1+\rho_2+\rho_3=2$, we can parameterize
\be
    c_1 = \sqrt{1-\rho_3}\,\e^{\i\varphi_1}, \;
    c_2 = \sqrt{1-\rho_2}\,\e^{\i\varphi_2}, \;
    c_3 = \sqrt{1-\rho_1},
\ee
where a global phase factor has been split off. Because $\rho$ is fixed, the search space is just $(\varphi_1,\varphi_2) \in [0,2\pi)^2$. In the next step we calculate the expectation value with respect to $H(0)$ from \eqref{eq:H-triangle} and get
\begin{align}\label{eq:app:H0-exp-val}
    \langle \Psi,\,&H(0)\Psi \rangle \\&= 4 + 2\left( \alpha_1\cos\varphi_1 -  \alpha_2\cos\varphi_2 - \alpha_3\cos(\varphi_1-\varphi_2) \right), \nonumber
\end{align}
where we defined
\be
\begin{aligned}
& \alpha_1 = \sqrt{(1-\rho_1) (1-\rho_3)}, \\ & \alpha_2 = \sqrt{(1-\rho_1) (1-\rho_2)}, \\ & \alpha_3 = \sqrt{(1-\rho_2) (1-\rho_3)}.
\end{aligned}
\ee
Since we limit ourselves to $\mathcal{P}_{3,2}^+$ all $\alpha_i>0$, while at the border of the triangular density domain an explicit solution is easily possible. This is because then always one $\rho_i=1$ and consequently two $\alpha_i$ are zero. At the edge where $\rho_A$ is located this leads to $\tilde F(\rho)=4-2\alpha_1$ and similarly to an expression with $\alpha_2,\alpha_3$ at the other edges.
In order to find the minimum of \eqref{eq:app:H0-exp-val} for the interior of $\mathcal{P}_{M,N}$ over all $\varphi_1,\varphi_2$ that gives the value of $\tilde F(\rho)$ we consider the expression
\be
    T(\varphi_1,\varphi_2) = \alpha_1\cos\varphi_1 -  \alpha_2\cos\varphi_2 - \alpha_3\cos(\varphi_1-\varphi_2)
\ee
for fixed values of $\alpha_i$ and solve
\begin{align}
    \label{eq:app:dTphi1}
    &0 = \frac{\partial T}{\partial \varphi_1} = -\alpha_1\sin\varphi_1+\alpha_3\sin(\varphi_1-\varphi_2),\\
    &0 = \frac{\partial T}{\partial \varphi_2} = \alpha_2\sin\varphi_2-\alpha_3\sin(\varphi_1-\varphi_2)
\end{align}
in order to find a minimum for $T(\varphi_1,\varphi_2)$.
Adding both equations gives
\be
-\alpha_1 \sin \varphi_1 + \alpha_2 \sin \varphi_2 =0 \; \Rightarrow \; \sin \varphi_2 =  \frac{\alpha_1}{\alpha_2} \sin \varphi_1.
\label{eq:app:result}
\ee
Rewriting \eqref{eq:app:dTphi1} as 
\be
0 = -\alpha_1 \sin \varphi_1 + \alpha_3 ( \sin \varphi_1 \cos \varphi_2 - \cos \varphi_1 \sin \varphi_2 )
\ee
and using the result \eqref{eq:app:result} from before then gives
\be
\sin \varphi_1 \left(- \alpha_1 + \alpha_3 \left( \cos \varphi_2  - \frac{\alpha_1}{\alpha_2} \cos \varphi_1 \right)\right) = 0.
\label{eq:app:case}
\ee
One option that immediately follows is that $\sin \varphi_1=0$ which then immediately yields $\sin \varphi_2=0$, which
gives $\varphi_1,\varphi_2 \in \{0,\pi\}$. 
We obtain
\begin{align}
(\varphi_1,\varphi_2)=(0,\pi) \; &\mapsto \; T(0,\pi)=\alpha_1 + \alpha_2 + \alpha_3, \\ 
(\varphi_1,\varphi_2)=(0,0) \; &\mapsto \; T(0,0)=\alpha_1 - \alpha_2- \alpha_3, \label{eq:app:00} \\ 
(\varphi_1,\varphi_2)=(\pi,0) \; &\mapsto \; T(\pi,0)=-\alpha_1 - \alpha_2 + \alpha_3,  \label{eq:app:p0}\\ 
(\varphi_1,\varphi_2)=(\pi,\pi) \; &\mapsto \; T(\pi,\pi)=-\alpha_1 + \alpha_2- \alpha_3. \label{eq:app:pp}
\end{align}
The case $(\varphi_1,\varphi_2)=(0,\pi)$ giving $\alpha_1+\alpha_2+\alpha_3>0$ is always larger than the other ones and can thus be disregarded.
In order to make $T$ minimal by choosing a density in $\mathcal{P}_{3,2}$ we have to, let us say for \eqref{eq:app:00}, minimize $\alpha_1$ and maximize the other two, which is clearly achieved by $\rho_2=1$ and $\rho_1=\rho_3=\tfrac{1}{2}$, exactly the exceptional point $\rho_A$ again. The lower bound for \eqref{eq:app:00}-\eqref{eq:app:pp} is thus $-\tfrac{1}{2}$ which means $\tilde F=4+2T \geq 3$, the correct ground-state eigenvalue of $H(0)$.
The remaining case of \eqref{eq:app:case} is when
\be
\begin{aligned}
&-\alpha_1 + \alpha_3 \left( \cos \varphi_2  - \frac{\alpha_1}{\alpha_2} \cos \varphi_1 \right)  = 0
\\\Rightarrow \quad &
\cos \varphi_2 =  \frac{\alpha_1}{\alpha_2} \left( \cos \varphi_1 +  \frac{\alpha_2}{\alpha_3} \right).
\label{eq:app:cos_expr}
\end{aligned}
\ee
We then find that
\be
\begin{aligned}
1&= \sin^2 \varphi_2 + \cos^2 \varphi_2 \\&=\left( \frac{\alpha_1}{\alpha_2} \right)^2 \sin^2 \varphi_1 + \left( \frac{\alpha_1}{\alpha_2} \right)^2 \left( \cos \varphi_1 + 
\frac{\alpha_2}{\alpha_3} \right)^2,
\end{aligned}
\ee
which gives
\be
\begin{aligned}
\left( \frac{\alpha_2}{\alpha_1} \right)^2 &= \sin^2 \varphi_1 + \left( \cos \varphi_1 + 
\frac{\alpha_2}{\alpha_3} \right)^2 \\&= 1 + 2 \frac{\alpha_2}{\alpha_3} \cos \varphi_1 + \left( \frac{\alpha_2}{\alpha_3} \right)^2,
\end{aligned}
\ee
from which we deduce
\begin{equation}\label{eq:app:2x-cos} 
\begin{aligned}
&\cos \varphi_1 
= \frac{\alpha_2 \alpha_3}{2 \alpha_1^2} - \frac{\alpha_2 }{2 \alpha_3} - \frac{\alpha_3 }{2 \alpha_2} 
\quad\text{and}\\
&\cos \varphi_2 = \frac{\alpha_3 }{2 \alpha_1} + \frac{\alpha_1 }{2 \alpha_3} - \frac{\alpha_1 \alpha_3 }{2 \alpha_2^2},
\end{aligned}
\end{equation}
where for the second equation we used \eqref{eq:app:cos_expr}.
We have therefore found an explicit expression for the angles of the stationary point
\be
\begin{aligned}
(\varphi_1,\varphi_2) = &\left( \arccos \left(\frac{\alpha_2 \alpha_3}{2 \alpha_1^2} - \frac{\alpha_2 }{2 \alpha_3} - \frac{\alpha_3 }{2 \alpha_2}\right),\right.\\
&\left.\;\;\, \arccos\left( \frac{\alpha_3 }{2 \alpha_1} + \frac{\alpha_1 }{2 \alpha_3} - \frac{\alpha_1 \alpha_3 }{2 \alpha_2^2}\right) \right).
\label{eq:app:opt_phi}
\end{aligned}
\ee
It now remains to evaluate $T$ in this optimal point. 
From \eqref{eq:app:2x-cos} we find
\be
\begin{aligned}
 \alpha_1 \cos \varphi_1 - \alpha_2 \cos \varphi_2  = - \frac{\alpha_1 \alpha_2}{\alpha_3}
 \label{eq:app:result2}
\end{aligned}
\ee
and we proceed to evaluate the term
\be
\begin{aligned}
 \cos (\varphi_1 -\varphi_2) &= \cos \varphi_1 \cos \varphi_2 + \sin \varphi_1 \sin \varphi_2 
\\&=  \cos \varphi_1   \cos \varphi_2  +  \frac{\alpha_1 }{\alpha_2} \sin^2 \varphi_1  
\\&=  \cos \varphi_1   \cos \varphi_2 + \frac{\alpha_1 }{\alpha_2} \left(1 - \cos^2 \varphi_1\right)  \\[-0.3em]
&=  \frac{\alpha_1 }{\alpha_2} - \frac{1}{\alpha_2} (  \alpha_1 \cos \varphi_1 - \alpha_2 \cos \varphi_2  ) \cos \varphi_1  \\&= \frac{\alpha_1 }{2 \alpha_2} + \frac{\alpha_2}{2 \alpha_1} - \frac{\alpha_1 \alpha_2}{2 \alpha_3^2},
\end{aligned}
\ee
where we used \eqref{eq:app:result}, \eqref{eq:app:2x-cos}, and \eqref{eq:app:result2}.
Therefore we obtain the expression
\be
\alpha_3  \cos (\varphi_1 -\varphi_2) = \frac{\alpha_1 \alpha_3 }{2 \alpha_2} + \frac{\alpha_2 \alpha_3}{2 \alpha_1} - \frac{\alpha_1 \alpha_2}{2 \alpha_3} 
\ee
and together with \eqref{eq:app:result2} this gives
\be
\begin{aligned}
T &= \alpha_1 \cos \varphi_1 - \alpha_2 \cos \varphi_2 - \alpha_3 \cos (\varphi_1 - \varphi_2) \\&=
 -\frac{\alpha_1 \alpha_3 }{2 \alpha_2} - \frac{\alpha_2 \alpha_3}{2 \alpha_1}  - \frac{\alpha_1 \alpha_2}{2 \alpha_3} 
\end{aligned}
\ee
at the stationary point. Now since 
\be
\frac{\alpha_1 \alpha_2}{ \alpha_3} = 1-\rho_1, \;
\frac{\alpha_2 \alpha_3}{ \alpha_1} = 1-\rho_2, \;
\frac{\alpha_1 \alpha_3}{ \alpha_2} = 1-\rho_3
\label{eq:app:alpha_eqs}
\ee
this gives a value at the minimum of 
\be
T = -\frac{1}{2} ( 3 - \rho_1- \rho_2- \rho_3) = - \frac{1}{2}
\quad\Rightarrow\quad
\tilde{F} = 4 + 2T =3.
\ee
This is also the lower bound for \eqref{eq:app:00}-\eqref{eq:app:pp} as already demonstrated, so we have definitely found a minimum here, but it remains to find out for which densities the derivation we performed was even valid. After all, we need to check that
\be
-1 \leq \cos \varphi_1 \leq 1 \quad \text{and thus} \quad -1 \leq \frac{\alpha_2 \alpha_3}{2 \alpha_1^2} - \frac{\alpha_2 }{2 \alpha_3} - \frac{\alpha_3 }{2 \alpha_2} \leq 1,
\ee
which equivalently gives
\be
- \alpha_1 \leq \frac{\alpha_2 \alpha_3}{2 \alpha_1} - \frac{\alpha_2 \alpha_1 }{2 \alpha_3} - \frac{\alpha_1\alpha_3 }{2 \alpha_2} \leq \alpha_1.
\ee
From \eqref{eq:app:alpha_eqs} we then have that
\begin{align}
\frac{\alpha_2 \alpha_3}{2 \alpha_1} &- \frac{\alpha_2 \alpha_1 }{2 \alpha_3} - \frac{\alpha_1\alpha_3 }{2 \alpha_2}
= \frac{1}{2} ( 1- \rho_2 - (1-\rho_1) - (1- \rho_3) ) \nonumber\\&= \frac{1}{2} (\rho_1+ \rho_3 -\rho_2-1) 
= \frac{1}{2} (1 -2 \rho_2)
\end{align}
and therefore $- 2 \alpha_1 \leq 1 - 2 \rho_2 \leq 2 \alpha_1$.
Squaring this equation gives the condition
\be
( 1 - 2 \rho_2 )^2 \leq 4 \alpha_1^2 = 4 (1-\rho_1) (1-\rho_3)=4 (1-\rho_1)  (\rho_1+ \rho_2-1),
\ee
which can be checked to be equivalent to
\be
\left( \sum_i \left(\rho_i-\frac{2}{3}\right)^{\!\!2} \right)^{\!\!\frac{1}{2}} \leq \frac{1}{\sqrt{6}},
\ee
precisely the condition for the incircle region $\mathcal{C}$ already derived in \eqref{eq:incircle-condition}. This means the result $\tilde F = 3$ is valid inside $\mathcal{C}$ while the three other results from \eqref{eq:app:00}-\eqref{eq:app:pp} are for the spike regions $\mathcal{S}_1,\mathcal{S}_2,\mathcal{S}_3$. This allows us to collect all results in the full expression for $\tilde F$ already given in \eqref{eq:full-triangle-tildeF},
\begin{equation}
    \tilde{F}(\rho) = \left\{ \begin{array}{ll}
        3 & \rho\in\mathcal{C} \\
        4 + 2  \left( \alpha_1 - \alpha_2 - \alpha_3 \right) & \rho \in \mathcal{S}_1  \\
        4+2 \left(-\alpha_1 - \alpha_2 + \alpha_3 \right)\quad & \rho \in \mathcal{S}_2 \\
        4+2\left( -\alpha_1 + \alpha_2 - \alpha_3  \right) & \rho \in \mathcal{S}_3 .
\end{array}
\right.
\end{equation}
Finally, let us remark that an equivalent approach to retrieve an expression for the constrained-search density functional would be to use the Lagrange multiplier method for finding the minimum of the Hermitian form $\langle \Psi, H_0\Psi \rangle$ under the constraints \eqref{eq:app:constraints}.


\begin{thebibliography}{67}%
\makeatletter
\providecommand \@ifxundefined [1]{%
 \@ifx{#1\undefined}
}%
\providecommand \@ifnum [1]{%
 \ifnum #1\expandafter \@firstoftwo
 \else \expandafter \@secondoftwo
 \fi
}%
\providecommand \@ifx [1]{%
 \ifx #1\expandafter \@firstoftwo
 \else \expandafter \@secondoftwo
 \fi
}%
\providecommand \natexlab [1]{#1}%
\providecommand \enquote  [1]{``#1''}%
\providecommand \bibnamefont  [1]{#1}%
\providecommand \bibfnamefont [1]{#1}%
\providecommand \citenamefont [1]{#1}%
\providecommand \href@noop [0]{\@secondoftwo}%
\providecommand \href [0]{\begingroup \@sanitize@url \@href}%
\providecommand \@href[1]{\@@startlink{#1}\@@href}%
\providecommand \@@href[1]{\endgroup#1\@@endlink}%
\providecommand \@sanitize@url [0]{\catcode `\\12\catcode `\$12\catcode
  `\&12\catcode `\#12\catcode `\^12\catcode `\_12\catcode `\%12\relax}%
\providecommand \@@startlink[1]{}%
\providecommand \@@endlink[0]{}%
\providecommand \url  [0]{\begingroup\@sanitize@url \@url }%
\providecommand \@url [1]{\endgroup\@href {#1}{\urlprefix }}%
\providecommand \urlprefix  [0]{URL }%
\providecommand \Eprint [0]{\href }%
\providecommand \doibase [0]{http://dx.doi.org/}%
\providecommand \selectlanguage [0]{\@gobble}%
\providecommand \bibinfo  [0]{\@secondoftwo}%
\providecommand \bibfield  [0]{\@secondoftwo}%
\providecommand \translation [1]{[#1]}%
\providecommand \BibitemOpen [0]{}%
\providecommand \bibitemStop [0]{}%
\providecommand \bibitemNoStop [0]{.\EOS\space}%
\providecommand \EOS [0]{\spacefactor3000\relax}%
\providecommand \BibitemShut  [1]{\csname bibitem#1\endcsname}%
\let\auto@bib@innerbib\@empty
\bibitem [{\citenamefont {Hohenberg}\ and\ \citenamefont
  {Kohn}(1964)}]{Hohenberg-Kohn1964}%
  \BibitemOpen
  \bibfield  {author} {\bibinfo {author} {\bibfnamefont {P.}~\bibnamefont
  {Hohenberg}}\ and\ \bibinfo {author} {\bibfnamefont {W.}~\bibnamefont
  {Kohn}},\ }\bibfield  {title} {\emph {\bibinfo {title} {Inhomogeneous
  electron gas},\ }}\href {\doibase 10.1103/PhysRev.136.B864} {\bibfield
  {journal} {\bibinfo  {journal} {Phys. Rev.}\ }\textbf {\bibinfo {volume}
  {136}},\ \bibinfo {pages} {B864} (\bibinfo {year} {1964})}\BibitemShut
  {NoStop}%
\bibitem [{\citenamefont {Levy}(1979)}]{Levy79}%
  \BibitemOpen
  \bibfield  {author} {\bibinfo {author} {\bibfnamefont {M.}~\bibnamefont
  {Levy}},\ }\bibfield  {title} {\emph {\bibinfo {title} {Universal variational
  functionals of electron densities, first-order density matrices, and natural
  spin-orbitals and solution of the v-representability problem},\ }}\href
  {\doibase 10.1073/pnas.76.12.6062} {\bibfield  {journal} {\bibinfo  {journal}
  {Proc. Natl. Acad. Sci. U. S. A.}\ }\textbf {\bibinfo {volume} {76}},\
  \bibinfo {pages} {6062} (\bibinfo {year} {1979})}\BibitemShut {NoStop}%
\bibitem [{\citenamefont {Lieb}(1983)}]{Lieb1983}%
  \BibitemOpen
  \bibfield  {author} {\bibinfo {author} {\bibfnamefont {E.~H.}\ \bibnamefont
  {Lieb}},\ }\bibfield  {title} {\emph {\bibinfo {title} {Density functionals
  for {C}oulomb-systems},\ }}\href {\doibase 10.1002/qua.560240302} {\bibfield
  {journal} {\bibinfo  {journal} {Int. J. Quantum Chem.}\ }\textbf {\bibinfo
  {volume} {24}},\ \bibinfo {pages} {243} (\bibinfo {year} {1983})}\BibitemShut
  {NoStop}%
\bibitem [{\citenamefont {Katriel}\ \emph {et~al.}(1981)\citenamefont
  {Katriel}, \citenamefont {Appellof},\ and\ \citenamefont
  {Davidson}}]{katriel1981mapping}%
  \BibitemOpen
  \bibfield  {author} {\bibinfo {author} {\bibfnamefont {J.}~\bibnamefont
  {Katriel}}, \bibinfo {author} {\bibfnamefont {C.~J.}\ \bibnamefont
  {Appellof}}, \ and\ \bibinfo {author} {\bibfnamefont {E.~R.}\ \bibnamefont
  {Davidson}},\ }\bibfield  {title} {\emph {\bibinfo {title} {Mapping between
  local potentials and ground state densities},\ }}\href {\doibase
  10.1002/qua.560190210} {\bibfield  {journal} {\bibinfo  {journal} {Int. J.
  Quantum Chem.}\ }\textbf {\bibinfo {volume} {19}},\ \bibinfo {pages} {293}
  (\bibinfo {year} {1981})}\BibitemShut {NoStop}%
\bibitem [{\citenamefont {Englisch}\ and\ \citenamefont
  {Englisch}(1984)}]{englisch2}%
  \BibitemOpen
  \bibfield  {author} {\bibinfo {author} {\bibfnamefont {H.}~\bibnamefont
  {Englisch}}\ and\ \bibinfo {author} {\bibfnamefont {R.}~\bibnamefont
  {Englisch}},\ }\bibfield  {title} {\emph {\bibinfo {title} {Exact density
  functionals for ground state energies, {I}{I}. {D}etails and remarks},\
  }}\href {\doibase 10.1002/pssb.2221240140} {\bibfield  {journal} {\bibinfo
  {journal} {Phys. Status Solidi B}\ }\textbf {\bibinfo {volume} {124}},\
  \bibinfo {pages} {373} (\bibinfo {year} {1984})}\BibitemShut {NoStop}%
\bibitem [{\citenamefont {Chayes}\ \emph {et~al.}(1985)\citenamefont {Chayes},
  \citenamefont {Chayes},\ and\ \citenamefont {Ruskai}}]{CCR1985}%
  \BibitemOpen
  \bibfield  {author} {\bibinfo {author} {\bibfnamefont {J.~T.}\ \bibnamefont
  {Chayes}}, \bibinfo {author} {\bibfnamefont {L.}~\bibnamefont {Chayes}}, \
  and\ \bibinfo {author} {\bibfnamefont {M.~B.}\ \bibnamefont {Ruskai}},\
  }\bibfield  {title} {\emph {\bibinfo {title} {Density functional approach to
  quantum lattice systems},\ }}\href {\doibase 10.1007/BF01010474} {\bibfield
  {journal} {\bibinfo  {journal} {J. Stat. Phys.}\ }\textbf {\bibinfo {volume}
  {38}},\ \bibinfo {pages} {497} (\bibinfo {year} {1985})}\BibitemShut
  {NoStop}%
\bibitem [{\citenamefont {Sch{\"o}nhammer}\ \emph {et~al.}(1995)\citenamefont
  {Sch{\"o}nhammer}, \citenamefont {Gunnarsson},\ and\ \citenamefont
  {Noack}}]{schonhammer1995density}%
  \BibitemOpen
  \bibfield  {author} {\bibinfo {author} {\bibfnamefont {K.}~\bibnamefont
  {Sch{\"o}nhammer}}, \bibinfo {author} {\bibfnamefont {O.}~\bibnamefont
  {Gunnarsson}}, \ and\ \bibinfo {author} {\bibfnamefont {R.~M.}\ \bibnamefont
  {Noack}},\ }\bibfield  {title} {\emph {\bibinfo {title} {Density-functional
  theory on a lattice: Comparison with exact numerical results for a model with
  strongly correlated electrons},\ }}\href {\doibase 10.1103/physrevb.52.2504}
  {\bibfield  {journal} {\bibinfo  {journal} {Phys. Rev. B}\ }\textbf {\bibinfo
  {volume} {52}},\ \bibinfo {pages} {2504} (\bibinfo {year}
  {1995})}\BibitemShut {NoStop}%
\bibitem [{\citenamefont {Xianlong}\ \emph {et~al.}(2006)\citenamefont
  {Xianlong}, \citenamefont {Polini}, \citenamefont {Tosi}, \citenamefont
  {Campo~Jr.}, \citenamefont {Capelle},\ and\ \citenamefont
  {Rigol}}]{xianlong2006}%
  \BibitemOpen
  \bibfield  {author} {\bibinfo {author} {\bibfnamefont {G.}~\bibnamefont
  {Xianlong}}, \bibinfo {author} {\bibfnamefont {M.}~\bibnamefont {Polini}},
  \bibinfo {author} {\bibfnamefont {M.~P.}\ \bibnamefont {Tosi}}, \bibinfo
  {author} {\bibfnamefont {V.~L.}\ \bibnamefont {Campo~Jr.}}, \bibinfo {author}
  {\bibfnamefont {K.}~\bibnamefont {Capelle}}, \ and\ \bibinfo {author}
  {\bibfnamefont {M.}~\bibnamefont {Rigol}},\ }\bibfield  {title} {\emph
  {\bibinfo {title} {Bethe ansatz density-functional theory of ultracold
  repulsive fermions in one-dimensional optical lattices},\ }}\href {\doibase
  10.1103/PhysRevB.73.165120} {\bibfield  {journal} {\bibinfo  {journal} {Phys.
  Rev. B}\ }\textbf {\bibinfo {volume} {73}},\ \bibinfo {pages} {165120}
  (\bibinfo {year} {2006})}\BibitemShut {NoStop}%
\bibitem [{\citenamefont {Giesbertz}\ and\ \citenamefont
  {Ruggenthaler}(2019)}]{giesbertz2019-1RDM}%
  \BibitemOpen
  \bibfield  {author} {\bibinfo {author} {\bibfnamefont {K.~J.~H.}\
  \bibnamefont {Giesbertz}}\ and\ \bibinfo {author} {\bibfnamefont
  {M.}~\bibnamefont {Ruggenthaler}},\ }\bibfield  {title} {\emph {\bibinfo
  {title} {One-body reduced density-matrix functional theory in finite basis
  sets at elevated temperatures},\ }}\href {\doibase
  10.1016/j.physrep.2019.01.010} {\bibfield  {journal} {\bibinfo  {journal}
  {Phys. Rep.}\ }\textbf {\bibinfo {volume} {806}},\ \bibinfo {pages} {1}
  (\bibinfo {year} {2019})}\BibitemShut {NoStop}%
\bibitem [{\citenamefont {Carrascal}\ \emph {et~al.}(2015)\citenamefont
  {Carrascal}, \citenamefont {Ferrer}, \citenamefont {Smith},\ and\
  \citenamefont {Burke}}]{carrascal2015hubbard}%
  \BibitemOpen
  \bibfield  {author} {\bibinfo {author} {\bibfnamefont {D.~J.}\ \bibnamefont
  {Carrascal}}, \bibinfo {author} {\bibfnamefont {J.}~\bibnamefont {Ferrer}},
  \bibinfo {author} {\bibfnamefont {J.~C.}\ \bibnamefont {Smith}}, \ and\
  \bibinfo {author} {\bibfnamefont {K.}~\bibnamefont {Burke}},\ }\bibfield
  {title} {\emph {\bibinfo {title} {The {H}ubbard dimer: {A} density functional
  case study of a many-body problem},\ }}\href {\doibase
  10.1088/0953-8984/27/39/393001} {\bibfield  {journal} {\bibinfo  {journal}
  {J. Phys.: Condens. Matter}\ }\textbf {\bibinfo {volume} {27}},\ \bibinfo
  {pages} {393001} (\bibinfo {year} {2015})}\BibitemShut {NoStop}%
\bibitem [{\citenamefont {Lima}\ \emph {et~al.}(2003)\citenamefont {Lima},
  \citenamefont {Silva}, \citenamefont {Oliveira},\ and\ \citenamefont
  {Capelle}}]{lima2003}%
  \BibitemOpen
  \bibfield  {author} {\bibinfo {author} {\bibfnamefont {N.~A.}\ \bibnamefont
  {Lima}}, \bibinfo {author} {\bibfnamefont {M.~F.}\ \bibnamefont {Silva}},
  \bibinfo {author} {\bibfnamefont {L.~N.}\ \bibnamefont {Oliveira}}, \ and\
  \bibinfo {author} {\bibfnamefont {K.}~\bibnamefont {Capelle}},\ }\bibfield
  {title} {\emph {\bibinfo {title} {Density functionals not based on the
  electron gas: Local-density approximation for a {L}uttinger liquid},\ }}\href
  {\doibase 10.1103/PhysRevLett.90.146402} {\bibfield  {journal} {\bibinfo
  {journal} {Phys. Rev. Lett.}\ }\textbf {\bibinfo {volume} {90}},\ \bibinfo
  {pages} {146402} (\bibinfo {year} {2003})}\BibitemShut {NoStop}%
\bibitem [{\citenamefont {Ij{\"a}s}\ and\ \citenamefont
  {Harju}(2010)}]{ijas2010lattice}%
  \BibitemOpen
  \bibfield  {author} {\bibinfo {author} {\bibfnamefont {M.}~\bibnamefont
  {Ij{\"a}s}}\ and\ \bibinfo {author} {\bibfnamefont {A.}~\bibnamefont
  {Harju}},\ }\bibfield  {title} {\emph {\bibinfo {title} {Lattice
  density-functional theory on graphene},\ }}\href {\doibase
  10.1103/PhysRevB.82.235111} {\bibfield  {journal} {\bibinfo  {journal} {Phys.
  Rev. B}\ }\textbf {\bibinfo {volume} {82}},\ \bibinfo {pages} {235111}
  (\bibinfo {year} {2010})}\BibitemShut {NoStop}%
\bibitem [{\citenamefont {Saubanere}\ and\ \citenamefont
  {Pastor}(2014)}]{saubanere2014lattice}%
  \BibitemOpen
  \bibfield  {author} {\bibinfo {author} {\bibfnamefont {M.}~\bibnamefont
  {Saubanere}}\ and\ \bibinfo {author} {\bibfnamefont {G.}~\bibnamefont
  {Pastor}},\ }\bibfield  {title} {\emph {\bibinfo {title} {Lattice
  density-functional theory of the attractive {H}ubbard model},\ }}\href
  {\doibase 10.1103/PhysRevB.90.125128} {\bibfield  {journal} {\bibinfo
  {journal} {Phys. Rev. B}\ }\textbf {\bibinfo {volume} {90}},\ \bibinfo
  {pages} {125128} (\bibinfo {year} {2014})}\BibitemShut {NoStop}%
\bibitem [{\citenamefont {Garrigue}(2018)}]{Garrigue2018}%
  \BibitemOpen
  \bibfield  {author} {\bibinfo {author} {\bibfnamefont {L.}~\bibnamefont
  {Garrigue}},\ }\bibfield  {title} {\emph {\bibinfo {title} {Unique
  continuation for many-body {S}chr{\"o}dinger operators and the
  {H}ohenberg--{K}ohn theorem},\ }}\href {\doibase 10.1007/s11040-018-9287-z}
  {\bibfield  {journal} {\bibinfo  {journal} {Math. Phys. Anal. Geom.}\
  }\textbf {\bibinfo {volume} {21}},\ \bibinfo {pages} {27} (\bibinfo {year}
  {2018})}\BibitemShut {NoStop}%
\bibitem [{\citenamefont {Penz}\ \emph {et~al.}(2019)\citenamefont {Penz},
  \citenamefont {Laestadius}, \citenamefont {Tellgren},\ and\ \citenamefont
  {Ruggenthaler}}]{penz2019guaranteed}%
  \BibitemOpen
  \bibfield  {author} {\bibinfo {author} {\bibfnamefont {M.}~\bibnamefont
  {Penz}}, \bibinfo {author} {\bibfnamefont {A.}~\bibnamefont {Laestadius}},
  \bibinfo {author} {\bibfnamefont {E.~I.}\ \bibnamefont {Tellgren}}, \ and\
  \bibinfo {author} {\bibfnamefont {M.}~\bibnamefont {Ruggenthaler}},\
  }\bibfield  {title} {\emph {\bibinfo {title} {Guaranteed convergence of a
  regularized {K}ohn--{S}ham iteration in finite dimensions},\ }}\href
  {\doibase 10.1103/physrevlett.123.037401} {\bibfield  {journal} {\bibinfo
  {journal} {Phys. Rev. Lett.}\ }\textbf {\bibinfo {volume} {123}},\ \bibinfo
  {pages} {037401} (\bibinfo {year} {2019})}\BibitemShut {NoStop}%
\bibitem [{\citenamefont {Penz}\ \emph {et~al.}(2020)\citenamefont {Penz},
  \citenamefont {Laestadius}, \citenamefont {Tellgren}, \citenamefont
  {Ruggenthaler},\ and\ \citenamefont {Lammert}}]{penz2020erratum}%
  \BibitemOpen
  \bibfield  {author} {\bibinfo {author} {\bibfnamefont {M.}~\bibnamefont
  {Penz}}, \bibinfo {author} {\bibfnamefont {A.}~\bibnamefont {Laestadius}},
  \bibinfo {author} {\bibfnamefont {E.~I.}\ \bibnamefont {Tellgren}}, \bibinfo
  {author} {\bibfnamefont {M.}~\bibnamefont {Ruggenthaler}}, \ and\ \bibinfo
  {author} {\bibfnamefont {P.~E.}\ \bibnamefont {Lammert}},\ }\bibfield
  {title} {\emph {\bibinfo {title} {Erratum: Guaranteed convergence of a
  regularized {K}ohn--{S}ham iteration in finite dimensions},\ }}\href
  {\doibase 10.1103/PhysRevLett.125.249902} {\bibfield  {journal} {\bibinfo
  {journal} {Phys. Rev. Lett.}\ }\textbf {\bibinfo {volume} {125}},\ \bibinfo
  {pages} {249902} (\bibinfo {year} {2020})}\BibitemShut {NoStop}%
\bibitem [{\citenamefont {Coe}\ \emph {et~al.}(2015)\citenamefont {Coe},
  \citenamefont {D'Amico},\ and\ \citenamefont
  {Fran{\c{c}}a}}]{coe2015uniqueness}%
  \BibitemOpen
  \bibfield  {author} {\bibinfo {author} {\bibfnamefont {J.~P.}\ \bibnamefont
  {Coe}}, \bibinfo {author} {\bibfnamefont {I.}~\bibnamefont {D'Amico}}, \ and\
  \bibinfo {author} {\bibfnamefont {V.~V.}\ \bibnamefont {Fran{\c{c}}a}},\
  }\bibfield  {title} {\emph {\bibinfo {title} {Uniqueness of
  density-to-potential mapping for fermionic lattice systems},\ }}\href
  {\doibase 10.1209/0295-5075/110/63001} {\bibfield  {journal} {\bibinfo
  {journal} {Europhys. Lett.}\ }\textbf {\bibinfo {volume} {110}},\ \bibinfo
  {pages} {63001} (\bibinfo {year} {2015})}\BibitemShut {NoStop}%
\bibitem [{\citenamefont {Dimitrov}\ \emph {et~al.}(2016)\citenamefont
  {Dimitrov}, \citenamefont {Appel}, \citenamefont {Fuks},\ and\ \citenamefont
  {Rubio}}]{dimitrov2016exact}%
  \BibitemOpen
  \bibfield  {author} {\bibinfo {author} {\bibfnamefont {T.}~\bibnamefont
  {Dimitrov}}, \bibinfo {author} {\bibfnamefont {H.}~\bibnamefont {Appel}},
  \bibinfo {author} {\bibfnamefont {J.~I.}\ \bibnamefont {Fuks}}, \ and\
  \bibinfo {author} {\bibfnamefont {A.}~\bibnamefont {Rubio}},\ }\bibfield
  {title} {\emph {\bibinfo {title} {Exact maps in density functional theory for
  lattice models},\ }}\href {\doibase 10.1088/1367-2630/18/8/083004} {\bibfield
   {journal} {\bibinfo  {journal} {New J. Phys.}\ }\textbf {\bibinfo {volume}
  {18}},\ \bibinfo {pages} {083004} (\bibinfo {year} {2016})}\BibitemShut
  {NoStop}%
\bibitem [{\citenamefont {Odlyzko}(1981)}]{Odlyzko}%
  \BibitemOpen
  \bibfield  {author} {\bibinfo {author} {\bibfnamefont {A.~M.}\ \bibnamefont
  {Odlyzko}},\ }\bibfield  {title} {\emph {\bibinfo {title} {On the ranks of
  some $(0,1)$-matrices with constant row sums},\ }}\href {\doibase
  10.1017/S1446788700033474} {\bibfield  {journal} {\bibinfo  {journal} {J.
  Aust. Math. Soc. Ser. A}\ }\textbf {\bibinfo {volume} {31}},\ \bibinfo
  {pages} {193} (\bibinfo {year} {1981})}\BibitemShut {NoStop}%
\bibitem [{\citenamefont {Kohn}(1983)}]{KohnPRL}%
  \BibitemOpen
  \bibfield  {author} {\bibinfo {author} {\bibfnamefont {W.}~\bibnamefont
  {Kohn}},\ }\bibfield  {title} {\emph {\bibinfo {title} {$v$-representability
  and density functional theory},\ }}\href {\doibase
  10.1103/PhysRevLett.51.1596} {\bibfield  {journal} {\bibinfo  {journal}
  {Phys. Rev. Lett.}\ }\textbf {\bibinfo {volume} {51}},\ \bibinfo {pages}
  {1596} (\bibinfo {year} {1983})}\BibitemShut {NoStop}%
\bibitem [{\citenamefont {Diestel}(2005)}]{diestel-graph-theory-book}%
  \BibitemOpen
  \bibfield  {author} {\bibinfo {author} {\bibfnamefont {R.}~\bibnamefont
  {Diestel}},\ }\href@noop {} {\emph {\bibinfo {title} {Graph theory}}},\
  \bibinfo {edition} {3rd}\ ed.\ (\bibinfo  {publisher} {Springer},\ \bibinfo
  {year} {2005})\BibitemShut {NoStop}%
\bibitem [{\citenamefont {Grünbaum}(2003)}]{convex-polytopes-book}%
  \BibitemOpen
  \bibfield  {author} {\bibinfo {author} {\bibfnamefont {B.}~\bibnamefont
  {Grünbaum}},\ }\href@noop {} {\emph {\bibinfo {title} {Convex polytopes}}},\
  \bibinfo {edition} {2nd}\ ed.\ (\bibinfo  {publisher} {Springer},\ \bibinfo
  {year} {2003})\BibitemShut {NoStop}%
\bibitem [{\citenamefont {Rispoli}(2008)}]{Rispoli2008-hypersimplex}%
  \BibitemOpen
  \bibfield  {author} {\bibinfo {author} {\bibfnamefont {F.~J.}\ \bibnamefont
  {Rispoli}},\ }\bibfield  {title} {\emph {\bibinfo {title} {The graph of the
  hypersimplex},\ }}\href@noop {} {\bibfield  {journal} {\bibinfo  {journal}
  {arXiv preprint}\ } (\bibinfo {year} {2008})},\ \Eprint
  {http://arxiv.org/abs/0811.2981} {arXiv:0811.2981 [math.CO]} \BibitemShut
  {NoStop}%
\bibitem [{\citenamefont {Harriman}(1981)}]{harriman1981orthonormal}%
  \BibitemOpen
  \bibfield  {author} {\bibinfo {author} {\bibfnamefont {J.~E.}\ \bibnamefont
  {Harriman}},\ }\bibfield  {title} {\emph {\bibinfo {title} {Orthonormal
  orbitals for the representation of an arbitrary density},\ }}\href {\doibase
  10.1103/PhysRevA.24.680} {\bibfield  {journal} {\bibinfo  {journal} {Phys.
  Rev. A}\ }\textbf {\bibinfo {volume} {24}},\ \bibinfo {pages} {680} (\bibinfo
  {year} {1981})}\BibitemShut {NoStop}%
\bibitem [{\citenamefont {Chung}(1994)}]{spectral-graph-book}%
  \BibitemOpen
  \bibfield  {author} {\bibinfo {author} {\bibfnamefont {F.~R.~K.}\
  \bibnamefont {Chung}},\ }\href@noop {} {\emph {\bibinfo {title} {Spectral
  graph theory}}}\ (\bibinfo  {publisher} {AMS},\ \bibinfo {year}
  {1994})\BibitemShut {NoStop}%
\bibitem [{\citenamefont {Biyikoglu}\ \emph {et~al.}(2007)\citenamefont
  {Biyikoglu}, \citenamefont {Leydold},\ and\ \citenamefont
  {Stadler}}]{graph-eigenvectors-book}%
  \BibitemOpen
  \bibfield  {author} {\bibinfo {author} {\bibfnamefont {T.}~\bibnamefont
  {Biyikoglu}}, \bibinfo {author} {\bibfnamefont {J.}~\bibnamefont {Leydold}},
  \ and\ \bibinfo {author} {\bibfnamefont {P.~F.}\ \bibnamefont {Stadler}},\
  }\href@noop {} {\emph {\bibinfo {title} {Laplacian eigenvectors of graphs}}}\
  (\bibinfo  {publisher} {Springer},\ \bibinfo {year} {2007})\BibitemShut
  {NoStop}%
\bibitem [{\citenamefont {ten Haaf}\ \emph {et~al.}(1995)\citenamefont {ten
  Haaf}, \citenamefont {van Bemmel}, \citenamefont {van Leeuwen}, \citenamefont
  {van Saarloos},\ and\ \citenamefont {Ceperley}}]{tenHaaf1995MonteCarlo}%
  \BibitemOpen
  \bibfield  {author} {\bibinfo {author} {\bibfnamefont {D.~F.~B.}\
  \bibnamefont {ten Haaf}}, \bibinfo {author} {\bibfnamefont {H.~J.~M.}\
  \bibnamefont {van Bemmel}}, \bibinfo {author} {\bibfnamefont {J.~M.~J.}\
  \bibnamefont {van Leeuwen}}, \bibinfo {author} {\bibfnamefont
  {W.}~\bibnamefont {van Saarloos}}, \ and\ \bibinfo {author} {\bibfnamefont
  {D.~M.}\ \bibnamefont {Ceperley}},\ }\bibfield  {title} {\emph {\bibinfo
  {title} {Proof for an upper bound in fixed-node {M}onte {C}arlo for lattice
  fermions},\ }}\href {\doibase 10.1103/PhysRevB.51.13039} {\bibfield
  {journal} {\bibinfo  {journal} {Phys. Rev. B}\ }\textbf {\bibinfo {volume}
  {51}},\ \bibinfo {pages} {13039} (\bibinfo {year} {1995})}\BibitemShut
  {NoStop}%
\bibitem [{\citenamefont {Pino}\ \emph {et~al.}(2007)\citenamefont {Pino},
  \citenamefont {Bokanowski}, \citenamefont {Ludeña},\ and\ \citenamefont
  {Boada}}]{pino2007}%
  \BibitemOpen
  \bibfield  {author} {\bibinfo {author} {\bibfnamefont {R.}~\bibnamefont
  {Pino}}, \bibinfo {author} {\bibfnamefont {O.}~\bibnamefont {Bokanowski}},
  \bibinfo {author} {\bibfnamefont {E.~V.}\ \bibnamefont {Ludeña}}, \ and\
  \bibinfo {author} {\bibfnamefont {R.~L.}\ \bibnamefont {Boada}},\ }\bibfield
  {title} {\emph {\bibinfo {title} {A re-statement of the {H}ohenberg–{K}ohn
  theorem and its extension to finite subspaces},\ }}\href {\doibase
  10.1007/s00214-007-0367-6} {\bibfield  {journal} {\bibinfo  {journal} {Theor.
  Chem. Acc.}\ }\textbf {\bibinfo {volume} {118}},\ \bibinfo {pages} {557}
  (\bibinfo {year} {2007})}\BibitemShut {NoStop}%
\bibitem [{\citenamefont {de~Figueiredo}\ and\ \citenamefont
  {Gossez}(1992)}]{FigueiredoGossez}%
  \BibitemOpen
  \bibfield  {author} {\bibinfo {author} {\bibfnamefont {D.~G.}\ \bibnamefont
  {de~Figueiredo}}\ and\ \bibinfo {author} {\bibfnamefont {J.-P.}\ \bibnamefont
  {Gossez}},\ }\bibfield  {title} {\emph {\bibinfo {title} {Strict monotonicity
  of eigenvalues and unique continuation},\ }}\href {\doibase
  10.1080/03605309208820844} {\bibfield  {journal} {\bibinfo  {journal}
  {Commun. Partial Differ. Equations}\ }\textbf {\bibinfo {volume} {17}},\
  \bibinfo {pages} {339} (\bibinfo {year} {1992})}\BibitemShut {NoStop}%
\bibitem [{\citenamefont {Regbaoui}(2001)}]{Regbaoui}%
  \BibitemOpen
  \bibfield  {author} {\bibinfo {author} {\bibfnamefont {R.}~\bibnamefont
  {Regbaoui}},\ }in\ \href {\doibase 10.1007/978-1-4612-0203-5_13} {\emph
  {\bibinfo {booktitle} {Carleman estimates and applications to uniqueness and
  control theory}}}\ (\bibinfo  {publisher} {Birkhäuser},\ \bibinfo {year}
  {2001})\ pp.\ \bibinfo {pages} {179--190}\BibitemShut {NoStop}%
\bibitem [{\citenamefont {Isozaki}\ and\ \citenamefont
  {Morioka}(2014)}]{Isozaki2014-Rellich-no-UCP}%
  \BibitemOpen
  \bibfield  {author} {\bibinfo {author} {\bibfnamefont {H.}~\bibnamefont
  {Isozaki}}\ and\ \bibinfo {author} {\bibfnamefont {H.}~\bibnamefont
  {Morioka}},\ }\bibfield  {title} {\emph {\bibinfo {title} {A {R}ellich type
  theorem for discrete {S}chrödinger operators},\ }}\href {\doibase
  10.3934/ipi.2014.8.475} {\bibfield  {journal} {\bibinfo  {journal} {Inverse
  Probl. Imaging}\ }\textbf {\bibinfo {volume} {8}},\ \bibinfo {pages} {475}
  (\bibinfo {year} {2014})}\BibitemShut {NoStop}%
\bibitem [{\citenamefont {Peyerimhoff}\ \emph {et~al.}(2017)\citenamefont
  {Peyerimhoff}, \citenamefont {Täufer},\ and\ \citenamefont
  {Veseli\'c}}]{Peyerimhoff2019-no-UCP}%
  \BibitemOpen
  \bibfield  {author} {\bibinfo {author} {\bibfnamefont {N.}~\bibnamefont
  {Peyerimhoff}}, \bibinfo {author} {\bibfnamefont {M.}~\bibnamefont
  {Täufer}}, \ and\ \bibinfo {author} {\bibfnamefont {I.}~\bibnamefont
  {Veseli\'c}},\ }\bibfield  {title} {\emph {\bibinfo {title} {Unique
  continuation principles and their absence for {S}chrödinger eigenfunctions
  on combinatorial and quantum graphs and in continuum space},\ }}\href
  {\doibase 10.17586/2220-8054-2017-8-2-216-230} {\bibfield  {journal}
  {\bibinfo  {journal} {Nanosyst.: Phys. Chem. Math.}\ }\textbf {\bibinfo
  {volume} {8}},\ \bibinfo {pages} {216} (\bibinfo {year} {2017})}\BibitemShut
  {NoStop}%
\bibitem [{\citenamefont {Longstaff}(1977)}]{Longstaff}%
  \BibitemOpen
  \bibfield  {author} {\bibinfo {author} {\bibfnamefont {W.~E.}\ \bibnamefont
  {Longstaff}},\ }\bibfield  {title} {\emph {\bibinfo {title} {Combinatorial
  solution of certain systems of linear equations involving $(0,1)$ matrices},\
  }}\href {\doibase 10.1017/S1446788700018899} {\bibfield  {journal} {\bibinfo
  {journal} {J. Aust. Math. Soc. Ser. A}\ }\textbf {\bibinfo {volume} {23}},\
  \bibinfo {pages} {266} (\bibinfo {year} {1977})}\BibitemShut {NoStop}%
\bibitem [{\citenamefont {Ullrich}\ and\ \citenamefont
  {Kohn}(2002)}]{ullrich2002}%
  \BibitemOpen
  \bibfield  {author} {\bibinfo {author} {\bibfnamefont {C.~A.}\ \bibnamefont
  {Ullrich}}\ and\ \bibinfo {author} {\bibfnamefont {W.}~\bibnamefont {Kohn}},\
  }\bibfield  {title} {\emph {\bibinfo {title} {Degeneracy in density
  functional theory: Topology in the v and n spaces},\ }}\href {\doibase
  10.1103/PhysRevLett.89.156401} {\bibfield  {journal} {\bibinfo  {journal}
  {Phys. Rev. Lett.}\ }\textbf {\bibinfo {volume} {89}},\ \bibinfo {pages}
  {156401} (\bibinfo {year} {2002})}\BibitemShut {NoStop}%
\bibitem [{\citenamefont {Prosalov}(2000)}]{prosalov-book}%
  \BibitemOpen
  \bibfield  {author} {\bibinfo {author} {\bibfnamefont {V.}~\bibnamefont
  {Prosalov}},\ }\href@noop {} {\emph {\bibinfo {title} {Problems and theorems
  in linear algebra}}}\ (\bibinfo  {publisher} {AMS},\ \bibinfo {year}
  {2000})\BibitemShut {NoStop}%
\bibitem [{\citenamefont {Rellich}(1937)}]{rellich1937}%
  \BibitemOpen
  \bibfield  {author} {\bibinfo {author} {\bibfnamefont {F.}~\bibnamefont
  {Rellich}},\ }\bibfield  {title} {\emph {\bibinfo {title}
  {St{\"o}rungstheorie der {S}pektralzerlegung, {I}. {M}itteilung},\ }}\href
  {\doibase 10.1007/BF01571652} {\bibfield  {journal} {\bibinfo  {journal}
  {Math. Ann.}\ }\textbf {\bibinfo {volume} {113}},\ \bibinfo {pages} {600}
  (\bibinfo {year} {1937})}\BibitemShut {NoStop}%
\bibitem [{\citenamefont {Rellich}(1969)}]{rellich-book}%
  \BibitemOpen
  \bibfield  {author} {\bibinfo {author} {\bibfnamefont {F.}~\bibnamefont
  {Rellich}},\ }\href@noop {} {\emph {\bibinfo {title} {Perturbation theory of
  eigenvalue problems}}}\ (\bibinfo  {publisher} {Gordon and Breach Science
  Publishers},\ \bibinfo {year} {1969})\BibitemShut {NoStop}%
\bibitem [{\citenamefont {Kato}(1995)}]{Kato-book}%
  \BibitemOpen
  \bibfield  {author} {\bibinfo {author} {\bibfnamefont {T.}~\bibnamefont
  {Kato}},\ }\href@noop {} {\emph {\bibinfo {title} {Perturbation theory for
  linear operators}}}\ (\bibinfo  {publisher} {Springer},\ \bibinfo {year}
  {1995})\BibitemShut {NoStop}%
\bibitem [{\citenamefont {Baumg\"artel}(1985)}]{Baumgaertel-book}%
  \BibitemOpen
  \bibfield  {author} {\bibinfo {author} {\bibfnamefont {H.}~\bibnamefont
  {Baumg\"artel}},\ }\href@noop {} {\emph {\bibinfo {title} {Analytic
  perturbation theory for matrices and operators}}}\ (\bibinfo  {publisher}
  {Birkh{\"a}user},\ \bibinfo {year} {1985})\BibitemShut {NoStop}%
\bibitem [{\citenamefont {Lammert}(2018)}]{Lammert2018}%
  \BibitemOpen
  \bibfield  {author} {\bibinfo {author} {\bibfnamefont {P.~E.}\ \bibnamefont
  {Lammert}},\ }\bibfield  {title} {\emph {\bibinfo {title} {In search of the
  {H}ohenberg--{K}ohn theorem},\ }}\href {\doibase 10.1063/1.5034215}
  {\bibfield  {journal} {\bibinfo  {journal} {J. Math. Phys.}\ }\textbf
  {\bibinfo {volume} {59}},\ \bibinfo {pages} {042110} (\bibinfo {year}
  {2018})}\BibitemShut {NoStop}%
\bibitem [{\citenamefont {Minc}(1988)}]{minc-book}%
  \BibitemOpen
  \bibfield  {author} {\bibinfo {author} {\bibfnamefont {H.}~\bibnamefont
  {Minc}},\ }\href@noop {} {\emph {\bibinfo {title} {Nonnegative matrices}}}\
  (\bibinfo  {publisher} {Wiley},\ \bibinfo {year} {1988})\BibitemShut
  {NoStop}%
\bibitem [{\citenamefont {Meyer}(2000)}]{meyer-book}%
  \BibitemOpen
  \bibfield  {author} {\bibinfo {author} {\bibfnamefont {C.~D.}\ \bibnamefont
  {Meyer}},\ }\href@noop {} {\emph {\bibinfo {title} {Matrix analysis and
  applied linear algebra}}}\ (\bibinfo  {publisher} {SIAM},\ \bibinfo {year}
  {2000})\BibitemShut {NoStop}%
\bibitem [{\citenamefont {Englisch}\ and\ \citenamefont
  {Englisch}(1983)}]{englisch1983vrep-PF}%
  \BibitemOpen
  \bibfield  {author} {\bibinfo {author} {\bibfnamefont {H.}~\bibnamefont
  {Englisch}}\ and\ \bibinfo {author} {\bibfnamefont {R.}~\bibnamefont
  {Englisch}},\ }\bibfield  {title} {\emph {\bibinfo {title}
  {V-representability in finite-dimensional spaces},\ }}\href {\doibase
  10.1088/0305-4470/16/18/003} {\bibfield  {journal} {\bibinfo  {journal} {J.
  Phys. A: Math. Gen.}\ }\textbf {\bibinfo {volume} {16}},\ \bibinfo {pages}
  {L693} (\bibinfo {year} {1983})}\BibitemShut {NoStop}%
\bibitem [{\citenamefont {{The idea for this proof is originally due to Paul
  E.\ Lammert, formulated at the YoungCAS workshop ``Do Electron Current
  Densities Determine All There Is to Know?''}}()}]{lammert-workshop}%
  \BibitemOpen
  \bibfield  {author} {\bibinfo {author} {\bibnamefont {{The idea for this
  proof is originally due to Paul E.\ Lammert, formulated at the YoungCAS
  workshop ``Do Electron Current Densities Determine All There Is to
  Know?''}}},\ }\href@noop {} {}\bibinfo {howpublished} {July 9-13, 2018, in
  Oslo, Norway}\BibitemShut {NoStop}%
\bibitem [{\citenamefont {Rockafellar}(1970)}]{rockafellar-book}%
  \BibitemOpen
  \bibfield  {author} {\bibinfo {author} {\bibfnamefont {R.~T.}\ \bibnamefont
  {Rockafellar}},\ }\href@noop {} {\emph {\bibinfo {title} {Convex analysis}}}\
  (\bibinfo  {publisher} {Princeton University Press},\ \bibinfo {year}
  {1970})\BibitemShut {NoStop}%
\bibitem [{\citenamefont {Kvaal}\ \emph {et~al.}(2014)\citenamefont {Kvaal},
  \citenamefont {Ekström}, \citenamefont {Teale},\ and\ \citenamefont
  {Helgaker}}]{Kvaal2014}%
  \BibitemOpen
  \bibfield  {author} {\bibinfo {author} {\bibfnamefont {S.}~\bibnamefont
  {Kvaal}}, \bibinfo {author} {\bibfnamefont {U.}~\bibnamefont {Ekström}},
  \bibinfo {author} {\bibfnamefont {A.~M.}\ \bibnamefont {Teale}}, \ and\
  \bibinfo {author} {\bibfnamefont {T.}~\bibnamefont {Helgaker}},\ }\bibfield
  {title} {\emph {\bibinfo {title} {Differentiable but exact formulation of
  density-functional theory},\ }}\href {\doibase 10.1063/1.4867005} {\bibfield
  {journal} {\bibinfo  {journal} {J. Chem. Phys.}\ }\textbf {\bibinfo {volume}
  {140}},\ \bibinfo {pages} {18A518} (\bibinfo {year} {2014})}\BibitemShut
  {NoStop}%
\bibitem [{\citenamefont {Lammert}(2007)}]{Lammert2007}%
  \BibitemOpen
  \bibfield  {author} {\bibinfo {author} {\bibfnamefont {P.~E.}\ \bibnamefont
  {Lammert}},\ }\bibfield  {title} {\emph {\bibinfo {title} {Differentiability
  of {L}ieb functional in electronic density functional theory},\ }}\href
  {\doibase 10.1002/qua.21342} {\bibfield  {journal} {\bibinfo  {journal} {Int.
  J. Quantum Chem.}\ }\textbf {\bibinfo {volume} {107}},\ \bibinfo {pages}
  {1943} (\bibinfo {year} {2007})}\BibitemShut {NoStop}%
\bibitem [{\citenamefont {Lammert}(2006)}]{lammert2006coarse}%
  \BibitemOpen
  \bibfield  {author} {\bibinfo {author} {\bibfnamefont {P.~E.}\ \bibnamefont
  {Lammert}},\ }\bibfield  {title} {\emph {\bibinfo {title} {Coarse-grained v
  representability},\ }}\href {\doibase 10.1063/1.2336211} {\bibfield
  {journal} {\bibinfo  {journal} {J. Chem. Phys.}\ }\textbf {\bibinfo {volume}
  {125}},\ \bibinfo {pages} {074114} (\bibinfo {year} {2006})}\BibitemShut
  {NoStop}%
\bibitem [{\citenamefont {Lammert}(2010)}]{lammert2010well}%
  \BibitemOpen
  \bibfield  {author} {\bibinfo {author} {\bibfnamefont {P.~E.}\ \bibnamefont
  {Lammert}},\ }\bibfield  {title} {\emph {\bibinfo {title} {Well-behaved
  coarse-grained model of density-functional theory},\ }}\href {\doibase
  10.1103/PhysRevA.82.012109} {\bibfield  {journal} {\bibinfo  {journal} {Phys.
  Rev. A}\ }\textbf {\bibinfo {volume} {82}},\ \bibinfo {pages} {012109}
  (\bibinfo {year} {2010})}\BibitemShut {NoStop}%
\bibitem [{\citenamefont {Rössler}\ \emph {et~al.}(2018)\citenamefont
  {Rössler}, \citenamefont {Verdozzi},\ and\ \citenamefont
  {Almbladh}}]{RoesslerVerdozzi2018}%
  \BibitemOpen
  \bibfield  {author} {\bibinfo {author} {\bibfnamefont {T.}~\bibnamefont
  {Rössler}}, \bibinfo {author} {\bibfnamefont {C.}~\bibnamefont {Verdozzi}},
  \ and\ \bibinfo {author} {\bibfnamefont {C.-O.}\ \bibnamefont {Almbladh}},\
  }\bibfield  {title} {\emph {\bibinfo {title} {A $v_0$-representability issue
  in lattice ensemble-{DFT} and its signature in lattice {TDDFT}},\ }}\href
  {\doibase 10.1140/epjb/e2018-90205-7} {\bibfield  {journal} {\bibinfo
  {journal} {Eur. Phys. J. B}\ }\textbf {\bibinfo {volume} {91}},\ \bibinfo
  {pages} {219} (\bibinfo {year} {2018})}\BibitemShut {NoStop}%
\bibitem [{\citenamefont {Levy}(1982)}]{Levy1982}%
  \BibitemOpen
  \bibfield  {author} {\bibinfo {author} {\bibfnamefont {M.}~\bibnamefont
  {Levy}},\ }\bibfield  {title} {\emph {\bibinfo {title} {Electron densities in
  search of {H}amiltonians},\ }}\href {\doibase 10.1103/PhysRevA.26.1200}
  {\bibfield  {journal} {\bibinfo  {journal} {Phys. Rev. A}\ }\textbf {\bibinfo
  {volume} {26}},\ \bibinfo {pages} {1200} (\bibinfo {year}
  {1982})}\BibitemShut {NoStop}%
\bibitem [{\citenamefont {Sch{\"o}nhammer}\ and\ \citenamefont
  {Gunnarsson}(1987)}]{schonhammer1987discontinuity}%
  \BibitemOpen
  \bibfield  {author} {\bibinfo {author} {\bibfnamefont {K.}~\bibnamefont
  {Sch{\"o}nhammer}}\ and\ \bibinfo {author} {\bibfnamefont {O.}~\bibnamefont
  {Gunnarsson}},\ }\bibfield  {title} {\emph {\bibinfo {title} {Discontinuity
  of the exchange-correlation potential in density functional theory},\ }}\href
  {\doibase 10.1088/0022-3719/20/24/010} {\bibfield  {journal} {\bibinfo
  {journal} {J. Phys. C: Solid State Phys.}\ }\textbf {\bibinfo {volume}
  {20}},\ \bibinfo {pages} {3675} (\bibinfo {year} {1987})}\BibitemShut
  {NoStop}%
\bibitem [{\citenamefont {Ceperley}(1991)}]{ceperley1991}%
  \BibitemOpen
  \bibfield  {author} {\bibinfo {author} {\bibfnamefont {D.~M.}\ \bibnamefont
  {Ceperley}},\ }\bibfield  {title} {\emph {\bibinfo {title} {Fermion nodes},\
  }}\href {\doibase 10.1007/BF01030009} {\bibfield  {journal} {\bibinfo
  {journal} {J. Stat. Phys.}\ }\textbf {\bibinfo {volume} {63}},\ \bibinfo
  {pages} {1237} (\bibinfo {year} {1991})}\BibitemShut {NoStop}%
\bibitem [{\citenamefont {Bajdich}\ \emph {et~al.}(2005)\citenamefont
  {Bajdich}, \citenamefont {Mitas}, \citenamefont {Drobn{\'y}},\ and\
  \citenamefont {Wagner}}]{bajdich2005}%
  \BibitemOpen
  \bibfield  {author} {\bibinfo {author} {\bibfnamefont {M.}~\bibnamefont
  {Bajdich}}, \bibinfo {author} {\bibfnamefont {L.}~\bibnamefont {Mitas}},
  \bibinfo {author} {\bibfnamefont {G.}~\bibnamefont {Drobn{\'y}}}, \ and\
  \bibinfo {author} {\bibfnamefont {L.~K.}\ \bibnamefont {Wagner}},\ }\bibfield
   {title} {\emph {\bibinfo {title} {Approximate and exact nodes of fermionic
  wavefunctions: Coordinate transformations and topologies},\ }}\href {\doibase
  10.1103/PhysRevB.72.075131} {\bibfield  {journal} {\bibinfo  {journal} {Phys.
  Rev. B}\ }\textbf {\bibinfo {volume} {72}},\ \bibinfo {pages} {075131}
  (\bibinfo {year} {2005})}\BibitemShut {NoStop}%
\bibitem [{\citenamefont {Mitas}(2006)}]{mitas2006}%
  \BibitemOpen
  \bibfield  {author} {\bibinfo {author} {\bibfnamefont {L.}~\bibnamefont
  {Mitas}},\ }\bibfield  {title} {\emph {\bibinfo {title} {Structure of fermion
  nodes and nodal cells},\ }}\href {\doibase 10.1103/PhysRevLett.96.240402}
  {\bibfield  {journal} {\bibinfo  {journal} {Phys. Rev. Lett.}\ }\textbf
  {\bibinfo {volume} {96}},\ \bibinfo {pages} {240402} (\bibinfo {year}
  {2006})}\BibitemShut {NoStop}%
\bibitem [{\citenamefont {Capelle}\ and\ \citenamefont
  {Vignale}(2001)}]{capelle2001nonuniqueness}%
  \BibitemOpen
  \bibfield  {author} {\bibinfo {author} {\bibfnamefont {K.}~\bibnamefont
  {Capelle}}\ and\ \bibinfo {author} {\bibfnamefont {G.}~\bibnamefont
  {Vignale}},\ }\bibfield  {title} {\emph {\bibinfo {title} {Nonuniqueness of
  the potentials of spin-density-functional theory},\ }}\href@noop {}
  {\bibfield  {journal} {\bibinfo  {journal} {Phys. Rev. Lett.}\ }\textbf
  {\bibinfo {volume} {86}},\ \bibinfo {pages} {5546} (\bibinfo {year}
  {2001})}\BibitemShut {NoStop}%
\bibitem [{\citenamefont {Ullrich}(2005)}]{ullrich2005nonuniqueness-spin}%
  \BibitemOpen
  \bibfield  {author} {\bibinfo {author} {\bibfnamefont {C.~A.}\ \bibnamefont
  {Ullrich}},\ }\bibfield  {title} {\emph {\bibinfo {title} {Nonuniqueness in
  spin-density-functional theory on lattices},\ }}\href {\doibase
  10.1103/PhysRevB.72.073102} {\bibfield  {journal} {\bibinfo  {journal} {Phys.
  Rev. B}\ }\textbf {\bibinfo {volume} {72}},\ \bibinfo {pages} {073102}
  (\bibinfo {year} {2005})}\BibitemShut {NoStop}%
\bibitem [{\citenamefont {Zaslavsky}(1982)}]{zaslavsky1982}%
  \BibitemOpen
  \bibfield  {author} {\bibinfo {author} {\bibfnamefont {T.}~\bibnamefont
  {Zaslavsky}},\ }\bibfield  {title} {\emph {\bibinfo {title} {Signed graphs},\
  }}\href {\doibase 10.1016/0166-218X(82)90033-6} {\bibfield  {journal}
  {\bibinfo  {journal} {Discrete Appl. Math.}\ }\textbf {\bibinfo {volume}
  {4}},\ \bibinfo {pages} {47} (\bibinfo {year} {1982})}\BibitemShut {NoStop}%
\bibitem [{\citenamefont {Belardo}\ \emph {et~al.}(2019)\citenamefont
  {Belardo}, \citenamefont {Cioab{\u{a}}}, \citenamefont {Koolen},\ and\
  \citenamefont {Wang}}]{belardo2019signed-graphs}%
  \BibitemOpen
  \bibfield  {author} {\bibinfo {author} {\bibfnamefont {F.}~\bibnamefont
  {Belardo}}, \bibinfo {author} {\bibfnamefont {S.~M.}\ \bibnamefont
  {Cioab{\u{a}}}}, \bibinfo {author} {\bibfnamefont {J.~H.}\ \bibnamefont
  {Koolen}}, \ and\ \bibinfo {author} {\bibfnamefont {J.}~\bibnamefont
  {Wang}},\ }\bibfield  {title} {\emph {\bibinfo {title} {Open problems in the
  spectral theory of signed graphs},\ }}\href@noop {} {\bibfield  {journal}
  {\bibinfo  {journal} {arXiv preprint}\ } (\bibinfo {year} {2019})},\ \Eprint
  {http://arxiv.org/abs/1907.04349} {arXiv:1907.04349 [math.CO]} \BibitemShut
  {NoStop}%
\bibitem [{\citenamefont {Li}\ and\ \citenamefont
  {Ullrich}(2008)}]{li2008time}%
  \BibitemOpen
  \bibfield  {author} {\bibinfo {author} {\bibfnamefont {Y.}~\bibnamefont
  {Li}}\ and\ \bibinfo {author} {\bibfnamefont {C.~A.}\ \bibnamefont
  {Ullrich}},\ }\bibfield  {title} {\emph {\bibinfo {title} {Time-dependent
  v-representability on lattice systems},\ }}\href {\doibase 10.1063/1.2955733}
  {\bibfield  {journal} {\bibinfo  {journal} {J. Chem. Phys.}\ }\textbf
  {\bibinfo {volume} {129}},\ \bibinfo {pages} {044105} (\bibinfo {year}
  {2008})}\BibitemShut {NoStop}%
\bibitem [{\citenamefont {Tokatly}(2011)}]{tokatly2011time}%
  \BibitemOpen
  \bibfield  {author} {\bibinfo {author} {\bibfnamefont {I.~V.}\ \bibnamefont
  {Tokatly}},\ }\bibfield  {title} {\emph {\bibinfo {title} {Time-dependent
  current density functional theory on a lattice},\ }}\href {\doibase
  10.1103/PhysRevB.83.035127} {\bibfield  {journal} {\bibinfo  {journal} {Phys.
  Rev. B}\ }\textbf {\bibinfo {volume} {83}},\ \bibinfo {pages} {035127}
  (\bibinfo {year} {2011})}\BibitemShut {NoStop}%
\bibitem [{\citenamefont {Farzanehpour}\ and\ \citenamefont
  {Tokatly}(2012)}]{farzanehpour2012time}%
  \BibitemOpen
  \bibfield  {author} {\bibinfo {author} {\bibfnamefont {M.}~\bibnamefont
  {Farzanehpour}}\ and\ \bibinfo {author} {\bibfnamefont {I.~V.}\ \bibnamefont
  {Tokatly}},\ }\bibfield  {title} {\emph {\bibinfo {title} {Time-dependent
  density functional theory on a lattice},\ }}\href {\doibase
  10.1103/PhysRevB.86.125130} {\bibfield  {journal} {\bibinfo  {journal} {Phys.
  Rev. B}\ }\textbf {\bibinfo {volume} {86}},\ \bibinfo {pages} {125130}
  (\bibinfo {year} {2012})}\BibitemShut {NoStop}%
\bibitem [{\citenamefont {Uimonen}\ \emph {et~al.}(2011)\citenamefont
  {Uimonen}, \citenamefont {Khosravi}, \citenamefont {Stan}, \citenamefont
  {Stefanucci}, \citenamefont {Kurth}, \citenamefont {van Leeuwen},\ and\
  \citenamefont {Gross}}]{uimonen2011}%
  \BibitemOpen
  \bibfield  {author} {\bibinfo {author} {\bibfnamefont {A.-M.}\ \bibnamefont
  {Uimonen}}, \bibinfo {author} {\bibfnamefont {E.}~\bibnamefont {Khosravi}},
  \bibinfo {author} {\bibfnamefont {A.}~\bibnamefont {Stan}}, \bibinfo {author}
  {\bibfnamefont {G.}~\bibnamefont {Stefanucci}}, \bibinfo {author}
  {\bibfnamefont {S.}~\bibnamefont {Kurth}}, \bibinfo {author} {\bibfnamefont
  {R.}~\bibnamefont {van Leeuwen}}, \ and\ \bibinfo {author} {\bibfnamefont
  {E.~K.~U.}\ \bibnamefont {Gross}},\ }\bibfield  {title} {\emph {\bibinfo
  {title} {Comparative study of many-body perturbation theory and
  time-dependent density functional theory in the out-of-equilibrium {A}nderson
  model},\ }}\href {\doibase 10.1103/PhysRevB.84.115103} {\bibfield  {journal}
  {\bibinfo  {journal} {Phys. Rev. B}\ }\textbf {\bibinfo {volume} {84}},\
  \bibinfo {pages} {115103} (\bibinfo {year} {2011})}\BibitemShut {NoStop}%
\bibitem [{\citenamefont {Stefanucci}\ and\ \citenamefont
  {Kurth}(2015)}]{stefanucci2015}%
  \BibitemOpen
  \bibfield  {author} {\bibinfo {author} {\bibfnamefont {G.}~\bibnamefont
  {Stefanucci}}\ and\ \bibinfo {author} {\bibfnamefont {S.}~\bibnamefont
  {Kurth}},\ }\bibfield  {title} {\emph {\bibinfo {title} {Steady-state density
  functional theory for finite bias conductances},\ }}\href {\doibase
  10.1021/acs.nanolett.5b03294} {\bibfield  {journal} {\bibinfo  {journal}
  {Nano Lett.}\ }\textbf {\bibinfo {volume} {15}},\ \bibinfo {pages} {8020}
  (\bibinfo {year} {2015})}\BibitemShut {NoStop}%
\bibitem [{\citenamefont {Jacob}\ \emph {et~al.}(2020)\citenamefont {Jacob},
  \citenamefont {Stefanucci},\ and\ \citenamefont {Kurth}}]{jacob2020}%
  \BibitemOpen
  \bibfield  {author} {\bibinfo {author} {\bibfnamefont {D.}~\bibnamefont
  {Jacob}}, \bibinfo {author} {\bibfnamefont {G.}~\bibnamefont {Stefanucci}}, \
  and\ \bibinfo {author} {\bibfnamefont {S.}~\bibnamefont {Kurth}},\ }\bibfield
   {title} {\emph {\bibinfo {title} {Mott metal-insulator transition from
  steady-state density functional theory},\ }}\href {\doibase
  10.1103/PhysRevLett.125.216401} {\bibfield  {journal} {\bibinfo  {journal}
  {Phys. Rev. Lett.}\ }\textbf {\bibinfo {volume} {125}},\ \bibinfo {pages}
  {216401} (\bibinfo {year} {2020})}\BibitemShut {NoStop}%
\bibitem [{\citenamefont {Karlsson}\ \emph {et~al.}(2011)\citenamefont
  {Karlsson}, \citenamefont {Privitera},\ and\ \citenamefont
  {Verdozzi}}]{karlsson2011}%
  \BibitemOpen
  \bibfield  {author} {\bibinfo {author} {\bibfnamefont {D.}~\bibnamefont
  {Karlsson}}, \bibinfo {author} {\bibfnamefont {A.}~\bibnamefont {Privitera}},
  \ and\ \bibinfo {author} {\bibfnamefont {C.}~\bibnamefont {Verdozzi}},\
  }\bibfield  {title} {\emph {\bibinfo {title} {Time-dependent
  density-functional theory meets dynamical mean-field theory: Real-time
  dynamics for the 3{D} {H}ubbard model},\ }}\href {\doibase
  10.1103/PhysRevLett.106.116401} {\bibfield  {journal} {\bibinfo  {journal}
  {Phys. Rev. Lett.}\ }\textbf {\bibinfo {volume} {106}},\ \bibinfo {pages}
  {116401} (\bibinfo {year} {2011})}\BibitemShut {NoStop}%
\bibitem [{\citenamefont {Kartsev}\ \emph {et~al.}(2013)\citenamefont
  {Kartsev}, \citenamefont {Karlsson}, \citenamefont {Privitera},\ and\
  \citenamefont {Verdozzi}}]{kartsev2013}%
  \BibitemOpen
  \bibfield  {author} {\bibinfo {author} {\bibfnamefont {A.}~\bibnamefont
  {Kartsev}}, \bibinfo {author} {\bibfnamefont {D.}~\bibnamefont {Karlsson}},
  \bibinfo {author} {\bibfnamefont {A.}~\bibnamefont {Privitera}}, \ and\
  \bibinfo {author} {\bibfnamefont {C.}~\bibnamefont {Verdozzi}},\ }\bibfield
  {title} {\emph {\bibinfo {title} {Three-dimensional dynamics of a fermionic
  {M}ott wedding-cake in clean and disordered optical lattices},\ }}\href
  {\doibase 10.1038/srep02570} {\bibfield  {journal} {\bibinfo  {journal} {Sci.
  Rep.}\ }\textbf {\bibinfo {volume} {3}},\ \bibinfo {pages} {2570} (\bibinfo
  {year} {2013})}\BibitemShut {NoStop}%
\end{thebibliography}
\end{document}